\PassOptionsToPackage{twoside=false}{geometry}
\documentclass[acmsmall,screen,nonacm]{acmart}

\newtoggle{supplementary}
\toggletrue{supplementary}

\newtoggle{diff}
\togglefalse{diff}

\AtBeginDocument{%
  }

\setcopyright{acmlicensed}
\copyrightyear{2018}
\acmYear{2018}
\acmDOI{XXXXXXX.XXXXXXX}

\theoremstyle{remark}
\newtheorem{remark}{Remark}

\usepackage{preamble} 
\makeatletter
\AtBeginDocument
 {
   \def\ltx@label#1{\cref@label{#1}}%add braces
   \def\label@in@display@noarg#1{\cref@old@label@in@display{#1}}%remove braces
\def\label@in@mmeasure@noarg#1{%
    \begingroup%
      \measuring@false%
      \cref@old@label@in@display{#1}%remove braces for multline, see https://tex.stackexchange.com/q/737204/2388
    \endgroup}%  
 } %
\makeatother
\begin{document}

\title{A Complementary Approach to Incorrectness Typing}

\author{Celia Mengyue Li}
\email{celia.li@bristol.ac.uk}
\affiliation{%
  \institution{University of Bristol}
  \country{UK}
}

\author{Sophie Pull}
\email{sophie.pull@bristol.ac.uk}
\affiliation{%
  \institution{University of Bristol}
  \country{UK}
}

\author{Steven Ramsay}
\email{steven.ramsay@bristol.ac.uk}
\affiliation{%
  \institution{University of Bristol}
  \country{UK}
}


\begin{abstract}
  We introduce a new two-sided type system for verifying the correctness and incorrectness of functional programs with atoms and pattern matching.  A key idea in the work is that types should range over sets of normal forms, rather than sets of values, and this allows us to define a complement operator on types that acts as a negation on typing formulas.  We show that the complement allows us to derive a wide range of refutation principles within the system, including the type-theoretic analogue of co-implication, and we use them to certify that a number of Erlang-like programs go wrong\del{, including one that is beyond the power of the state-of-the-art static analysis Dialyzer}.  An expressive axiomatisation of the complement operator via subtyping is shown decidable, and the type system as a whole is shown to be not only sound, but also complete for normal forms.
\end{abstract}

\begin{CCSXML}
  <ccs2012>
     <concept>
         <concept_id>10003752.10003790.10002990</concept_id>
         <concept_desc>Theory of computation~Logic and verification</concept_desc>
         <concept_significance>300</concept_significance>
         </concept>
     <concept>
         <concept_id>10003752.10003790.10011740</concept_id>
         <concept_desc>Theory of computation~Type theory</concept_desc>
         <concept_significance>500</concept_significance>
         </concept>
     <concept>
         <concept_id>10003752.10010124.10010138</concept_id>
         <concept_desc>Theory of computation~Program reasoning</concept_desc>
         <concept_significance>500</concept_significance>
         </concept>
   </ccs2012>
\end{CCSXML}
  
\ccsdesc[300]{Theory of computation~Logic and verification}
\ccsdesc[500]{Theory of computation~Type theory}
\ccsdesc[500]{Theory of computation~Program reasoning}

\keywords{type systems, refutation, higher-order program verification, incorrectness}


\maketitle
\section{Introduction}\label{sec:intro}

Since at least Milner's influential 1978 paper \cite{milner-jcss1978}, an important role of type systems in programming languages has been to provide a formal mechanism for reasoning statically about how programs will behave at runtime.  In this role, types $A$ are understood as a syntax for describing properties of (the behaviour of) program expressions $M$, and a derivation of the typing $M : A$ amounts to a proof that $M$ satisfies $A$.

It is often said that a typical purpose of such reasoning systems is the detection of runtime errors \cite{dunfield-krishnaswami-acmcs2021,pitts-types-notes}.  However, in a traditional type system, the means by which this benefit is realised is somewhat indirect, because the type system is designed only to certify the absence of such errors.  When a program is not well typed, it may be that it harbours a runtime error, but it could also be that the program simply lies beyond reasoning power of the type system.

To be able to detect runtime errors, we need a system that can reason about programs that go wrong.  However, such a capacity is incompatible with the prevailing approach to type systems for programming languages.  Typically, the soundness of the type system implies that \emph{well-typed programs cannot go wrong}, and since it is also typical that we can only construct type derivations for well-typed programs, it follows that we cannot derive properties of a program that goes wrong.  Such programs simply don't exist as far as the type system is concerned.

An alternative approach was proposed by Ramsay and Walpole in their POPL'24 paper \cite{ramsay-walpole-popl2024}.  Their \emph{two-sided type systems} allow for an extended form of typing judgement, whose shape is generally as follows:
\[
  M_1:A_1,\,\ldots,\,M_k:A_k \types N_1:B_1,\,\ldots,\,N_m:B_m
\]
Two-sided type systems lift the restriction that typing formulas $M:A$ on the left of the turnstile must be variable typings (i.e. $M$ must be a variable), and the restriction that there must be exactly one typing formula on the right-hand side.  However, the meaning of such a judgement smoothly generalises that of the traditional typing judgement: if \emph{all} of the $M_i$ reduce to values of type $A_i$, then \emph{some} $N_i$ is sure to reduce to a value of type $B_i$ (or diverge in the process).

In addition to the usual kinds of typing judgements, which appear as a special case, two-sided judgements can express that certain conditions are \emph{necessary} for proper evaluation.  For example, consider a typical definition of exponentiation $x^y$ as a fixed point involving multiplication, subtraction and an equality test (defined only) on integers\footnote{In this paper we write the integer literal corresponding to $n$ as $\pn{n}$}: $\abbv{power} \coloneqq \fixtm{f}{\abs{xy}{\iftm{y=\pn{0}}{\pn{1}}{x * (f\,x\,(y-\pn{1}))}}}$.  In the system of this paper, it is both true and provable that:
\[
  \abbv{power}\,x\,y : \intty \types y:\intty
\]
In other words, \emph{if} the term $\abbv{power}\,x\,y$ reduces to a value, \emph{then} the second input, $y$, \emph{must} have been an integer.  By contrast, however, the following judgement, is \emph{not} true:
\[
  \abbv{power}\,x\,y : \intty \types x : \intty
\]
It states that, for the power function to produce an integer value, the first input must have been an integer: \emph{if} $\abbv{power}\,x\,y$ evaluates to a integer, \emph{then} $x$ is an integer.
A quick perusal of the definition of \abbv{power}, however, suffices to convince oneself that it is not \emph{necessary}; the term $\abbv{power}\,(\abs{z}{z})\,\pn{0}$ will evaluate to a integer value without complaint.

Two-sided type systems sidestep the problem of $\types M:A$ being derivable only for $M$ that cannot go wrong by allowing to prove a kind of negation $M:A \types\ $.  The meaning of this judgement is: if $M$ evaluates to a value of type $A$ then false, or simply $M$ cannot evaluate to a value of type $A$.  With a type $\okty$ representing all values, one can prove, for example:
\begin{added}
  \[
    \abbv{power}\,\pn{0}\,(\abs{z}{z}) : \okty \types
  \]
\end{added}
Since $\okty$ is the type of \emph{all} values, a proof of this judgement is a certificate that exponentiation of \begin{added}zero and the identity function\end{added}, \emph{in that order}, either diverges or is sure to go wrong.

However, the original development of two-sided type systems in \cite{ramsay-walpole-popl2024} had some significant weaknesses, which we shall summarise here, and discuss in detail in the main body of the paper.  First, although the form of judgment is as above, most of the metatheory was developed only under the \emph{single-subject restriction} in which at most one of the $M_1,\ldots,M_k,N_1,\ldots,N_m$ is not a variable.  Second, although the two-sided setup naturally lends itself to a notion of complement on types, it was accommodated at the expense of \begin{added}a weaker\end{added} soundness guarantee.  Third, the systems of that work are formed from a large number of primitive rules, some overlapping, and with many not having a clear justification for their inclusion in the theory.

\subsection{A Two-Sided System with Complement on Normal Forms}

In this paper, we present a new two-sided type system which both overcomes these weaknesses and is an interesting type system in its own right.  As well as certifying that programs cannot go wrong (as usual), the system can also prove that programs cannot go right.  This latter capacity gives it a similar power to O'Hearn's Incorrectness Logic for first-order imperative programs \cite{ohearn-popl2019}, and Lindahl and Sagonas' Success Types for Erlang \cite{lindahl-sagonas-ppdp2006} (though neither of these have the former capacity).  Inspired by this connection, we use our system to certify the defectiveness of a number of textbook Erlang examples due to \cite{hebert2013}, including the following program, which \begin{added}is a textbook example used to demonstrate the limitations of Dialyzer's type inference, and which our system can prove to be incorrect\end{added}:

\begin{lstlisting}[xleftmargin=.2\textwidth]
  let money = $\abs{xy}{(\atom{give},(x,y))}$ in
  let item = $\lambda x.$ match x with                        
                              $\left|\begin{array}{lcl}(\atom{count},\,(\atom{give},\,(y,\,z))) &\mapsto& y;\\(\atom{account},\,(\atom{give},\,(y,\,z))) &\mapsto& z; \quad \mathsf{in}\end{array}\right.$
  let op = $\abs{xy}{x + y}$ in
  let tup = $\moneyTm\,\pn{5}\,\atom{you}$ in
        op (item ($\atom{count}$, tup)) (item ($\atom{account}$, tup))
\end{lstlisting}

% \begin{deleted}
%   The reasoning encoded in the typing derivation is roughly as follows.  If the expression in the innermost body of the let, $\someOp\,(\itemTm\,(\atom{count},\,\mathit{tup}))\,(\itemTm,\,(\atom{account},\,\mathit{tup}))$, would reduce to a value at runtime, then it \emph{cannot} be that all of the following are true:}
%   \begin{enumerate}[(i)]
%     \item \del{$\someOp$ is a function that requires that its second argument evaluates to an integer}
%     \item \del{and, $\itemTm$ is a function that, to return an integer, requires an input that does not reduce to a value of shape $(\atom{account},\,(\atom{give},\,(P,\,Q)))$, unless $Q$ is an integer}
%     \item \del{and, $\textit{tup}$ does reduce to a value of shape $(\atom{give},\,(P,\,Q))$, unless $Q$ is an integer.}
%   \end{enumerate}
%   \del{However, it is possible to derive a typing for each of $\someOp$, $\itemTm$ and $\textit{tup}$ that certifies that they \emph{do} satisfy (i), (ii) and (iii), and so we must conclude that the program as a whole did not evaluate to a value.}
% \end{deleted}

\begin{added}
  The reasoning encoded in the typing derivation is roughly as follows. The expression in the innermost body of the let, $\someOp\,(\itemTm\,(\atom{count},\,\mathit{tup}))\,(\itemTm,\,(\atom{account},\,\mathit{tup}))$ fails to reduce to a value whenever \emph{all} of the following are true:
  \begin{enumerate}[(i)]
    \item $\someOp$ is a function that requires that its second argument evaluates to an integer
    \item and, $\itemTm$ is a function that, to return an integer, requires an input that \emph{does not} reduce to a value of shape $(\atom{account},\,(\atom{give},\,(P,\,Q)))$ where $Q$ is not an integer
    \item and, $\textit{tup}$ \emph{does} reduce to a value of shape $(\atom{give},\,(P,\,Q))$ where $Q$ is not an integer.
  \end{enumerate}
  Since a typing can be derived for each of $\someOp$, $\itemTm$ and $\textit{tup}$ certifying that they \emph{do} satisfy (i), (ii) and (iii), we must conclude that the innermost body of the let, and therefore the program as a whole does not reduce to a value.
  % However, it is possible to derive a typing for each of $\someOp$, $\itemTm$ and $\textit{tup}$ that certifies that they \emph{do} satisfy (i), (ii) and (iii), and so we must conclude that the program as a whole did not evaluate to a value.
\end{added}

A key new idea in our work is that types should not be limited to describing sets of values, but instead should describe sets of normal forms more generally.  To motivate this idea, let us make a little more precise what is meant by \emph{going wrong}.  In general, each type system has its own notion of what it means to go wrong, because type systems have been designed to guarantee various different notions of safety, such as the classical type safety of Milner, memory safety, thread safety and so on.  However, it seems reasonable to consider safety properties in general, that is, that programs do not reach some kind of \emph{bad state}.

An elegant way to define bad states in a functional programming languages is as those program expressions that are neither values nor that can make a reduction step: we will call them \emph{stuck terms}.  In the language we study, we will consider both a classical notion of type safety, so that e.g. applying a number to a function, $\pn{3}\,(\abs{x}{x})$, is a stuck term, and we will consider pattern-matching safety, so that, e.g. failing to match \begin{added}every\end{added} branch in a case analysis expression, \begin{added}$\matchtm{(\atom{a},\,\pn{3})}{(\atom{b},\,x) \mapsto x}$\end{added}, is also stuck.

In the traditional approach to safety via type systems, one would establish a soundness theorem to ensure that any typable program $M:A$ cannot reduce to a stuck term.  In other words, if typable program $M$ normalises then its normal form must be a value of the type $A$.  This is therefore consistent with the usual semantic interpretation of $A$ as a certain set of values (modulo divergence).  Since, for incorrectness, we want to reason about programs that go wrong, it follows that we should allow types that number among their inhabitants also the stuck terms.  Since stuck terms and values, together, constitute all possible normal forms, we see the type system as a formal system for reasoning about how terms will reduce to normal form, in general, and types simply as sets of normal forms.

The interpretation of types as sets of normal forms allows us to axiomatise a \begin{added}complement\end{added} operator on types, the type $\CompPure{A}$ is interpreted as all normal forms that are not in the type $A$, \emph{without} \begin{added}weakening\end{added} our soundness guarantees.  Moreover, alongside a reasonably expressive subtyping relation, it allows us to form the type system from a small number of well-motivated rules, and yet induce an extremely rich theory -- we show that a wide variety of interesting typing principles are derivable, including \begin{added}analogues of\end{added} all the key refutational rules of \cite{ramsay-walpole-popl2024}.

In fact, complement with respect to normal forms arises naturally in any consideration of typing for incorrectness, and not only in the two-sided setting.  If we want to be able to assign types to terms that go wrong, then we must be able to describe, within the type system, the conditions under which terms of different shapes should be considered stuck.  For example, an integer addition $M + N$ is stuck if $M$ is stuck, or if $M$ is some value other than an integer.  In other words, $M+N$ is stuck whenever $M$ is some \emph{normal form} that is \emph{not} in the type $\intty$ of integer values, that is, when $M:\CompPure{\intty}$ is derivable.  Similarly, the second projection $\pi_2(M)$ of some term $M$ is stuck whenever $M$ is either itself stuck or when $M$ is some value other than a pair value.  In other words, when $M$ is some \emph{normal form} that is \emph{not} in the type of all pair values, or $M:\CompPure{\pairvalty}$ for short.  Function applications $M\,N$ are stuck whenever $M$ is a normal form \emph{not} in the type of all functions, and so on.  In fact, we are able to show that not only is our type system sound, but it admits a complete classification of normal forms: every normal form can be proven either to be a value or to be a stuck \begin{added}term\end{added}.  By contrast, most type systems in the literature do not satisfy even the value part of this property, that is, it is typical that there are values that are not well typed.

\paragraph{Contributions} In sum, our contributions are as follows:
\begin{enumerate}[(i)]
  \item We define a two-sided type system for a functional programming language with atoms and pattern matching.  The system exploits two-sidedness to give a natural and expressive rule for typing pattern-matching.  Unlike \cite{ramsay-walpole-popl2024}, the system is formed from a small number of well-motivated rules.
  \item We show that the refutational power of the system is, nevertheless, excellent, by deriving (analogues of) the key refutation typing rules from the original work (Theorem~\ref{thm:left-rules-derivable}).  Moreover, we show the potential of the idea by certifying that a sequence of textbook Erlang programs due to \citet{hebert2013} all go wrong\del{, including one that cannot be certified by the state-of-the-art Dialyzer system}.
  \item We introduce the coarrow type $A \coto B$ as a type-theoretic counterpart to coimplication (as featured in co- and bi-intuitionistic logic \cite{rauszer-stlog1974,crolard-tcs2001,tranchini-jal2017}), and show that it is the appropriate conceptual tool with which to reason refutationally about functions.
  \item We show soundness of the subtype relation via a natural set-theoretic model (Theorem~\ref{thm:subtyping-soundness}), and we show its decidability by describing an equivalent alternative characterisation (Theorem~\ref{thm:subtyping-equivalence}), whose decision problem can be reduced to bottom-up proof search (Theorem~\ref{thm:subtyping-decidability}).
  \item We show that the system is equivalent to a one-sided system with multiple conclusions (Theorem~\ref{thm:one-two-equivalence}), that is, in every judgement $\Gamma \types \Delta$, the assumptions $\Gamma$ may be taken to comprise only typings of shape $x:A$.  By introducing a new technical device for working with non-single-subject $\Delta$, we argue progress and preservation for the one-sided system, and thus obtain soundness of the original \begin{added}system\end{added} (Theorem~\ref{thm:progress-preservation}).
\end{enumerate}

\paragraph{Outline} The rest of the paper is structured as follows.  We introduce the programming language and its semantics in Section~\ref{sec:prog-lang}; the types, subtyping relation and type system in Section~\ref{sec:two-sided-system}.  We then derive the key principles of refutation and put them to work certifying H\'ebert's Erlang examples in Section~\ref{sec:refutation}.  This leads into a discussion of how to express necessity in conjunction with curried functions and the introduction of the coarrow as counterpart to coimplication in Section~\ref{sec:coto}.  We then develop the metatheory of subtyping in Section~\ref{sec:subtyping-props} and of the type system in Section~\ref{sec:type-system-metatheory}.  We conclude with a discussion of related work in Section~\ref{sec:related}.

\section{A Functional Programming Language}\label{sec:prog-lang}
We define a functional language with integers, pairing, recursive functions, Erlang-style atoms and matching.

\begin{definition}[Terms and Patterns]\label{def:terms}
    Assume a denumerable set of term variables $x$, $y$, $z$ and a denumerable set of atoms $\atom{a}$, $\atom{b}$, $\atom{c}$ including the distinguished atoms $\atom{true}$ and $\atom{false}$.
    The patterns of the language, typically $p$, $q$ are:
    \[
        p,q \Coloneqq x \mid \atom{a} \mid (p, q)\\
    \]
    Patterns are required to be \emph{linear} in the sense that each variable occurs in $p$ at most once.

    The \emph{terms} of the language, typically $M$, $N$, $P$, $Q$, are defined by the following grammar. Here, $\pn{n}$ denotes the numeral for integer $n$, the symbol $\binOpPrim{}{}$ represents any reasonable choice of binary operations on integers, e.g. $+$, $-$, $*$, and $\relOpPrim{}{}$ any reasonable choice of binary relations on integers.
    \[
        M,N \Coloneqq x \mid \atom{a} \mid \pn{n} \mid (M,N) \mid \begin{added}\abs{x}{\;\!\!M}\end{added} \mid M\;\!N \mid \fixtm{x}{\;\!\!M} \mid {M}\mathord{\otimes}{N} \mid M\mathord{\Join}N \mid \matchtm{M}{\mid_{i=1}^k \;\!\!p_i \!\mapsto\! N_i} \\
    \]
    The patterns $p_1,\,\ldots,\,p_k$ in a match expression are required to be \emph{disjoint}, that is, a given term will match at most one.
    We consider terms identified up to renaming of bound variables and adopt the usual conventions regarding their treatment.  The set of free variables of a term $M$ is written $\fv(M)$.  A term with no free variables is said to be \emph{closed}.

    A \emph{term substitution}, typically $\sigma$, $\tau$, is \begin{added}a\end{added} finite map from term variables to terms.  We write the application of a substitution $\sigma$ to a term $M$ postfix, as in $M\sigma$.  We use $\dom(\sigma)$ for the domain of the map $\sigma$.  
    % We say that a substitution is \emph{closed} just if the terms in its range are closed.  
    We write concrete substitutions explicitly as a list of maplets, as in $[M_1/x_1,\ldots,M_n/x_n]$.

    Finally, the \emph{values}, typically $V$, $W$, are \begin{added}those \emph{closed} terms that can be\end{added} described by the following:
    \[
        V,\,W \Coloneqq \atom{a} \mid \pn{n} \mid (V,\,W) \mid \abs{x}{M}
    \]
\end{definition}

\vspace{-.1cm}
Most of the language is standard, we mention only that: atoms $\atom{a}$ are literals without internal structure, they may only be compared via pattern matching, and the pattern match expression $\matchtm{M}{\mid_{i=1}^k p_i \mapsto N_i}$ scrutinises the term $M$, and if $M$ matches one of the (disjoint) patterns $p_i$ then it will evaluate as $N_i$.

\paragraph{Abbreviations}
% It will be helpful to define some abbreviations.
% For each $n\in\mathbb{N}$, we write $\pn{n}$ for the \emph{numeral} $\succtm^n\,(\zerotm)$. 
% The let construct is used to obtain the two components of a pair by matching but, in examples, we will typically always use it on the argument of an abstraction.  Hence, it is helpful to define $\abs{(x,y)}{M}$ as an abbreviation for $\abs{z}{\lettm{(x,y)}{z}{M}}$.
For $i\in[1,2]$ we write $\pi_i(M)$ as an abbreviation for $\matchtm{M}{(x_1,\,x_2) \mapsto x_i}$.  We write $\divtm$ for $\fixtm{x}{x}$.  We write $\iftm{M}{N}{P}$ as an abbreviation for $\matchtm{M}{\atom{true} \mapsto N \mid \atom{false} \mapsto P}$.    

The language reduces according to a call-by-name strategy\footnote{This matters little for the development, but we were interested to explore CBN as \cite{ramsay-walpole-popl2024} anticipates some difficulty when changing the meaning of typing formulas on the left of the turnstile to accommodate the fact that they may not evaluate.  However, we observed that the system is already sound for call-by-name without any change to the semantics, the crux of which is contained in Remark~\ref{rmk:cbn-funs}}.  However, in the interests of keeping the semantics of pattern matching simple, and a little closer to that of Erlang, we require that pairing and matching are strict.  We reduce the left-component of a pair first, and only if it reduces to a value do we reduce the second component.

% \begin{definition}
%     One-step reduction $M \ped N$ is defined inductively by the rules in \cref{def:reduction}. We write $M \peds N$ just if $M$ and $N$ are related by the reflexive transitive closure of this relation.  If there is no $N$ such that $M \ped N$, then we say that $M$ is a \emph{normal form}.  A term that can reduce to a normal form is said to be \emph{normalisable}, or to \emph{have a normal form}.  We write $M \Downarrow N$ to assert that $M$ reduces to normal form $N$.  Note that, since reduction is deterministic, a term that has a normal form will inevitably reach it.
% \end{definition}

\begin{definition}[Reduction]
    \label{def:reduction}
    The \emph{evaluation contexts} are defined by the following grammar:
    \[
        \arraycolsep=1.4pt
        \begin{array}{rcl}
            \Ctx & \Coloneqq & \Hole \mid \Ctx\,N \mid (\Ctx, N) \mid (V, \Ctx) \mid \opPrim{\Ctx}{N} \mid \opPrim{\pn{n}}{\Ctx} \mid \relOpPrim{\Ctx}{N} \mid \relOpPrim{\pn{n}}{\Ctx} \mid \matchtm{\Ctx}{\mid_{i=1}^k p_i \mapsto N_i}
        \end{array}
    \]
    Given a context $\Ctx$, we write $\Ctx[M]$ for the term obtained by replacing the hole $\Hole$ by $M$.  The \emph{one-step reduction relation}, $M \ped N$, is the binary relation on terms obtained as the closure of the following schema under evaluation contexts.  The terms on the left-hand side of the schema are said to be \emph{redexes}.  We write $\mathord{\rhd}^{\!*}$ for the reflexive transitive closure.
    \[
        \arraycolsep=1.4pt
        \begin{array}{rclcrcl}
            (\abs{x}{M})\,N                           & \ped & M[N/x] &\phantom{woooooo}&  \fixtm{x}{M}                            & \ped & M[\fixtm{x}{M}/x]                                \\[2mm]
            \pi_{i}(V_1,V_2)                          & \ped & V_{i} \quad\;\;\, (i\in \{1, 2\}) &&  \opPrim{\pn{n}}{\pn{m}}                   & \ped & \pn{n \otimes m}             \\[2mm]
            \matchtm{V}{\mid_{i=1}^k p_i \mapsto N_i}  & \ped &  \begin{added}N_j\end{added}\sigma \quad (V \eqsyntac \begin{added}p_j\end{added}\sigma) && \relOpPrim{\pn{n}}{\pn{m}} & \ped & \mbox{\textquotesingle}{(\relOpPrim{n}{m})}
        \end{array}
    \]
\end{definition}

\begin{definition}[Normal forms, values, stuck terms and going wrong]
    A term that cannot make a step is said to be in \emph{normal form}, and we reserve the letter $U$ as a metavariable for those terms that are in normal form.  \emph{Closed} normal forms can be partitioned into two disjoint sets: the values, which we have already introduced, and the stuck terms.
    A \begin{added}closed\end{added} term is said to be a \emph{stuck term}, for which we use the metavariables $S$ and $T$, just if it is a normal form but \emph{not} a value.  A term is said to \emph{go wrong} just if it reduces to a stuck term, and to \emph{evaluate} just if it reduces to a value.
\end{definition}

\begin{lemma}[Stuck Terms]
    The stuck terms may equivalently be described \begin{added}as those closed terms that can be described\end{added} by the following grammar:
    \[
        \begin{array}{rcll}
            S,\,T & \Coloneqq & (S,\,M) \mid (V,\,S)                                                                         &                                                                                                           \\
                  & \mid      & \binOpPrim{U}{M} \mid \binOpPrim{\pn{n}}{U} \mid \relOpPrim{U}{M} \mid \relOpPrim{\pn{n}}{U} & \text{when $U$ is not of shape $\pn{n}$}                                                                  \\
                  & \mid      & U\,M                                                                                         & \text{when $U$ is not of shape $\abs{x}{M}$}                                                              \\
                  & \mid      & \matchtm{U}{\mid_{i=1}^k p_i \mapsto N_i}                                                    & \text{when $U$ is a value and $\forall i .\,\forall \sigma.\,U \neq p_i\sigma$, or $U$ stuck}
        \end{array}
    \]
\end{lemma}

Observe the use of negation in the above definition: it is this that will later make the complement operator on types a particularly useful conceptual tool in reasoning about terms that go wrong.  For example, we will be able to encode the side condition on stuck terms of shape $U\,M$ by the typing formula $U : \CompPure{\funty}$.  Note also that, in stuck terms of shape $(S,\,M)$, the term $M$ in the second component is arbitrary -- it may not even be a normal form.  For example, $(\pn{1}\,(\abs{x}{x}),\,\divtm)$ is stuck, because pairs are evaluated eagerly from left-to-right, reduction stops once the first component becomes stuck.

% \begin{figure}\vspace{1ex}
%     \input{figures/small-step-semantics.tex}
%     \caption{Reduction rules}\label{fig:small-step-semantics}
% \end{figure}
\section{Two-Sided System}\label{sec:two-sided-system}

% In their paper at POPL'24, Ramsay and Walpole proposed an extension of the basic idea of type system \cite{ramsay-walpole-popl2024}.  Their analysis was that, although the basic shape of judgement, namely $x_1:A_1,\,\ldots,\,x_k:A_k \types N : A$, is well-motivated in the role of type systems as a definitional device for languages of syntactically well-formed terms, it is perhaps overly restrictive in their role as a proof system for establishing atomic formulas of shape $M:A$.  For example, why, on the left of the turnstile, should we limit ourselves to assuming only atomic formulas $x:A$ whose subject is some variable?  Ramsay and Walpole's \emph{two-sided type systems} allow for proving more liberal judgements, generally of shape:
% \[
%   M_1:A_1,\,\ldots,\,M_k:A_k \types N_1:B_1,\,\ldots,\,N_m:B_m
% \]
% As in a traditional type system, atomic formulas are typings $M:A$, but they are made first-class citizens, in the sense that they may be assumed on the left without restriction.  Moreover, as is common in many proof systems, one may conclude any number of atomic formulas on the right of the turnstile.  The meaning of such a judgement smoothly generalises that of the traditional typing judgement, namely that if \emph{all} of the $M_i$ evaluate to values of type $A_i$, then \emph{some} $N_i$ is sure to evaluate to a value of type $B_i$ (or diverge in the process).

We start by defining the language of types and associated pattern type  substitutions.

\begin{definition}[Types]\label{def:simple-types}\label{def:type-meanings}
  The types of the system, typically $A$, $B$, $C$, are given by the following:
  \[
    \begin{array}{rcl}
      A, B & \Coloneqq & \topty \mid \okty \mid \atomty{} \mid \intty{} \mid \atom{a} \mid (A,\,B) \mid A \to B \mid \bigvee_{i=1}^k A_i \mid \CompPure{A}
    \end{array}
  \]
  We define a \emph{pattern type substitution} as a finite mapping, typically $\theta$, from \emph{term} variables to types.  Like term substitutions, we will often write pattern type substitutions concretely using the notation $[A_1/x_1,\ldots,A_n/x_n]$, and apply them postfix.

  The types can be assigned meanings as sets of closed normal forms.  \begin{added}Making precise this meaning is useful for intuition, and it will be used formally in the statement and proof of the soundness of subtyping, Theorem~\ref{thm:subtyping-soundness}\end{added}.  For each type $A$, we define the set of closed normal forms $\mng{A}$ and the set of closed terms $\mng*{A}$, by mutual induction on $A$:
  \[
    \begin{array}{rclcrcl}
      \mng*{A}      & = & \multicolumn{5}{l}{ \{\,M \mid \text{if $M$ reduces to a normal form $\begin{added}U\end{added}$, then $\begin{added}U\end{added} \in \mng{A}$}\,\}}                                                                                                                                                                                                                                                                    \\[2mm]
      \mng{\topty}  & = & \{\,\begin{added}U\end{added} \mid \text{$\begin{added}U\end{added}$ is a closed normal form}\,\}
                    &   & \mng{\intty}                                                                                                                                                                                                                                                                                                                   & = & \{\, \pn{n} \mid n \in \mathbb{Z} \,\}                                             \\
      \mng{\okty}   & = & \{\,V \mid \text{$V$ is a value}\,\}
                    &   & \mng{\textstyle\bigvee_{i=1}^k \begin{added}A_i\end{added}}                                                                                                                                                                                                                                                                    & = & \textstyle\bigcup_{i=1}^k \mng{\begin{added}A_i\end{added}}                        \\
      \mng{\atomty} & = & \{\,\atom{a} \mid \text{$\atom{a}$ is an atom} \,\}
                    &   & \mng{\CompPure{A}}                                                                                                                                                                                                                                                                                                             & = & \mng{\topty} \setminus \mng{A}                                                     \\

      \mng{A \to B} & = & \{\,\abs{x}{M} \mid \forall N \in \mng*{A}.\:M[N/x] \in \mng*{B} \,\}                                                                                                                                                                                                                                                          &   & \mng{\atom{a}}                                              & = & \{\,\atom{a}\,\} \\
      \mng{(A,\,B)} & = & \multicolumn{5}{l}{\{\, (\begin{added}V\end{added},\,\begin{added}U\end{added}) \mid \begin{added}V\end{added} \in \mng{A} \cap \mng{\okty},\, \begin{added}U\end{added} \in \mng{B} \,\} \cup \{ (\begin{added}S\end{added},\,N) \mid \begin{added}S\end{added} \in \mng{A} \cap \mng{\CompPure{\okty}},\,N \in \mng*{B}\,\}}
    \end{array}
  \]
\end{definition}

As can be seen from the semantics of types, the idea is that each type represents a certain set of normal forms.  $\topty$ is the type \begin{added}of\end{added} all normal forms, $\okty$ is the type of all values, $\atomty$ the type of all atoms and $\intty$ the type of all integers.

We use the syntax $(A,\,B)$ for the type of pairs whose components are of type $A$ and $B$ respectively, rather than, say $A \times B$.  This is convenient for our later discussion of the typing of pattern matching, because it means that we can discuss a term instance $p\sigma$ of a pattern $p$, and its type $p\theta$ using the same notation and a \emph{pattern type substitution}.  For example, we will be able to prove that $\types (x,y)[(\abs{z}{z})/x,\,\pn{2}/y] : (x,y)[(\intty \to \intty)/x,\,\intty/y]$.   Note that the type $(A,\,B)$ does not necessarily contain only values -- it can be that an element of this type is a stuck term whenever either $A$ or $B$ may contain a stuck term.  In particular, one can see from the definition that the right summand contains all stuck terms of shape $(\begin{added}S\end{added},\,N)$ where $\begin{added}S\end{added}$ is itself stuck and $N$ is arbitrary (may not even be a normal form).  Also, it may be that $(A,\,B)$ is empty in the case that $A$ \begin{deleted}\del{or $B$}\end{deleted} is empty (e.g. when $A = \CompPure{\topty}$).

As usual, the type $A \to B$ represents those functions that guarantee to map a given term of type $A$ into a normal form of type $B$, or diverge in the process.  Note that we use the term form of the semantics $\mng*{A}$ for both argument and result since our calling convention is call-by-name (but see also Remark~\ref{rmk:cbn-funs}).  Note that the type $A \to B$ is \emph{never} empty, since it always contains $\abs{x}{\abbv{div}}$ -- a function that inevitably diverges certainly satisfies the condition on membership.

The semantics of types justifies the following further abbreviations:
\begin{center}
  \begin{minipage}{.24\textwidth}
    \begin{align*}
      \botty  & \coloneqq \CompPure{\topty} \\
      \pairty & \coloneqq (\topty,\,\topty)
    \end{align*}
  \end{minipage}
  \begin{minipage}{.24\textwidth}
    \begin{align*}
      \pairvalty & \coloneqq (\okty,\,\okty)   \\
      \funty     & \coloneqq \botty \to \topty
    \end{align*}
  \end{minipage}
  \begin{minipage}{.24\textwidth}
    \begin{align*}
      A \onlyto B                     & \coloneqq \CompPure{A} \to \CompPure{B}                         \\
      \textstyle\bigwedge_{i=1}^k A_i & \coloneqq \CompPure{(\textstyle\bigvee_{i=1}^k \CompPure{A_i})}
    \end{align*}
  \end{minipage}
  \begin{minipage}{.24\textwidth}
    \begin{align*}
      \boolty & \coloneqq \atom{true} \vee \atom{false} \\
              &
    \end{align*}
  \end{minipage}\vspace{8pt}
\end{center}

\noindent
Notice that $\pairty$ is the type of all normal forms constructed using pairing, and thus includes those pairings that describe stuck terms.  By contrast, $\pairvalty$ is the type of all pair values (pairs whose components are again values).  Note also that we are able to construct a type whose interpretation is exactly the set of all functions, $\funty$, simply as an abbreviation.  This would be difficult if our types were restricted to sets of values only, since a function $\abs{x}{M}$, whose $M$ immediately goes wrong, would not live inside any type of shape $A \to B$.

\begin{remark}\label{rmk:cbn-funs}
  The type $A \onlyto B$, pronounced $A$ \emph{only to} $B$, which was introduced in \cite{ramsay-walpole-popl2024}, is  defined simply as an abbreviation for $\CompPure{A} \to \CompPure{B}$.  This makes sense because the defining condition for $A \onlyto B$ is ``if $M[N/x]$ reduces to a normal form in $\mng{B}$ then $N \in \mng*{A}$'' and this is the contrapositive of the condition ``if $N \notin \mng*{A}$ then $M[N/x] \in \mng*{\CompPure{B}}$''.  Unfolding the definitions, this latter condition is equivalent to ``if $N$ reduces to a normal form in $\mng{\CompPure{A}}$ then $M[N/x] \in \mng*{\CompPure{B}}$''.  Finally, note that, a basic property of reduction in our language is that any abstraction $\abs{x}{M}$ will satisfy this condition iff it satisfies the condition ``if $N \in \mng*{\CompPure{A}}$ then $M[N/x] \in \mng*{\CompPure{B}}$''.  To see this, suppose $\abs{x}{M}$, on input $N$ that reduces to an $\CompPure{A}$ normal form,  either diverges or returns a $\CompPure{B}$ normal form.  Then $\abs{x}{M}$, on input $N$ which \emph{either diverges or} reduces to an $\CompPure{A}$ normal form, must still in the end either diverge or return a $\CompPure{B}$ normal form.
\end{remark}

\subsection{Subtype Relation}
\label{subsec:subtype-relation}

A key component of the type system is a subtype relation, \begin{added}$A \subtype B$\end{added}, which axiomatises some key properties of complement.  \begin{added}Later, in Theorem~\ref{thm:subtyping-soundness}, we make precise the idea that the meaning of such an assertion is $\mng{A} \subseteq \mng{B}$\end{added}.

\begin{figure}
  \input{figures/simple-subtyping.tex}
  \caption{Subtype relation on types\label{fig:simple-subtyping}}
\end{figure}

\begin{definition}[Subtyping Relations]\label{def:simple-subtyping}
  The subtyping relation is defined inductively by the rules in \cref{fig:simple-subtyping}.  The relation $A \distype B$ holds just if $A \neq B$ and, \emph{either} $A$ and $B$ are both atoms \emph{or} $A,B \in \{\intty,\pairty,\funty,\atomty\}$.
  Define \emph{subtype equivalence}, written $A \tyeq B$, just if $A \subtype B$ and $B \subtype A$.  Since the subtype relation is, by definition, reflexive and transitive, it is immediate that $\sim$ is an equivalence on types.
\end{definition}

The rules \rlnm{Refl}, \rlnm{Trans}, \rlnm{Pair}, \rlnm{Fun} and \rlnm{Atom} are standard. The rule \rlnm{Top} introduces a universal type, and when combined with \rlnm{CompL}, provides the dual notation of the bottom type, $\botty$, which is a subtype of any other type.

The rule \rlnm{Ok} allows us, \begin{added}together with \rlnm{UnionR} and \rlnm{Trans}\end{added}, to derive that each of $\intty, \pairvalty, \funty$ and $\atomty$ \begin{added}is a subtype\end{added} of $\okty$ i.e any \begin{added}normal form\end{added} which is of these types \begin{added}is a value\end{added}. The rule \rlnm{Disj} expresses the disjointness of certain types using the complement operator, \begin{added}for example $\intty$, $\pairty$, $\funty$, $\atomty$, as well as distinct atoms $\atom{a}, \atom{b}$ and so on, are provably disjoint from one another\end{added} \del{and using complement asserts that each type is \emph{not} contained in the other}. This is extended in \del{\rlnm{CompL} and} \rlnm{CompR} - if $A$ is disjoint from $B$, then their complements must be contained in the other.  These two rules ensure that our complement operator behaves like an order-reversing involution, up to equivalence.

In the rule \rlnm{PairC}, we have $\okty$ instead of $\topty$ in the second pair to mirror the asymmetric definition of $\mng{(A, B)}$. It is not necessarily true that $(\topty, \CompPure{B}) \subtyping \CompPure{(A, B)}$ since, if it was, we would be able to derive for example \begin{added}$(\pn{3}\ (\abs{x}{x}),\,\divtm) \in \mng{(\topty,\,\CompPure{\intty})} \subseteq \mng{\CompPure{(\CompPure{\okty},\,\intty)}}$, but we know that $(\pn{3}\ (\abs{x}{x}),\,\divtm) \in \mng{(\CompPure{\okty},\,\intty)}$\end{added}.

The rules \rlnm{UnionL} and \rlnm{UnionR} follow naturally from a set theoretic perspective. They explain the union type as a least upper bound on types.

% \begin{corollary}
%   If $A \subtype B$ then $\mng*{A} \subseteq \mng*{B}$.
% \end{corollary}
% \begin{proof}
%   $\mng*{A} = \{M \mid \text{if $M \Downarrow N$ then $N \in \mng{A}$}\} \subseteq \{M \mid \text{if $M \Downarrow N$ then $N \in \mng{B}$}\} = \mng*{B}$
% \end{proof}

% Subtyping induces a ortholattice i.e. it forms a bounded prelattice equipped with an operator $\CompPure{-}$ that acts as an orthocomplement, up to equivalence.  The requirements are made precise in the following.
% \begin{theorem}[Induced Ortholattice]\label{thm:ortholattice}
%   The following hold for all $A$, $B$ and $C$:
%   \begin{center}
%     \begin{minipage}{.48\textwidth}
%       \begin{tabular}{rl}
%         \rlnm{Join-Intro} & $A \subtype A \vee B$ and $B \subtype A \vee B$\\
%         \rlnm{Meet-Intro} & if $C \subtype A$ and $C \subtype B$ then $C \subtype A \wedge B$\\
%         \rlnm{Join-Elim} & if $A \subtype C$ and $B \subtype C$ then $A \vee B \subtype C$\\
%         \rlnm{Meet-Elim} & $A \wedge B \subtype A$ and $A \wedge B \subtype B$
%       \end{tabular}
%     \end{minipage}
%     \begin{minipage}{.48\textwidth}
%       \begin{tabular}{rl}
%         \rlnm{Involution} & $\CompPure{\CompPure{A}} \tyeq A$\\
%         \rlnm{Complement} & $A \vee \CompPure{A} \tyeq \topty$ and $A \wedge \CompPure{A} \tyeq \botty$\\
%         \rlnm{Reversal} & if $A \subtype B$ then $\CompPure{B} \subtype \CompPure{A}$\\
%         \rlnm{Boundedness} & $A \subtype \topty$ and $\botty \subtype A$
%       \end{tabular}
%     \end{minipage}
%   \end{center}
% \end{theorem}

\subsection{Typing Relation}

We view the type system as a formal mechanism for establishing the truth of properties of the form ``term $M$ either diverges, or reduces to a normal form of the type $A$'', written $M:A$.

\begin{definition}[Typing Formula]\label{def:typing-formula}
  A \emph{typing formula} (or simply \emph{typing}) is an expression of the form $M : A$, where $M$ is a term, and $A$ is a type.  In particular, if $M$ is a term variable, i.e. the typing formula has the form $x:A$, we call this a \emph{variable typing}.  The free variables of a typing formula $M:A$ are just $\fv(M)$.  We extend the notion to sets of formulas, writing $\fv(\Delta)$ for $\textstyle\bigcup_{M:A \in \Delta} \fv(M)$.
\end{definition}

% \begin{definition}[Typing Context]\label{def:typing-context}
%   A \emph{typing context} (denoted \(\Gamma, \Delta, \dots\)) is a finite set of typing formulas
%   \(\{ M_1 : A_1, \dots, M_n : A_n \}\), where each \(M_i : A_i\) is a typing formula (as defined in \cref{def:typing-formula}).

%   In the special case that every typing formula in a typing context is a variable typing, the typing context is called a \emph{variable typing context}.
% \end{definition}

The type system allows for proving sequents of the form $\Gamma \types \Delta$ with meaning, \begin{added}informally\end{added}: if all the formulas in $\Gamma$ are true then some formula in $\Delta$ is true\footnote{\begin{added}We will not formalise this interpretation directly, and instead soundness will be given via progress and preservation theorems.\end{added}}.

\begin{definition}[Typing Judgement and Derivability]
  A \emph{typing judgement} is an assertion of the form $\Gamma \types \Delta$, where $\Gamma$ and $\Delta$ are possibly empty, finite sets of typing formulas.  A type judgement is said to be \emph{derivable} just if it is possible to construct a proof tree concluding the judgement using the rules of Figure~\ref{fig:simple-two-sided-rules}.  Excepting the \rlnm{Var} rule, the conclusion of each rule distinguishes a single typing, which is called the \emph{principal formula} of the rule.

  The side condition on the \rlnm{Match} rule requires that $\mathsf{PatSubst}(p,\,\theta)$, which holds just if $\theta$ is a pattern type substitution with $\dom(\theta) = \fv(p)$ and, for all types $A$ in the range of $\theta$, $A \subtype \okty$.
\end{definition}

\begin{figure}
  \input{figures/simple-two-sided.tex}
  \caption{Two-Sided Type System.}\label{fig:simple-two-sided-rules}
\end{figure}

Two-sided type systems are construed as a kind of sequent calculus for typing formulas.  That is, the rules of the proof system can be divided into those that explain how to \emph{verify} a formula of a certain shape on the right-hand side of the sequent, and those that explain how to \emph{refute} a formula of a certain shape on the \begin{added}left\end{added}.

Since the right rules explain how to verify a typing, they are just the familiar rules of a traditional type system, but generalised slightly to allow for an arbitrary context $\Delta$ on the right (though one could instead consider an intuitionistic sequent calculus and drop this generalisation).  Therefore we will not explain \del{further} the standard rules \rlnm{Top}, \rlnm{Num}, \rlnm{Atom}, \rlnm{SubR}, \rlnm{BinOp}, \rlnm{RelOp}, \rlnm{Fix}, \rlnm{Pair}, \rlnm{Abs} and \rlnm{App}.

\paragraph{The Match Rule}
In the \rlnm{Match} rule, we abuse notation by considering a pattern-type substitution $\theta$, which is formally a finite map from term variables to types, also as a set of typings $\{x:A \mid \theta(x) = A\}$.  The rule allows free choice of a pattern substitution $\theta$, so long as it satisfies the side condition $\mathsf{PatSubst}(p,\,\theta)$.  This requires that any type $\theta(x)$ chosen as the type of a free variable $x$ of pattern $p$ provably does not contain any stuck term.  This is essential for soundness, because otherwise it could be that the scrutinised term $M$ could go wrong even when $\types M:\bigvee_{i=1}^k p_i\theta_i$\begin{added}, for example the judgement $\types \matchtm{(1,2)\ 3}{x \mapsto 1} : \intty$ would be provable\end{added}.

\begin{figure}
  \begin{mdframed}
    \begin{mathpar}
      \infer[WeakL]{
        \Gamma \types \Delta
      }
      {
        \Gamma,\,M:A \types \Delta
      }
      \and
      \infer[WeakR]{
        \Gamma \types \Delta
      }
      {
        \Gamma \types M:A,\,\Delta
      }
      \and
      \infer[Prj$i$]{
        \Gamma \types M : (A_1,\,A_2),\,\Delta
      }
      {
        \Gamma \types \pi_i(M) : \begin{added}A_i\end{added},\,\Delta
      }\;A_1,\,A_2 \subtype \okty
      \and
      \infer[PrjE$i$]{
        \Gamma,\, M: \mathsf{PairVal} \types \Delta
      }
      {
        \Gamma,\, \pi_i(M) : \okty \types \Delta
      }
      \\
      \infer[If]{
        \Gamma \types M : \boolty,\,\Delta
        \and
        \Gamma \types N : A,\,\Delta
        \and
        \Gamma \types P : A,\,\Delta
      }
      {
        \Gamma \types \iftm{M}{N}{P} : A,\,\Delta
      }
      \and
      \infer[IfE]{
        \Gamma,\,M : \boolty \types \Delta
      }
      {
        \Gamma,\,\iftm{M}{N}{P} : \okty \types \Delta
      }
    \end{mathpar}
  \end{mdframed}
  \caption{Admissible rules for weakening, projections and conditionals.}\label{fig:admissible-positive}
\end{figure}

To understand the mechanics of the \rlnm{Match} rule, let us first observe that all the rules of Figure~\ref{fig:admissible-positive} are admissible.
\begin{added}Combining \rlnm{Match} and \rlnm{WeakL}\end{added}, it follows that the less expressive, but perhaps more familiar, one-sided rule for matching is also admissible.
\begin{mathpar}
  \infer[Match$_0$]{
  \Gamma \types M : \textstyle\bigvee_{i=1}^k \theta_i p_i\\
  \Gamma,\,\theta_i \types N_i : A \;(\forall i \in [1,k])
  }
  {
  \Gamma \types \matchtm{M}{|_{i=1}^k p_i \mapsto N_i} : A
  }\;
  \begin{array}{l}
    \forall i \in [1,k].\:\mathsf{PatSubst}(p_i,\,\theta_i) \\
    \quad \wedge\ \fv(p_i) \cap \fv(\Gamma,\Delta) = \varnothing
  \end{array}
\end{mathpar}
This simpler match rule is sufficient when every case of the match produces a conclusion of the desired type $A$, such as the following instance with $\theta_1 = \varnothing$, $\theta_2 = [\intty/\del{n}\add{y}]$ and $\Gamma = \{x:(\atom{data},\,\intty)\}$:
\begin{mathpar}
  \infer*[left=Match$_0$]{
    \infer*[left=SubR]{
      \infer*[left=Var]{ }{
        \Gamma \types x : (\atom{data},\,\intty)
      }
    }
    {
      \Gamma \types x : \atom{err} \vee (\atom{data},\,\intty)
    }
    \\
    \infer*[left=Num]{ }
    {
      \Gamma \types \del{0}\add{\pn{0}} : \intty
    }
    \\
    \infer*[left=Var]{ }
    {
      \Gamma,\,\del{n}\add{y}:\intty \types \del{n}\add{y}:\intty
    }
  }
  {
    \Gamma \types \matchtm{x}{\atom{err} \mapsto \del{0}\add{\pn{0}} \mid (\atom{data},\del{n}\add{y}) \mapsto \del{n}\add{y} } : \intty
  }
\end{mathpar}

However, this simpler form of rule is insufficient when there is some match case that produces a result of an incompatible type, yet the case does not apply.  For example, if we were to modify the above program so that errors are propagated:
\[
  \matchtm{x}{\atom{err} \mapsto \atom{err} \mid (\atom{data},\del{n}\add{y}) \mapsto \del{n}\add{y} }
\]
% \begin{deleted}\del{Here, we can use the assumption of the second premise to refute that the case ever occurred}\end{deleted} 
\begin{added}Here, we can use the extra left assumption of the second premise in the \rlnm{Match} rule -- specifically, $x : \atom{err}$ -- to refute that the case ever occurred.\end{added} Writing $A$ for the type $(\atom{data},\,\intty)$, we have:
\begin{mathpar}
  \infer*{
    \infer*[left=SubR]{
      \infer*[left=Var]{ }{
        \Gamma \types x : A
      }
    }
    {
      \Gamma \types x : \atom{err} \vee A
    }
    \\
    \infer*[left=SubL]{
      \infer*[Left=CompL]{
        \infer*[Left=Var]{ }{
          \Gamma \types x : A,\,\atom{err}:\intty
        }
      }
      {
        \Gamma,\, x:\CompPure{A} \types \atom{err}:\intty
      }
    }
    {
      \Gamma,\, x : \atom{err} \types \atom{err}:\intty
    }
    \\
    \infer*[left=\begin{added}Var\end{added}]{ }
    {
      \Gamma,\,\del{n}\add{y}:\intty,\,x : A \types \del{n}\add{y}:\intty
    }
  }
  {
    \Gamma \types \matchtm{x}{\atom{err} \mapsto \atom{err} \mid (\atom{data},\del{n}\add{y}) \mapsto \del{n}\add{y} } : \intty
  }
\end{mathpar}

In the previous example, \begin{added}the guard expression $x$ did not occur in the branches of the match expression.  When it does, the \rlnm{Match} rule provides the assumptions $M:p_i\theta_i$, describing the branch condition, to refine the type of the guard inside the branches\end{added}.  This is essential for typing the following innocent-looking variation, in which we assume only that the guard matches one of the patterns (and that the data is an integer in the second case), i.e. $\Gamma = \{x:\atom{err} \vee (\atom{data},\,\intty)\}$, and again writing $A$ to abbreviate $(\atom{data},\,\intty)$:
\begin{mathpar}
  \infer*{
    \infer*[left=Var]{ }{
      \Gamma \types x : \atom{err} \vee A
    }
    \\
    \infer*[left=Num]{ }
    {
      \Gamma,\,x : \atom{err} \types \del{0}\add{\pn{0}}:\intty
    }
    \\
    \infer*[left=Prj2]{
      \infer*[left=Var]{ }{
        \Gamma,\,\del{n}\add{y}:\intty,\,x : A \types x : A
      }
    }
    {
      \Gamma,\,\del{n}\add{y}:\intty,\,x : A \types \pi_2(x) : \intty
    }
  }
  {
    \Gamma \types \matchtm{x}{\atom{err} \mapsto \del{0}\add{\pn{0}} \mid (\atom{data},\del{n}\add{y}) \mapsto \pi_2(x) } : \intty
  }
\end{mathpar}

\noindent
This way we can capture properties of the behaviour of matching expressions.  For example, the following are all derivable types for $M = \abs{x}{\matchtm{x}{\atom{err} \mapsto \atom{err} \mid (\atom{data},\del{n}\add{y}) \mapsto \del{n}\add{y} }}$:

\begin{center}
  \begin{minipage}{.45\linewidth}
    \[
      \begin{aligned}
        \types M & : \atom{err} \to \atom{err}         \\
        \types M & : (\atom{data},\,\intty) \to \intty
      \end{aligned}
    \]
  \end{minipage}
  \begin{minipage}{.45\linewidth}
    \[
      \begin{aligned}
        \types M & : \atom{err} \vee (\atom{data},\,\intty) \to \atom{err} \vee \intty      \\
        \types M & : \CompPure*{\atom{err} \vee (\atom{data},\,\okty)} \to \CompPure{\okty}
      \end{aligned}
    \]
  \end{minipage}
\end{center}
% The last of these follows directly from the rule \rlnm{MatchE}:
% \begin{mathpar}
%   \infer*[left=Abs]{
%     \infer*[left=MatchE]{
%       \infer*[left=Var]{ }{
%         x: \CompPure*{\atom{err} \vee (\atom{data},\,\okty)} \types x: \CompPure*{\atom{err}
%           \vee (\atom{data},\,\okty)}
%       }
%     }
%     {
%       x: \CompPure*{\atom{err} \vee (\atom{data},\,\okty)} \types \matchtm{x}{\atom{err} \mapsto \atom{err} \mid (\atom{data},n) \mapsto n } : \CompPure{\okty}
%     }
%   }
%   {
%     \types \abs{x}{\matchtm{x}{\atom{err} \mapsto \atom{err} \mid (\atom{data},n) \mapsto n }} :  \CompPure*{\atom{err} \vee (\atom{data},\,\okty)} \to \CompPure{\okty}
%   }
% \end{mathpar}

The bottom-right typing allows us to certainly conclude that the application of this function to a term of type $(\atom{dat},\,\intty)$, that is, where we have misspelled the atom, will not reduce to a value; e.g.
\begin{mathpar}
  \types (\abs{x}{\matchtm{x}{\atom{err} \mapsto \atom{err} \mid (\atom{data},\del{n}\add{y}) \mapsto \del{n}\add{y} }})\,(\atom{dat},\,3) : \CompPure{\okty}
\end{mathpar}
However, we will look in detail at how the system can be used for refuting that a term reduces to a value in the sequel.

\paragraph{The CompR Rule}
The \rlnm{CompR} explains that to verify $M:\CompPure{A}$, it suffices to refute $M:A$.  This is because\del{, as can be seen in our soundness result}, verifying \begin{added}$M \in \mng*{\CompPure{A}}$\end{added} amounts to showing that $M$ either diverges or reduces to a normal form of type $\CompPure{A}$, and refuting $M:A$ amounts to showing that $M$ does not reduce to a normal form of the type $A$.  Since every term either diverges or normalises, it is straightforward to see that these are the exact same property.

\begin{remark}\label{rmk:comp-comparison-with-itpde}
  A weakness of Ramsay and Walpole's work is that a complement operator was not supported by the type system.  In particular, as discussed in \cite{ramsay-walpole-popl2024}, the rule \rlnm{CompR} was not sound, that is, its conclusion does not necessarily follow from its premise.  This is because in that work, types are sets of \emph{values}, and therefore the complement of a type $A$ are those \emph{values} that are not in the type $A$.  Thus it does \emph{not} suffice to refute $M:A$ in order to verify $M:\CompPure{A}$, e.g. when $M$ goes wrong.
  % Consider an instance of \rlnm{CompR} where $\Gamma$ and $\Delta$ are empty.  According to the semantics, the meaning of the judgement in the premise is: if $M$ evaluates to a value of type $A$ then false.  In other words, $M$ does not evaluate to a value of type $A$.  This amounts to a proof that $M$ either (i) diverges or (ii) evaluates, but to a non-$A$ value, or (iii) goes wrong.  On the other hand, the meaning of the judgement in the conclusion is: if true then $M$ either diverges or evaluates to a value outside of $A$.  Clearly, this conclusion does not follow when the premise is true by virtue of point (iii), when $M$ goes wrong.
  Ramsay and Walpole remedied this unsoundness by considering a system with an alternative semantics, inspired by success types \cite{lindahl-sagonas-ppdp2006}.
  % Under this semantics, the meaning of the general judgement $M_1:A_1,\,\ldots,\,M_k:A_k \types N_1:B_1,\,\ldots,\,N_m:B_m$ is rather: if all the $M_i$ each evaluate to a value of type $A_i$, then some $N_i$ either diverges, or goes wrong, or evaluates to a value of type $B_i$. 
  However, this semantics has \begin{deleted}\del{an}\end{deleted}\begin{added}a\end{added} fundamental shortcoming -- verifying a typing no longer excludes the possibility that it \del{may} may go wrong, so it is impossible to use this system to show that a given program cannot go wrong.  Like success types, one can \emph{only}  certify that a program does \emph{not} reduce to a value.
\end{remark}

\section{Refutation}\label{sec:refutation}

\begin{figure}
    \input{figures/simple-two-sided-left.tex}
    \caption{Derivable rules related to refutation.\label{fig:simple-two-sided-left}}
\end{figure}

The most interesting feature of the type system is the ability to deduce consequences of terms reducing to a normal form of a certain type and, ultimately, to refute typing formulas (by deducing absurdity as a consequence).
As expressed by the soundness of the type system (Theorem~\ref{thm:progress-preservation}), refuting a typing $M:\okty$, that is, proving a judgement $M:\okty \types\ $ with $M:\okty$ on the left-hand side and nothing on the right, gives a guarantee that $M$ will \emph{not} reduce to a value -- in other words, $M$ either diverges or it goes wrong.  Thus, refutation of type assignments provides a mechanism for finding genuine bugs -- with no false positives (assuming one counts divergence as a kind of bug).

The \emph{left} rules of the system (those for which the principal formula in the conclusion of the rule is on the left of the turnstile) explain how to refute a typing formula.  Using \rlnm{OpE}, one can directly refute that applying an arithmetic operation to an argument that is a function evaluates to a value, or, using \rlnm{AppE}, one can refute that applying a integer as if it were a function would evaluate.
\begin{mathpar}
    \infer*[left=OpE]{
        \infer*[Left=SubL]{
            \infer*[Left=CompL]{
                \infer*[Left=Abs]{
                    \infer*[Left=Top]{ }{ \pn{3} : \intty ,\,x:\botty \types x:\topty}
                }
                {
                    \pn{3} : \intty \types \abs{x}{x} : \funty
                }
            }
            {
                \pn{3} : \intty,\,(\abs{x}{x}) : \CompPure{\funty} \types
            }
        }
        {
            \pn{3} : \intty,\,(\abs{x}{x}) : \intty \types
        }
    }
    {
        \pn{3} \otimes (\abs{x}{x}) : \okty \types
    }
    \and
    \infer*[left=AppE]{
        \infer*[Left=SubL]{
            \infer*[Left=CompL]{
                \infer*[Left=Num]{ }{\types \pn{3} : \intty}
            }
            {
                \pn{3} : \CompPure{\intty} \types
            }
        }
        {
            \pn{3} : \funty \types
        }
    }
    {
        \pn{3}\,(\pn{1},\,\pn{2}) : \okty \types
    }
\end{mathpar}
Similarly, by \rlnm{MatchE} and \rlnm{PairE}, we can refute evaluation for e.g. scrutinizing a term that does not match any pattern, or trying to form a pair where one of the components does not evaluate.

However, in each case, it is the term itself that directly goes wrong.  What about if we want to show that, for example, a redex $(\abs{x}{M})\,N$ is contracted successfully, but that something in the body $M[N/x]$ will later get stuck; or that a match expression $\matchtm{\atom{a}}{\atom{a} \mapsto N \mid \atom{b} \mapsto P}$ successfully evaluates one of the branches, but that evaluation of the branch gets stuck?  For this purpose we introduce the left rules of Figure~\ref{fig:simple-two-sided-left}.  However, we will later show that these rules are inherent to the system: each is already derivable from the rules of Figure~\ref{fig:simple-two-sided-rules}.

The rule \begin{added}\rlnm{BotL}\end{added} is dual to \rlnm{Top}, and expresses that one may always refute that a term evaluates to a normal form inside the empty type.  The rules \rlnm{NumL}, \rlnm{FunL}, \rlnm{AtomL} and \rlnm{PairL1}  express the self-evident facts that one can refute that a normal form lives inside the complement of its natural type.

The \rlnm{AppL} rule gives a principle for refuting an application when the operator is, in fact, known to evaluate to a function.  It explains that, to refute that application $M\,N$ evaluates to an $A$, we can show that $M$ is a function which, to deliver an $A$, necessarily requires a $B$ as input, i.e. $M:B \onlyto A$, and then we \del{need only}\add{only need to} also show that the actual parameter $N$ does not evaluate to a $B$.  Note that, by the abbreviation, $B \onlyto A$ is $\CompPure{B} \to \CompPure{A}$ so it literally states that every input outside of $B$ will result in an output outside of $A$, which is equivalent \begin{added}to the contrapositive\end{added} (see also Remark~\ref{rmk:cbn-funs}).

The rule \rlnm{FixL} gives a principle for refuting that a fixpoint $\fixtm{x}{M}$ evaluates to an $A$.  To do so, we need only show that the argument being an $A$ is a consequence of the body evaluating to an $A$.  If the fixpoint is thought of as the limit of some chain of approximations, then the rule allows for working backwards along the chain, showing that when some better approximation is known to be outside of $A$, then the approximation preceding it was necessarily outside of $A$ too.

The rule \begin{added}\rlnm{MatchL}\end{added} gives a principle for refuting that a match expression $\matchtm{M}{\mid_{i=1}^k p_i \mapsto N_i}$ evaluates to an $A$ when control does get passed to the branches because the scrutinee matches some pattern.  Note that, for convenience, when $\theta$ is used as a set of variable typings, we use the abbreviation $\CompPure{\theta}$ for the set $\{x:\CompPure{A} \mid x:A \in \theta\}$.  To justify the refutation, we must show that $M$ will evaluate to some value matching a pattern and then, for each branch $i \in [1,k]$, we refute that the branch returned an $A$ by either:
\begin{itemize}
    \item refuting $N_i : A$, thus proving directly that this branch did not return an $A$
    \item or, refuting $M : p_i\theta$, thus proving that this branch could not have been taken
    \item or, proving $\CompPure{\theta_i}$, thus showing that (although the branch may return an $A$) a consequence of the branch returning an $A$ is that the assignment of types to the pattern variables is inconsistent with the actual matching $\theta$.
\end{itemize}
Note that, since $\CompPure{\theta}$ appears on the right of the judgement, it suffices only to show that \emph{some} pattern variable is necessarily inconsistent with the actual matching.

The rules \rlnm{PairLL} and \rlnm{PairRL} gives the familiar principle that, to refute a pair $(M_1,\,M_2)$ evaluates to a pair of an $A_1$ and an $A_2$ specifically, it suffices to refute that one component $M_i$ evaluates to an $A_i$. However, \add{to reflect the asymmetric definition of $\mng{(A,\,B)}$ (as discussed for rule \rlnm{PairC} in \cref{subsec:subtype-relation}), the \rlnm{PairRL} rule imposes an additional constraint: if the refutation is based on the second component, the first must also reduce to a value, provided it does not diverge. }\del{when justifying according to the second component. }\del{The rules \rlnm{PrjL$i$} give a principle for refuting that a projection $\pi_i(M)$ yields a normal form of some type $A$.  In this case, it suffices to show that the pairing will have an incompatible component in position $i$. }\add{This asymmetry also affects the rules \rlnm{PrjL$i$} for refuting projections; to refute that a projection $\pi_i(M)$ yields a normal form of some type $A$, we show an incompatible component at position $i$, \emph{and}, for second projection ($i=2$) specifically, we must also ensure the first component would not go wrong. }

Finally, the rule \rlnm{IfL} allows us to refute that a conditional expression reduces to a value of type $A$ by refuting that both of its branches reduce to a value of type $A$.

% \begin{lemma}
%   The following rules are derivable.
%   \begin{mathpar}
%     \infer[SwapL]{  
%       \Gamma \types M:\CompPure{A},\,\Delta
%     }
%     {
%       \Gamma,\,M:A \types \Delta
%     }
%     \and
%     \infer[SwapR]{
%       \Gamma,\,M:\CompPure{A} \types \Delta 
%     }
%     {
%       \Gamma \types M:A,\,\Delta
%     }
%   \end{mathpar}
% \end{lemma}
% \begin{proof}
% The first follows from the fact that $A \subtype \CompPure*{\CompPure{A}}$ and the second follows from the fact that $\CompPure*{\CompPure{A}} \subtype A$:
% \begin{mathpar}
%   \infer*[left=SubL]{
%     \infer*[left=CompL]{
%         \Gamma \types M:\CompPure{A} \Delta
%     }
%     {
%       \Gamma,\,M:\CompPure*{\CompPure{A}} \types \Delta
%     }
%   }
%   {
%     \Gamma,\,M:A \types \Delta
%   }
%   \and
%   \infer*[left=SubR]{
%     \infer*[left=CompR]{
%         \Gamma,\,M:\CompPure{A} \types \Delta
%     }
%     {
%       \Gamma \types M:\CompPure*{\CompPure{A}},\,\Delta
%     }
%   }
%   {
%     \Gamma \types M:A,\,\Delta
%   }
% \end{mathpar}
% \end{proof}

\begin{theorem}\label{thm:left-rules-derivable}
    All the rules of Figure~\ref{fig:simple-two-sided-left} are derivable.
\end{theorem}
\begin{added}
    \begin{proof}\label{proof:left-rules-derivable}
        First observe that the swapping rules are derivable:
        \begin{mathpar}
            \infer*[left=SubL]{
                \infer*[left=CompL]{
                    \Gamma \types M:\CompPure{A} \Delta
                }
                {
                    \Gamma,\,M:\CompPure*{\CompPure{A}} \types \Delta
                }
            }
            {
                \Gamma,\,M:A \types \Delta
            }
            \and
            \infer*[left=SubR]{
                \infer*[left=CompR]{
                    \Gamma,\,M:\CompPure{A} \types \Delta
                }
                {
                    \Gamma \types M:\CompPure*{\CompPure{A}},\,\Delta
                }
            }
            {
                \Gamma \types M:A,\,\Delta
            }
        \end{mathpar}
        All remaining rules are obtained by reformulating a rule of Figure~{\ref{fig:simple-two-sided-rules}} using \rlnm{CompL}, \rlnm{CompR}, \rlnm{SwapL} and \rlnm{SwapR}, e.g.:
        \begin{description}
            \item[\rlnm{FixL}] Suppose $\Gamma,\,M:A \types x:A,\,\Delta$.  Then by \rlnm{CompL} and \rlnm{CompR}, $\Gamma,\,x:\CompPure{A} \types M:\CompPure{A},\,\Delta$.  Then $\Gamma \types \fixtm{x}{M} : \CompPure{A},\,\Delta$ by \rlnm{Fix}, and the result follows from \rlnm{SwapL}.
        \end{description}
    \end{proof}
\end{added}

\begin{added}
    \begin{remark}
        \label{rmk:absence-of-arrow-left-rule}
        Using the method from the proof of \cref{thm:left-rules-derivable}, we can derive from \rlnm{Abs} a rule for refuting the typing of an abstraction:
        \begin{mathpar}
            \infer*[left=AbsL]{
                \Gamma,\, M:\CompPure{B} \types x:\CompPure{A},\, \Delta
            }
            {
                \Gamma,\,\abs{x}{M}:\CompPure*{A \to B} \types \Delta
            }
        \end{mathpar}
        This rule is equivalent to the \rlnm{CoAbs} rule from \cref{fig:coabs-coapp}, which we will discuss further in \cref{sec:coto}.

        One might also expect a left rule for refuting formulas of shape $M:A \to B$. Semantically, to refute $M : A \to B$ requires showing there exists some $N \in \mng*{A}$ such that $M\,N \notin \mng*{B}$. This suggests a rule like:
        \begin{mathpar}
            \infer*[left=Arr-Left]{
                \Gamma \types N:A,\,\Delta \\
                \Gamma,\, M\,N:B \types \Delta
            }
            {
                \Gamma,\,{M}:A \to B \types \Delta
            }
        \end{mathpar}
        However, this rule would not be sound, since refuting $M\,N:B$ does not establish $M\,N\notin \mng*{B}$ -- it can be that $M\,N$ diverges.  Such a rule would also violate a nice property of our rules, namely that every term occurring in a premise of the rule is a subterm of a term occurring in the conclusion.
    \end{remark}
\end{added}

\begin{remark}\label{rmk:left-rules-comparison-with-itpde}
    Our type system in Figure~\ref{fig:simple-two-sided-rules} contains just one left rule for each shape of stuck term, explaining the conditions under which such a term would get stuck.  In contrast, in the two-sided systems of \cite{ramsay-walpole-popl2024}, there are a large number of primitive left rules, whose choice appears quite arbitrary and lack a convincing justification for their inclusion.  For example, the two-sided system for PCF includes a rule \rlnm{Dis} for refuting that a term $M$ evaluates to a value of some type by means of verifying that it either diverges or evaluates to a value of some disjoint type, but they choose \emph{not} to include a rule like our derived rule \rlnm{PairL1}.  The reason given is that there is much overlap between these rules, for example they can both be used to show $(\pn{1},\,\pn{2}) : \mathsf{Nat} \types$.  However, they are ultimately incomparable in expressive power, \rlnm{Dis} is applicable even when the term subject $M$ is not literally a pair, whereas \rlnm{PairL1} applies even when the subject $(P,\,Q)$ may not evaluate to a pair value, for example when $P$ or $Q$ goes wrong.  In our system, however, both rules are derivable so there is no choice to make -- the system of Figure~\ref{fig:simple-two-sided-rules} inherently contains both reasoning principles already.

    \begin{added}
        Note, there are other differences between the rules of our system and analogous rules in the system of \cite{ramsay-walpole-popl2024}, which stem from the change to allow the typing of normal forms.  For example, the rule \rlnm{IfzL2} of \cite{ramsay-walpole-popl2024} has an analogue in our \rlnm{IfL}.  However, the latter has an additional premise $\Gamma \types M:\mathsf{Bool},\,\Delta$.  Otherwise judgements such as $\mathsf{if}\ \pn{2}\ \mathsf{then}\ \pn{3}\ \mathsf{else}\ \pn{4} : \CompPure{\okty} \types$ would be derivable, and thus the system would be unsound.
    \end{added}
\end{remark}

\subsection{H\'ebert's Examples}

We will demonstrate some of these rules by examining the following sequence of three programs taken from Fred H\'ebert's introductory text on Erlang \cite{hebert2013}, where they are used to illustrate the power, and limitations, of success types and the Dialyzer for finding bugs.  Although Erlang is a call-by-value language and ours is not, it makes no essential difference to these examples since our pattern matching is strict.

For the purpose of presenting the examples, a let construct is useful.  As is well known, this can be defined as an abbreviation $\lettm{x}{M}{N}$ for the redex $(\abs{x}{N})\,M$.  Thus, it is easy to see that the following typing rules are derivable:
\begin{mathpar}
    \infer[Let]{
        \Gamma \types M : B,\,\Delta
        \and
        \Gamma,\,x:B \types N:A,\,\Delta
    }
    {
        \Gamma \types \lettm{x}{M}{N} : A,\,\Delta
    }
    \and
    \infer[LetL]{
        \Gamma,\,N:A \types x:B,\,\Delta
        \and
        \Gamma,\,M:B \types \Delta
    }
    {
        \Gamma,\,\lettm{x}{M}{N} : A \types \Delta
    }
\end{mathpar}
The first of these rules should be familiar.  The second explains how to refute that a let reduces to a normal form of type $A$.  To do so, we can first derive necessary conditions on $x$, namely that it has some type $B$, in order for the term $N$ to reduce to an $A$.  Then we need only show that the term $M$\begin{added}, whose result is bound by $x$,\end{added} \emph{doesn't} satisfy this necessary type $B$.

The first program simply defines a wrapper around addition and misapplies it to an atom.
\[
    \lettm{\someOp}{\abs{xy}{x + y}}{\someOp\,\pn{5}\,\atom{you}}
\]

First observe that, if applying $\someOp$ to $\pn{5}$ and $\atom{you}$ were to yield a value, then it must be that $\someOp$ is \emph{not} a function that \emph{necessarily} needs its second argument to be an integer.  We will express this property, that a curried function requires its second argument to be an integer, by the type $\topty \to \intty \onlyto \okty$.  We will let the reasoning by which we chose this type remain a mystery for the moment, except to say that this type represents those functions which, if they are given any first argument, either diverge or return a function that necessarily requires its argument to be an integer (to return a value).  We will return to this topic in Section~\ref{sec:coto}.  For brevity, let $\Gamma \coloneqq \{\someOp : \topty \to \intty \onlyto \okty\}$ and $\Delta \coloneqq \{\someOp : \CompPure*{\topty \to \intty \onlyto \okty}\}$.
\begin{mathpar}
    \infer*[left=AppL]{
        \infer*[Left=CompR]{
            \infer*[Left=App]{
                \infer*[Left=Var]{ }
                {
                    \Gamma \types \someOp : \topty \to \intty \onlyto \okty
                }
                \and
                \infer*[Left=Top]{ }
                {
                    \Gamma \types \pn{5} : \topty
                }
            }
            {
                \Gamma \types \someOp\,\pn{5} : \intty \onlyto \okty
            }
        }
        {
            \types \someOp{}\,\pn{5} : \intty \onlyto \okty,\,\Delta
        }
        \and
        \infer*[left=SubL]{
            \infer*[Left=AtomL]{ }{\atom{you} : \CompPure*{\atom{you}} \types \Delta}
        }
        {
            \atom{you} : \intty \types \Delta
        }
    }
    {
        \someOp{}\,\pn{5}\,\atom{you} : \okty \types \Delta
    }
\end{mathpar}

Next, observe that, by definition, $\someOp$ \emph{does} necessarily require its second argument to be an integer (in order to return a value), i.e. it is straightforward to derive: $\types \abs{xy}{x + y} : \topty \to \intty \onlyto \okty$.
% \begin{mathpar}
%     \infer*[left=Abs]{
%         \infer*[left=Abn]{
%             \infer*[left=OpE]{
%                 \infer*[left=Var]{ }
%                 {
%                     x:\topty,\,x:\intty,\,y:\intty \types y:\intty
%                 }
%             }
%             {
%                 x:\topty,\,x+y : \okty \types y:\intty
%             }
%         }
%         {
%             x:\topty \types \abs{y}{x + y} : \intty \onlyto \okty
%         }
%     }
%     {
%         \types \abs{xy}{x + y} : \topty \to \intty \onlyto \okty
%     }
% \end{mathpar}

Then we can put these two derivations together to refute that this program reduces to a value:
\begin{mathpar}
    \infer*[left=LetL]{
        \someOp\,\pn{5}\,\atom{you} : \okty \types \someOp : \CompPure*{\topty \to \intty \onlyto \okty}
        \and
        \infer*[left=CompL]{
            \types \abs{xy}{x + y} : \topty \to \intty \onlyto \okty
        }
        {
            \abs{xy}{x + y} : \CompPure*{\topty \to \intty \onlyto \okty} \types
        }
    }
    {
        \lettm{\someOp}{\abs{xy}{x + y}}{\someOp\,\pn{5}\,\atom{you}} : \okty \types
    }
\end{mathpar}

The second program in H\'ebert's sequence expresses the same problem more indirectly.  There are defined functions $\countTm$ and $\accountTm$ used for projecting out the numerical value of the transaction and the account identity respectively, and the transaction record is constructed using $\moneyTm$.  Then, instead of calling $\someOp$ directly, the second argument to $\someOp$ is itself a call to $\accountTm$ which, when passed the result of $\moneyTm$, returns an atom.

\begin{lstlisting}[xleftmargin=.2\textwidth]
  let money = $\abs{xy}{(\atom{give},(x,y))}$ in
  let count = $\abs{x}{\matchtm{x}{(\atom{give},(y,z)) \mapsto y}}$ in
  let account = $\abs{x}{\matchtm{x}{(\atom{give},(y,z)) \mapsto z}}$ in
  let op = $\abs{xy}{x + y}$ in
  let tup = $\moneyTm\,\pn{5}\,\atom{you}$ in
         op (count tup) (account tup)
\end{lstlisting}

\noindent
We observe that, for the body of the innermost let to reduce to a value, it \emph{cannot} be that all of the following are true:
\begin{enumerate}[(i)]
    \item $\someOp$ requires its second argument to evaluate to an integer
    \item \emph{and}, for $\accountTm$ to return an integer, its argument must have evaluated to a value of shape $(\atom{give},\,(M,\,N))$, with $N$ an integer
    \item \emph{and}, $\mathit{tup}$ is not of this shape
\end{enumerate}
In other words, if this term evaluates, then one of (i)--(iii) is false.  So let $\Delta_0$ be the following encoding of this situation:
\[
    \someOp : \CompPure*{\topty \to \intty \onlyto \okty},\,\accountTm : \CompPure*{(\atom{give},(\topty,\intty)) \onlyto \intty},\,\mathit{tup} : (\atom{give},(\topty,\intty))
\]
Then, it is straightforward to prove $\someOp\,(\countTm\,\mathit{tup})\,(\accountTm\,\mathit{tup}) : \okty \types \Delta_0$ using \rlnm{AppL}, \rlnm{App}, \rlnm{CompR} and \rlnm{Var}.
Moving outwards by one let, we find the definition of $\mathit{tup}$ and we observe that for (the definition of) $\mathit{tup}$ to evaluate to a tuple of shape $(\atom{give},\,(M,\,N))$, with $N$ an integer, it \emph{cannot} be that:
\begin{enumerate}[({iii})']
    \item For $\moneyTm$ to return a tuple of shape $(\atom{give},\,(M,\,N))$, with $N$ evaluating to an integer,  requires its second argument to be an integer.
\end{enumerate}
This is because $\moneyTm$ is not given an integer as second argument in the call in the definition of $\mathit{tup}$.
% In other words:
% \begin{mathpar}
%     \moneyTm\,\pn{5}\,\atom{you} : (\atom{give},\,\topty,\,\intty) \types \moneyTm : \CompPure*{\topty \to \intty \onlyto (\atom{give},\,(\topty,\,\intty))}
% \end{mathpar}
We can prove this either by working backwards using \rlnm{AppL} or we can exchange the two sides of the judgement using \rlnm{SwapL} and \rlnm{CompR},
% and then recall that $\intty \onlyto (\atom{give},\,(\topty,\,\intty))$ is an abbreviation for $\CompPure{\intty} \to \CompPure{(\atom{give},\,(\topty,\,\intty))}$ 
and prove the resulting judgement as we would in a traditional type system:
\begin{mathpar}
    \infer*[left=SwapL]{
        \infer*[left=CompR]{
            \moneyTm : \topty \to \CompPure{\intty} \to \CompPure{(\atom{give},\,(\topty,\,\intty))} \types \moneyTm\,\pn{5}\,\atom{you} : \CompPure{(\atom{give},\,(\topty,\,\intty))}
        }
        {
            \types \moneyTm : \CompPure*{\topty \to \intty \onlyto (\atom{give},\,(\topty,\,\intty))}, \moneyTm\,\pn{5}\,\atom{you} : \CompPure{(\atom{give},\,(\topty,\,\intty))}
        }
    }
    {
        \moneyTm\,\pn{5}\,\atom{you} : (\atom{give},\,(\topty,\,\intty)) \types \moneyTm : \CompPure*{\topty \to \intty \onlyto (\atom{give},\,(\topty,\,\intty))}
    }
\end{mathpar}
In either case, we have that, for the innermost let to evaluate, we must have that at least one of (i), (ii) and (iii)' is false, so let us abbreviate this conclusion by $\Delta_1$:
\begin{align*}
    \someOp    & : \CompPure*{\topty \to \intty \onlyto \okty}                             \\
    \accountTm & : \CompPure*{(\atom{give},(\topty,\intty)) \onlyto \intty}                \\
    \moneyTm   & : \CompPure*{\topty \to \intty \onlyto (\atom{give},\,(\topty,\,\intty))}
\end{align*}
Then we can combine the above derivations, using \rlnm{LetL} and the weakening rule, to prove:
\begin{mathpar}
    \infer*{
        \someOp\,(\countTm\,\textit{tup})\,(\accountTm\,\textit{tup}) \begin{added}: \okty\end{added} \types \mathit{tup} : (\atom{give},(\topty,\intty)),\Delta_1
        \quad
        \moneyTm\,\pn{5}\,\atom{you} : (\atom{give},(\topty,\,\intty)) \types \Delta_1
    }
    {
        \lettm{\textit{tup}}{\moneyTm\,\pn{5}\,\atom{you}}{\someOp\,(\countTm\,\textit{tup})\,(\accountTm\,\textit{tup})} : \okty \types \Delta_1
    }
\end{mathpar}
This, in turn, becomes the left premise of \rlnm{LetL} for the let expression $\lettm{\someOp}{\abs{xy}{x + y}}{\ldots}$, and so on, and it remains only to show that the various local identifiers \emph{do} satisfy properties (i)--(iii)', thus contradicting that the program as a whole reaches a value.  Let us write $\Delta_2$ for the environment $\Delta_1 \setminus \{\someOp:\CompPure*{\topty \to \intty \onlyto \okty}\}$, and $P$ for $\lettm{\textit{tup}}{\moneyTm\,\pn{5}\,\atom{you}}{\someOp\,(\countTm\,\textit{tup})\,(\accountTm\,\textit{tup})}$, and noting that the judgement $\types \abs{xy}{x+y}: \topty \to \intty \onlyto \okty$ can be proven as above:  
\begin{mathpar}
    \infer*[left=LetL]{
        P : \okty \types \Delta_1
        \and
        \infer*[left=CompL]{
            \types \abs{xy}{x + y} : \topty \to \intty \onlyto \okty,\,\Delta_2
        }
        {
            \abs{xy}{x + y} : \CompPure*{\topty \to \intty \onlyto \okty} \types \Delta_2
        }
    }
    {
        \lettm{\someOp}{\abs{xy}{x + y}}{P} : \okty \types \Delta_2
    }
\end{mathpar}
Eventually, by this process, we can refute that the program as a whole evaluates.
% \begin{mathpar}
%     \lettm{\moneyTm}{\abs{xy}{(\atom{give},\,(x,\,y))}}{\ldots \;\lettm{\textit{tup}}{\moneyTm\,\pn{5}\,\atom{you}}{\someOp\,(\countTm\,\textit{tup})\,(\accountTm\,\textit{tup})}} : \okty \types
% \end{mathpar}

The final program in H\'ebert's sequence is one \begin{added}for which Dialyzer's automatic type inference fails\end{added}, but nevertheless will crash at runtime.  It is a small modification of the former program where, instead of two distinct functions $\countTm$ and $\accountTm$, there is a single function $\itemTm$ whose behaviour is either like $\countTm$ or like $\accountTm$ depending on a switch.
\begin{lstlisting}[xleftmargin=.2\textwidth]
      item = $\lambda x.$ match x with                        
                           $\left|\begin{array}{lcl}(\atom{count},\,(\atom{give},\,(y,\,z))) &\mapsto& y;\\(\atom{account},\,(\atom{give},\,(y,\,z))) &\mapsto& z;\end{array}\right.$

\end{lstlisting}
Since the program has essentially the same structure as the previous, we may adopt the same strategy for refuting that it can reach a value.  Here we can deduce that, for the body of the innermost let to have reduced to a value, it \emph{must not} be that all of the following are true:
\begin{enumerate}[(i)]
    \item $\someOp$ is a function that requires that its second argument evaluates to an integer
    \item and, $\itemTm$ is a function that, to return an integer, requires an input that does not evaluate to a value of shape $(\atom{account},\,(\atom{give},\,(P,\,Q)))$, with $Q$ not an integer
    \item and, $\textit{tup}$ does evaluate to a value of shape $(\atom{give},\,(P,\,Q))$ with $Q$ not an integer.
\end{enumerate}
As previously, our argument proceeds by showing that, in fact, \emph{all} of (i)--(iii) do hold, contradicting the assumption that the program did evaluate.
The approach is similar to the previous two examples, but of particular interest is the change from separate $\countTm$ and $\accountTm$ to a single function $\itemTm$\del{, since this takes the example beyond the power of Dialyzer}.  However, using our sophisticated \rlnm{MatchL} rule we can prove (ii):
\begin{mathpar}
    \types \abs{x}{\matchtm{x}{\ldots}} : \CompPure{(\atom{account},\,(\atom{give},\,(\okty,\,\CompPure{\intty} \wedge \okty)))}
    \onlyto \intty
\end{mathpar}
Note: we have to use $\CompPure{\intty} \wedge \okty$ to ensure that the type $(\atom{account},\,(\atom{give},\,(\okty,\,\CompPure{\intty} \wedge \okty)))$ \begin{added}is a subtype of $\okty$, according to the side-condition on the rule\end{added}.
To do so, we assume that the body of the function evaluates to an integer and show that the argument $x$ therefore did not match the $\atom{account}$ branch without $z$ being an integer.  Of course, we could also show that it follows that $x$ did not match the $\atom{count}$ branch without $y$ being an integer, but this is not relevant to the overall argument.  By \rlnm{CompR}, it suffices to assume that $x$ \emph{does} match the $\atom{account}$ branch without $z$ being an integer, and then use \rlnm{MatchL} to refute that the match expression evaluated to an integer.
\begin{mathpar}
    x:(\atom{account},\,(\atom{give},\,(\okty,\,\CompPure{\intty} \wedge \okty))),\,\matchtm{x}{\ldots} : \intty \types
\end{mathpar}
By taking $\theta_1 = [\CompPure{\intty} \wedge \okty/y,\,\okty/z]$ and $\theta_2 = [\okty/y,\,\CompPure{\intty} \wedge \okty/z]$, this assumption on $x$, let us write $\Delta_x$ for short, ensures that we can prove that the match expression did evaluate some branch:
\begin{mathpar}
    \Delta_x \types x:(\atom{count},\,(\atom{give},\,(\CompPure{\intty} \wedge \okty,\,\okty))) \vee (\atom{account},\,(\atom{give},\,(\okty,\,\CompPure{\intty} \wedge \okty)))
\end{mathpar}
Then it is straightforward to refute that some branch returned an integer, by showing that in each of the two branches, returning an integer is inconsistent with this assignment to the free pattern variables $y$ and $z$ (i.e. that $\CompPure{\theta_i}$ follows in each case).  For example, in the second branch:
\begin{mathpar}
    \infer*[left=Sub]
    {
        \infer*[left=Var]{ }{
            x:(\atom{account},\,(\atom{give},\,(\okty,\,\CompPure{\intty} \wedge \okty))),\,z:\intty \types y:\CompPure{\okty},\,z:\intty
        }
    }
    {
        x:(\atom{account},\,(\atom{give},\,(\okty,\,\CompPure{\intty} \wedge \okty))),\,z:\intty \types y:\CompPure{\okty},\,z:\CompPure*{\CompPure{\intty} \wedge \okty}
    }
\end{mathpar}
Having derived this type for $\itemTm$, the rest of the argument follows in the same way as the two previous examples, so that overall, we can refute that the program reduces to a value.

\begin{remark}
    \begin{added}
        Dialyzer fails to certify this example automatically because it can only infer the (Erlang) type \lstinline!('account' | 'count', {'give', 5, 'you'}) -> 'you' | 5!, in which the dependence between the atom $\atom{account}$ in the input and the atom $\atom{you}$ in the output has been lost.
        However, its analysis can be guided by extending the example with the following ``overloaded type specification'', which explicitly enumerates the dependencies between the two calls and their return values, and then it can certify the bug.
        \begin{lstlisting}[language=Erlang, belowskip=-0.8 \baselineskip, aboveskip=0.4 \baselineskip]
    -spec item('account', {'give', 5, 'you'}) -> 'you';
                     ('count', {'give', 5, 'you'}) -> 5.
  \end{lstlisting}
    \end{added}
\end{remark}

\section{The CoArrow Type}\label{sec:coto}

In the discussion of H\'ebert's first example, we used the type $\topty \to \intty \onlyto \okty$ to assert the property that a curried function requires its second argument to be an integer.  This works because the type mixes the two kinds of arrows in a specific way: it is the type of those functions that, if you give them any first argument will either diverge or return a function that requires an integer in order to return \begin{added}a value\end{added}.

However, it is not clear there is a general pattern to forming assertions of the form ``is a curried function which, to return an $A$, requires that its $n^{\text{th}}$ argument is a $B$'' using the necessity arrow.  For example, \emph{neither} of the following two types express that, to return a value, the first argument of a two-argument curried function must be an integer:
\begin{description}
  \item[$\intty \onlyto \topty \onlyto \okty$.] This is the type of functions that, if they return a function of type $\topty \onlyto \okty$, their argument must have been an integer.  Hence, this property holds of any function that does not return a function, such as $\abs{x}{\pn{1}}$.
  \item[$\intty \to \topty \onlyto \okty$.]  This is the type of functions that, if they are given an integer as the first argument, guarantee to return a function, about which we have no further guarantees (if the function returns a value, then it's argument must have had type $\topty$, which is no information at all).  An example is $\abs{xy}{\pn{1}}$.
\end{description}
% \begin{center}
%   \begin{tabular}{rp{.75\textwidth}}
%     $\intty \onlyto \topty \onlyto \okty$ & This is the type of functions that, if they return a function of type $\topty \onlyto \okty$, their argument must have been an integer.  Hence, this property holds of any function that does not return a function, such as $\abs{x}{\pn{1}}$.                                                                         \\[2mm]
%     $\intty \to \topty \onlyto \okty$     & This is the type of functions that, if they are given an integer as the first argument, guarantee to return a function, about which we have no further guarantees (if the function returns a value, then it's argument must have had type $\topty$, which is no information at all).  An example is $\abs{xy}{\pn{1}}$.
%   \end{tabular}
% \end{center}
By unfolding the abbreviations, one can see that such a property \emph{could} be expressed by the type $\CompPure{\intty} \to \topty \to \CompPure{\okty}$.  However, it is not at all clear how to use this type with the \rlnm{AppL} rule.

The problem is that the necessity arrow does not curry well.  Apart from the difficulties illustrated above, this can be seen clearly in the proof theory.  When we type a curried abstraction using the traditional arrow, the two uses of the \rlnm{Abs} rule follow directly on one another, shown below left.
\begin{mathpar}
  \infer*[left=Abs]{
    \infer*[Left=Abs]{
      x:A,\,y:B \types M : C
    }
    {
      x:A \types \abs{y}{M} : B \to C
    }
  }
  {
    \types \abs{xy}{M} : A \to B \to C
  }
  \and
  \infer*[left=Abn]{
    \abs{y}{M} : B \onlyto C \types x : A
  }
  {
    \types \abs{xy}{M} : A \onlyto B \onlyto C
  }
\end{mathpar}
In the conclusion of the rule, the typing formula is of shape $\abs{z}{P} : A_1 \to A_2$ and occurs on the right of the judgement and in the premise we have a formula of exactly the same shape.  However, this is not the case for the abstraction rule for the necessity function type \rlnm{Abn} (which was included in \cite{ramsay-walpole-popl2024} as a primitive rule, but is easy to derive in our system), shown above \begin{added}on the\end{added} right.
This rule doesn't apply to the judgement $\abs{y}{M} : B \onlyto C \types x:A$.  Indeed, we can only proceed by subsumption on the left, which is not helpful.

\begin{figure}
  \begin{mdframed}
    \begin{mathpar}
      \infer[CoAbs]{
        \Gamma,\,M : A \types x : B,\,\Delta
      }
      {
        \Gamma,\,\abs{x}{M} : B \coto A \types \Delta
      }
      \and
      \infer[CoApp]{
        \Gamma,\,M:B \coto A \types \Delta
        \and \Gamma,\,N:B \types \Delta
      }
      {
        \Gamma,\,M\,N:A \types \Delta
      }
      % \and
      % \infer*[left=CoAbs]{
      %   \infer*[Left=CoAbs]{
      %     M : C \types y:B,\,x:A
      %   }
      %   {
      %     \abs{y}{M} : B \coto C \types x:A
      %   }
      % }
      % {
      %   \abs{xy}{M} : A \coto B \coto C \types
      % }
    \end{mathpar}
  \end{mdframed}
  \caption{Derivable abstraction and application rules \emph{a} \emph{la} bi-intuitionistic logic.}\label{fig:coabs-coapp}
\end{figure}

Although the necessity arrow $A \onlyto B$ appears natural, it is not, on its own, quite the right concept through which to properly understand the property that a function requires a certain property of its argument.  In fact, a better choice is to use the \emph{complement} of the necessity arrow.

At first this may seem strange, because the complement of the necessity arrow, $\CompPure*{A \onlyto B}$, is hardly a function type at all.  This type, viewed as a set under Definition~\ref{def:type-meanings}, contains integers and pairs and so on -- even stuck terms belong to this type!  However, \emph{refuting} that a term has this type amounts to showing that the term is a function that requires an $A$ as input if it is to return a $B$.  Moreover, we can chain this type to describe the properties of curried functions in a meaningful way.  Let us call this function type the \emph{coarrow}, or \emph{coto} for short.  Let us also introduce the abbreviation: $A \coto B \coloneqq \CompPure*{A \onlyto B}$.  Then the typing rule \rlnm{CoAbs} of Figure~\ref{fig:coabs-coapp} is derivable.
Using the coarrow, we can type curried abstractions as easily as we type them with the traditional arrow type, because of the implicit negation that surrounds the type of the returned function.

This derivation also helps to clarify the meaning of a type of shape $A \coto B \coto C$.  We can \emph{refute} that a curried function $\abs{xy}{M}$ satisfies this type whenever we can show that $M$ evaluating to a $C$ implies that \emph{either} $x$ must have been an $A$ \emph{or} $y$ must have been a $B$.  In particular, by using $\botty$, we can therefore state \emph{the negation of} properties of the form desired above.  For example:
\begin{description}
  \item[$\botty \coto \intty \coto \okty$.] This is the type of normal forms that are \emph{not} functions that, to return \begin{added}a value\end{added}, require their second argument to be an integer.  It is easy to see, by unfolding the abbreviations, that it is subtype equivalent to the complement of the type $\topty \to \intty \onlyto \okty$, which we used in H\'ebert's first example.
  \item[$\intty \coto \botty \coto \okty$.] This is the type of normal forms that are \emph{not} functions that, to return \begin{added}a value\end{added}, require their first argument to be an integer.  Unfolding the abbreviations, this type is  subtype equivalent to the complement of that we suggested above: $\CompPure{\intty} \to \topty \to \CompPure{\okty}$.
\end{description}
% \begin{center}
%   \begin{tabular}{rp{.75\textwidth}}
%     $\botty \coto \intty \coto \okty$ & This is the type of normal forms that are \emph{not} functions that, to return \begin{added}a value\end{added}, require their second argument to be an integer.  It is easy to see, by unfolding the abbreviations, that it is subtype equivalent to the complement of the type $\topty \to \intty \onlyto \okty$, which we used in H\'ebert's first example. \\[2mm]
%     $\intty \coto \botty \coto \okty$ & This is the type of normal forms that are \emph{not} functions that, to return \begin{added}a value\end{added}, require their first argument to be an integer.  Unfolding the abbreviations, this type is easily seen to be subtype equivalent to the complement of that we suggested above: $\CompPure{\intty} \to \topty \to \CompPure{\okty}$.
%   \end{tabular}
% \end{center}
Refuting that a term has one of these types, therefore, verifies that it satisfies the corresponding property of functions.  Moreover, the \rlnm{CoApp} variant of \rlnm{AppL}, of Figure~\ref{fig:coabs-coapp}, is easily derivable.

Hopefully the reader can agree that the proof theory works out quite smoothly (see also the examples below), but the nature of this type is so unusual -- the `function type' $\topty \coto \botty$ doesn't contain any functions -- that one might want some independent confirmation that it is really a natural choice.  

In fact, the logical counterpart to this type is the key feature in the extension of intuitionistic logic to bi-intuitionistic logic (also called Heyting-Brouwer \cite{rauszer-stlog1974}, or subtractive logic \cite{crolard-tcs2001}).  As is well known, the usual function type $A \to B$ is related to the propositional implication, $A$ implies $B$, by the Curry-Howard Correspondence.  Propositional implication has a natural dual, which is variously called co-implication, subtraction or pseudo-difference.  Where implication can be characterised by the adjunction below left (whose forward direction corresponds to the usual \rlnm{Abs} rule), co-implication can be defined  by the below right\footnote{Some authors prefer to define coimplication rather as the converse to this relation, i.e. with the $A$ and $B$ exchanged, see the second footnote of \cite{tranchini-jal2017}.}, whose correspondence with our \rlnm{CoAbs} rule is immediately clear.
\begin{mathpar}
  A \wedge B \types C \quad\text{iff}\quad A \types B \to C \and B \coto A \types C \quad\text{iff}\quad A \types B \vee C
\end{mathpar}

\begin{example}
  The following example is a good illustration of how to use the coarrow.  Observe that we can refute that the \emph{twice} function, $\abs{fx}{f\,(f\,x)}$, can be assigned the type $(\intty \coto \intty) \coto \intty \coto \intty$.  Writing $\Delta$ for $\{x:\intty,\,f:\intty \coto \intty\}$, we have the derivation:
  \begin{mathpar}
    \infer*[left=CoAbs]{
      \infer*[Left=CoAbs]{
        \infer*[Left=CoApp]{
          \infer*[Left=Var]{ }{
            f : \intty \coto \intty \types \Delta
          }
          \and
          \infer*[left=CoApp]{
            \infer*[Left=Var]{ }{f:\intty \coto \intty \types \Delta}
            \and
            \infer*[Left=Var]{ }{x:\intty \types \Delta}
          }
          {
            f\,x : \intty \types \Delta
          }
        }
        {
          f\,(f\,x) : \intty \types \Delta
        }
      }
      {
        \abs{x}{f\,(f\,x)} : \intty \coto \intty \types f : \intty \coto \intty
      }
    }
    {
      \abs{fx}{f\,(f\,x)} : (\intty \coto \intty) \coto \intty \coto \intty \types
    }
    % \and
    % \infer*[left=CoAbs]{
    %   \infer*[Left=OpE]{
    %     \infer*[Left=Var]{ }{
    %       x:\intty,\,\pn{2}:\intty \types x:\intty
    %     }
    %   }
    %   {
    %     x + \pn{2} : \intty \types x : \intty
    %   }
    % }
    % {
    %   \abs{x}{x + \pn{2}} : \intty \coto \intty \types
    % }
  \end{mathpar}
  From the conclusion of this refutation, we are assured that, if the twice function returns an integer, it must be that \emph{either}:
  \begin{enumerate}[(i)]
    \item its second argument was an integer
    \item \emph{or}, its first argument was not a function that required an integer in order to return one
  \end{enumerate}
  Another way to read this implication is that, if the twice function returns an integer and its second argument was \emph{not} an integer, then it must be that its first argument didn't require an integer in order to return one.

  Moreover, it is straightforward to show that $\abs{x}{x + \pn{2}}$ is a function that, to return an integer, requires an integer as input.  That is, we can \emph{refute} that it has type $\intty \coto \intty$.  Therefore, we can prove that $(\abs{fx}{f\,(f\,x)})\,(\abs{x}{x + \pn{2}})$ really requires an integer in order to evaluate to an integer.  That is, we have the following refutation:
  \begin{mathpar}
    \infer*[left=CoApp]{
      \abs{fx}{f\,(f\,x)} : (\intty \coto \intty) \coto \intty \coto \intty \types
      \and
      \abs{x}{x + \pn{2}} : \intty \coto \intty \types
    }
    {
      (\abs{fx}{f\,(f\,x)})\,(\abs{x}{x + \pn{2}}) : \intty \coto \intty \types
    }
  \end{mathpar}
\end{example}

% The identification of the coarrow as the natural concept around which to base refutations involving functions opens up an interesting possibility to design a system which only allows for refutation.  Such a system would manipulate judgements that are of shape $M:A \types x_1:B_1,\,\ldots,\,x_n:B_n$, i.e. dual to the traditional shape of type system judgements.

\section{Metatheory of Subtyping}\label{sec:subtyping-props}

\newcommand{\subtypeA}{\leq}
\begin{figure*}
  \begin{mdframed}
    \begin{mathpar}
      \infer[Refl]{ }{A \subtypeA A}
      \and
      \infer[Top]{ }{A \subtypeA \topty}
      \and
      \infer[OkR]{A \subtypeA B}{A \subtypeA \okty}\; B \in \mathcal{B}
      \and
      \infer[Atom]{ }{\atom{a} \subtypeA \atomty}
      \and
      \infer[CompRAtom1]{ }{\atom{b} \subtypeA \CompPure{\atom{a}}}\!\!\!\!\!\!a \neq b
      \and
      \infer[CompLAtoms]{ }{\CompPure{\atomty} \subtypeA \CompPure{\atom{a}}}
      \and
      \infer[CompRAtom2]{A \del{\subtype} \begin{added}\subtypeA\end{added} B}{A \del{\subtype} \begin{added}\subtypeA\end{added} \CompPure{\atom{a}}}\!\!\!\!\!\!\!\!B \in \{\intty,\pairty,\funty\}
      \and
      \infer[CompRInt]{A \subtypeA B}{A \subtypeA \CompPure{\intty}}\;B \in \{\pairty,\funty,\atomty\}
      \and
      \infer[CompRPair]{A \subtypeA B}{A \subtypeA \CompPure{(B,C)}}\;B \in \{\intty,\funty,\atomty\}
      \and
      \infer[CompRArr]{A \subtypeA B}{A \subtypeA \CompPure*{D \to E}}\;B \in \{\intty,\pairty,\atomty\}
      \and
      \infer[CompRAtoms]{A \subtypeA B}{A \subtypeA \CompPure{\atomty}}\;B \in \{\intty,\pairty,\funty\}
      \and
      \infer[CompRPairL]{A \subtypeA \CompPure{C}}{(A,B) \subtypeA \CompPure{(C,D)}}
      \and
      \infer[CompRPairR]{A \subtypeA \okty \and B \subtypeA \CompPure{D}}{(A,B) \subtypeA \CompPure{(C,D)}}
      \and
      \infer[CompRUn]{A \subtypeA \CompPure{A_i}\;(\forall i \in [1,k])}{A \subtypeA \CompPure*{\textstyle\bigvee_{i=1}^k A_i}}
      \and
      \infer[CompRC]{A \subtypeA B}{A \subtypeA \CompPure*{\CompPure{B}}}
      \and
      \infer[Pair]{A \subtypeA A' \and B \subtypeA B'}{(A,B) \subtypeA (A',B')}
      \and
      \infer[Fun]{A' \subtypeA A \and B \subtypeA B'}{A \to B \subtypeA A' \to B'}
      \and
      \infer[UnionR]{A \subtypeA B_j}{A \subtypeA \textstyle\bigvee_{i=1}^k B_i}\;j \in [1,k]
      \and
      \infer[UnionL]{B_i \subtypeA A\;(\forall i \in [1,k])}{\textstyle\bigvee_{i=1}^k B_i \subtypeA A}
      \and
      \infer[CompLTop]{ }{\CompPure{\topty} \subtypeA B}
      \and
      \infer[CompLPair]{\CompPure{A} \subtypeA \CompPure{A'} \and \CompPure{B} \subtypeA \CompPure{B'}}{\CompPure{(A,B)} \subtypeA \CompPure{(A',B')}}
      \and
      \infer[CompLArr]{\CompPure{A'} \subtypeA \CompPure{A} \and \CompPure{B} \subtypeA \CompPure{B'}}{\CompPure*{A \to B} \subtypeA \CompPure*{A' \to B'}}
      \and
      \infer[CompLOk]{\CompPure{B} \subtypeA A}{\CompPure{\okty} \subtypeA A}\; B \in \mathcal{B}
      \and
      \infer[CompLUn]{\CompPure{B_j} \subtypeA A}{\CompPure*{\textstyle\bigvee_{i=1}^k B_i} \subtypeA A}\;j \in [1,k]
      \and
      \infer[CompLC]{A \subtypeA B}{\CompPure*{\CompPure{A}} \subtypeA B}
    \end{mathpar}
  \end{mdframed}
  \caption{Alternative rules for subtyping.}\label{fig:sub-alt}
\end{figure*}

The rules of the subtype relation are natural from our understanding of types as certain sets of normal forms, so it is not surprising that \begin{added}these sets form a model of the theory, as revealed by \cref{thm:subtyping-soundness}\end{added}.  The contrapositive of the theorem allows us to use the set-theoretic model to exclude that certain subtype inequalities are derivable.  For example, we cannot have $A \to B \subtype \intty$ for any types $A$ and $B$, nor can we derive $A \to B \subtype \CompPure{(B \to C)}$ for any $B$ and $C$, since the set on the left always contains $\abs{x}{\divtm}$ but the set on the right does not.

\begin{theorem}[Soundness of Subtyping]\label{thm:subtyping-soundness}
  If $A \subtype B$ then $\mng{A} \subseteq \mng{B}$.
\end{theorem}

However the axiomatisation itself is not straightforward to analyse because of the unrestricted transitivity rule and the general rules regarding complementation.  Yet it is essential to understand certain invariants of this relation to properly argue the soundness of the type system.  For example, for soundness, we require that whenever $A \to B \subtype C \to D$ is derivable, then it must be that already $C \subtype A$ and $B \subtype D$.  If this inequality is derived using the \rlnm{Fun} rule, then all is well, but as usual the unrestricted transitivity rule allows for other more convoluted proofs.  For example, we can prove $\intty \to \intty \subtype \intty \to \okty$ using transitivity with the mid-point $\CompPure*{\textstyle\bigvee_{i=1}^2 \CompPure{(\CompPure{\CompPure{\intty}} \to (\intty \vee \atom{a} \vee \atom{b_i}))}}$.  Clearly, there is no need to invoke transitivity at all in the proof of this inequality, but even supposing we decide to invoke transitivity, this chosen mid-point is excessively complicated -- through a combination of the union and complementation rules it can be seen equivalent to the simpler $\intty \to (\intty \vee \atom{a})$.

Thus we are led to present an alternative characterisation of the subtyping relation in which the space of possible proofs is significantly more constrained, although, as we shall see, the provable inequalities are the same.  The rules of Figure~\ref{fig:sub-alt} are obtained from the the rules of the original subtyping relation by `baking in' transitivity, and enforcing a more principled use of the \rlnm{CompL} and \rlnm{CompR} rules.

\begin{definition}[Alternative Subtyping]
  The alternative subtyping relation $A \subtypeA B$ is defined by the rules of Figure~\ref{fig:sub-alt}.  Here $\mathcal{B}$ abbreviates the set $\{\intty,\,\pairvalty,\,\funty,\,\atomty\}$.
\end{definition}

The motivation for restricting the complementation rules is as follows.  When constructing a proof of $A \subtype \CompPure{B}$, bottom-up, one may sometimes find a choice between using disjointness \rlnm{Disj}, \rlnm{PairC} or the \rlnm{CompR} rule.  If the \rlnm{CompR} rule is chosen, then there are many choices with how to conclude the subproof of $B \subtype \CompPure{A}$.  For example, one may use disjointness with $A$, or some rule that exploits the structure of $B$, or even using \begin{added}\rlnm{CompR}\end{added} again to uselessly return to the original conclusion.  In our alternative system of Figure~\ref{fig:sub-alt}, we split \rlnm{CompR} into a large number of cases, and we restrict that proofs should follow a particular strategy based on the syntax of $B$.  In particular, in most cases when the form of the judgement $A \subtype \CompPure{B}$ would call for a use of $\rlnm{CompR}$ in the original system, the alternative system will insist that this should be followed by some rule exploiting the structure of $B$, followed by a use of \rlnm{CompL} to `undo' the rearrangement.  For example, the rule \rlnm{CompRUn} of the alternative system corresponds to the following derivation in the original:

\begin{mathpar}
  \infer*[left=CompR]{
    \infer*[Left=UnionL,vcenter]{
      \infer*[left=CompL]{
        A \begin{added}\subtype\end{added} \CompPure{B_1}
      }
      {
        B_1 \begin{added}\subtype\end{added} \CompPure{A}
      }
      \and
      \cdots
      \and
      \infer*[left=CompL]{
        A \begin{added}\subtype\end{added} \CompPure{B_k}
      }
      {
        B_k \begin{added}\subtype\end{added} \CompPure{A}
      }
    }
    {
      \textstyle\bigvee_{i=1}^k B_i \begin{added}\subtype\end{added} \CompPure{A}
    }
  }
  {
    A \subtype \CompPure*{\textstyle\bigvee_{i=1}^k B_i}
  }
\end{mathpar}

The alternative characterisation of subtyping makes it much easier to prove the invariants we need in order to establish the soundness of the type system.  For example, a brief inspection of the rules reveals that \rlnm{Fun} is now the \emph{only} way to derive $A \to B \subtypeA C \to D$.  In particular, there is no transitivity rule and the complementation rules are restricted by the syntax of the complemented type.  However, it is straightforward to show that the original rules are admissible, and then equivalence with the original system follows by induction.

% \begin{lemma}
%   The following are both true:
%   \begin{enumerate}[(i)]
%     \item There is some $C$ such that $C \subtypeA A$ and $\CompPure{C} \subtypeA A$, iff $\CompPure{A} \subtypeA A$.
%     \item There is some $C$ such that $A \subtypeA C$ and $A \subtype \CompPure{C}$, iff $A \subtype \CompPure{A}$.
%   \end{enumerate}
% \end{lemma}
% \begin{proof}
%   For (i), first suppose $C \subtypeA A$ and $\CompPure{C} \subtypeA A$.  Then, by the order-reversing property of complement (\cref{thm:ortholattice}), $\CompPure{A} \subtype \CompPure{C}$.  Therefore, by transitivity, $\CompPure{A} \subtype A$.  Next, suppose $\CompPure{A} \subtype A$.  Then, the witness can be taken as $C := \CompPure{A}$ since $A \subtypeA A$ by \rlnm{Refl}.  Item (ii) is symmetrical.
% \end{proof}

% Therefore, the rules \rlnm{CompTop} and \rlnm{CompTop'}, and the rules \rlnm{CompBot} and \rlnm{CompBot'} are equivalent and we shall freely use both variants, according to  convenience.

\begin{theorem}[Equivalence of Subtyping Relations]\label{thm:subtyping-equivalence}
  $A \subtype B$ iff $A \subtypeA B$
\end{theorem}

\newcommand{\subExp}{\mathsf{SubExp}}
\newcommand{\univ}[1]{\mathcal{U}_{#1}}

As well as making it easier to establish certain invariants, the alternative characterisation gives us a straightforward approach to argue its decidability.

\begin{theorem}[Decidability]\label{thm:subtyping-decidability}
  The problem, given $A$ and $B$ determine derivability of $A \subtype B$, is decidable.
\end{theorem}
\begin{proof}
  Let $A$ and $B$ be types given as input.  By the equivalence theorem, it suffices to decide $A \subtypeA B$.  Let $\subExp(A)$ be the set of all subexpressions\footnote{As usual, we consider $A$ to be a subexpression of itself.} of the type $A$ and $\subExp(B)$ be the set of all subexpressions of the type $B$.  Then consider the following set of types:
  \[
    \mathcal{A}_{A,B} \coloneqq \subExp(A) \cup \subExp(B) \cup \{\topty,\,\intty,\,\pairty,\,\pairvalty,\,\funty,\,\atomty,\,\okty\}\\
  \]
  This set consists of all the finitely many types that occur as subexpressions of the input types $A$ and $B$ as well as certain fixed types such as $\topty$, $\intty$ and so on.   Then define the following finite set of types:
  \[
    \univ{A,B} \coloneqq \mathcal{A}_{A,B} \cup \{D^c \mid D \in \mathcal{A}_{A,B}\}
  \]
  One can see, by inspection of the alternative subtyping rules in Figure~\ref{fig:sub-alt}, that whenever a rule instance is concluded by some inequality $C \subtypeA D$, and $C$ and $D$ are in $\univ{A,B}$, then every type occurring in the premises of the rule instance are also contained in $\univ{A,B}$.  Hence, in any proof of $A \subtypeA B$ it must be that every subproof is concluded by an inequality of shape $C \subtypeA D$, where both $C$ and $D$ are in $\univ{A,B}$, of which there are only finitely many.  Consequently, a bottom-up proof search that avoids cycles will decide $A \subtypeA B$.
\end{proof}

\section{Metatheory of Type System}\label{sec:type-system-metatheory}

\begin{figure}
    \begin{mdframed}
        \begin{mathpar}
            \infer[Comp]{
                \Gamma,\,\Constituent{x:A} \types \Delta
            }
            {
                \Gamma \types \Principal{x:\CompPure{A}},\,\Delta
            }
            \and
            \infer[Match]{
                \Gamma \types \Constituent{M : \textstyle\bigvee_{i=1}^k p_i\theta_i},\,\Delta\\
                \Gamma,\,\theta_i \types \Constituent{M : \CompPure*{p_i\theta_i}},\,\Constituent{N_i : A},\,\Delta \;\; (\forall i)
            }
            {
                \Gamma \types \Principal{\matchtm{M}{\mid_{i=1}^k p_i \mapsto N_i} : A},\,\Delta
            }
            \\
            \infer[OpE]{
                \Gamma \types \Constituent{M:\CompPure{\intty}},\,\Constituent{N:\CompPure{\intty}},\,\Delta
            }
            {
                \Gamma \types \Principal{\relOpPrim{M}{N} : \CompPure{\okty}},\,\Delta
            }
            \and
            \infer[RelOpE]{
                \Gamma \types \Constituent{M:\CompPure{\intty}},\,\Constituent{N:\CompPure{\intty}},\,\Delta
            }
            {
                \Gamma \types \Principal{\binOpPrim{M}{N} : \CompPure{\okty}},\,\Delta
            }
            \and
            \infer[PairE]{
                \Gamma \types \Constituent{M:\CompPure{\okty}},\,\Constituent{N:\CompPure{\okty}},\,\Delta
            }
            {
                \Gamma \types \Principal{(M,N) : \CompPure{\okty}},\,\Delta
            }
            \\
            \infer[AppE]{
                \Gamma \types \Constituent{M:\CompPure{\funty}},\,\Delta
            }
            {
                \Gamma \types \,\Principal{M\,N : \CompPure{\okty}},\,\Delta
            }
            \and
            \infer[MatchE] {
            \Gamma \types \Constituent{M:\CompPure*{\textstyle\bigvee_{i=1}^k p_i\thetaok{p_i}}},\,\Delta
            }
            {
            \Gamma \types \Principal{\matchtm{M}{|_{i=1}^k p_i \mapsto N_i} : {\CompPure{\okty}}},\,\Delta
            }
        \end{mathpar}
    \end{mdframed}
    \caption{One-sided typing rules (differences).}\label{fig:one-sided-delta}
\end{figure}

For developing the metatheory of the type system, we first simplify the situation.  Two-sided type systems have arbitrary typing formulas on the left- and right-hand sides of the judgement.  However, we show that it is possible to reformulate the type system so that\del{ only} arbitrary typings are needed \begin{added}only\end{added} on the right (though we may still need any number of them).

\begin{definition}[One-Sided System]
    Judgements of the one-sided system have shape $\Gamma \types \Delta$ in which every typing in $\Gamma$ is a variable typing, of shape $x:A$.  The rules are those of the two-sided system but with \rlnm{SubL} removed, and \rlnm{CompL}, \rlnm{Match}, \rlnm{OpE}, \rlnm{RelOpE}, \rlnm{PairR}, \rlnm{AppE} and \rlnm{MatchE} replaced by the rules of Figure~\ref{fig:one-sided-delta}.
\end{definition}

The idea is as follows.  It is easy to see that, the one-sided and two-sided versions of \rlnm{Match}, \rlnm{OpE}, \rlnm{RelOpE}, \rlnm{PairE}, \rlnm{AppE} and \rlnm{MatchE} are inter-derivable in the two-sided system by bracketing rule uses with \rlnm{CompL} and \rlnm{CompR}.  Therefore, one can suppose that the original, two-sided system uses \begin{added}rules in \cref{fig:one-sided-delta}\end{added} instead.  In that case, a short inspection of the remaining two-sided rules shows that the only way to conclude a typing on the left-hand side is by \rlnm{SubL} or \rlnm{CompL}, and thus the latter must eventually be used.  This implies that to refute $M:A$, we must have already verified $M:B$ on the right, for some $B$ such that $A \subtype \CompPure{B}$.  Therefore, every refutation $M:A$ corresponds to a verification of some $M:B$, for which we did not need any left rules.  Formally, the proof is by induction.

\newcommand{\typesA}{\types_2}
\newcommand{\typesB}{\types_1}

\begin{theorem}[Equivalence]\label{thm:one-two-equivalence}
    Let us write $\Gamma \typesA \Delta$ for provability in the two-sided system and $\Gamma \typesB \Delta$ for provability in the one-sided system.  Then both of the following:
    \begin{quote}
        (i) \ If $\Gamma \typesA \Delta$ then $\typesB \CompPure{\Gamma},\,\Delta$.  \qquad
        (ii) \ If $\Gamma \typesB \Delta$ then $\Gamma \typesA \Delta$.
    \end{quote}
\end{theorem}

Note: a judgement like $x:A,\,x:\CompPure{A} \types$, although provable in the two-sided system is \emph{not} provable in the one-sided system.  On the other hand, the theorem shows that the ``equivalent'' judgement $\types x:A,\,x:\CompPure{A}$ \emph{is} provable in the one-sided system.  The asymmetry stems from the fact that the one-sided rules have \rlnm{Comp} that may faithfully simulate \rlnm{CompR} on variable typings, but no rule to faithfully simulate \rlnm{CompL} on variable typings.

This is a useful simplification, but there still remains the problem of multiple conclusions on the right.  This is a significant barrier to a standard progress and preservation style argument for the soundness of the type system.  Consider progress for a judgement of shape $\Gamma \types M:\okty,\,\Delta$.  When $\Gamma$ and $\Delta$ are arbitrary, and not only variable typings, it is not true that $M$ in normal form will be a value.  This is the case in, for example, the judgement:
\[
    \types \pn{2}\,(\abs{z}{z}) : \okty,\,\abs{x}{x} : \intty \to \intty
\]
This judgement expresses the property: \emph{either} $\pn{2}\,(\abs{z}{z})$ diverges or yields an integer value, \emph{or} $\abs{x}{x}$ is an integer function.  Since the second conclusion is true, the truth of the judgement as a whole simply does not depend on the\begin{deleted}\del{ meaningless}\end{deleted} drivel in the first conclusion, and so we cannot expect to be able to rely on anything about it.

\begin{remark}\label{rmk:soundness-comparison-with-itpde}
    In fact, this highlights a deficiency in the original work.  Due to this difficulty, progress and preservation in \cite{ramsay-walpole-popl2024} was proven only for judgements satisfying the \emph{single-subject constraint}, that is, in which at most one term in the judgement is not a variable.  Thus, they establish soundness only for fragments of their two-sided type systems.
\end{remark}

Here, rule \rlnm{App} is applied at the root of the derivation tree to obtain $\pn{2}\,(\abs{z}{z})$. However, the validity of the judgement does not rely on the validity of $\pn{2}\,(\abs{z}{z}):\okty$, but the validity of $\abs{x}{x}$.  For example, we can have a smaller proof where $\pn{2}\,(\abs{z}{z})$ in untouched.

For preservation, although it is true that $Q \ped Q'$ and $\Gamma \types Q:A,\,\Delta$ implies $\Gamma \types Q':A,\,\Delta$ holds unconditionally, the proof is not straightforward.   In a typical preservation proof, we take the typing derivation for the term \textit{to be} reduced ($Q$), and adjust it to a typing derivation for the term \textit{after} reduction ($Q'$). For a single-succedent system, a derivation that assigns a type to $Q$ is always decided by $Q$ itself. When the term reduces, we can construct a derivation of $Q'$ by using the subderivation that types the (perhaps reduced) subterm of $Q$. However, in a multi-succedent system, the derivation of $\Gamma \types Q:A, \Delta$ might be completely independent from the term $Q$. Consider the derivation:
\begin{mathpar}
    \infer*[left=App]{
        \infer*[Left=Abs]{
            \infer*[Left=Var]{ }{
                x:\intty \types \pn{2}:\okty \to \okty,\,x:\intty
            }
        }
        {
            \types \pn{2}:\okty \to \okty,\,\abs{x}{x} : \intty \to \intty
        }
        \infer*[left=Abs]{
            \infer*[Left=Var]{ }{
                x:\intty \types x:\intty,\,\pn{3} : \okty
            }
        }
        {
            \types \pn{3} : \okty,\,\abs{x}{x} : \intty \to \intty
        }
    }
    {
        \types \pn{2}\ \pn{3} : \okty,\,\abs{x}{x} : \intty \to \intty
    }
    % \begin{deleted}
    % \and
    % \infer*[Left=Abs]{
    %     \infer*[Left=Var]{ }{
    %         x:\intty \types \pn{2}\ \pn{3}:\okty,\,x:\intty
    %     }
    % }
    % {
    %     \types \pn{2}\ \pn{3} : \okty,\,\abs{x}{x} : \intty \to \intty
    % }
    % \end{deleted}
\end{mathpar}
In such cases, reconstructing a derivation for $Q'$ becomes problematic because the subderivations might bear no relation to $Q$. To address this problem, we need a way to identify which conclusion formulas are \emph{relevant} to the proof.

However, this task is far from straightforward. Vertically, the \textit{principal} (or \textit{focused}) term can shift arbitrarily among conclusions at each inference step, introducing redundancy at varying levels. Horizontally, the relevance can interfere among terms: a typing assignment might only be meaningful when another typing assignment in the same conclusion is meaningful. Moreover, a typing judgement may also have many different derivations, each with different relevant formulas in their conclusion, for example when, in $\types M:A,\,N:B$, both $\types M:A$ and $\types N:B$ are independently provable. This reflects a deeper issue: the typing system lacks an internal mechanism to track the semantic justification (i.e., the root cause) of a typing judgment.

To bridge the gap between derivations and their semantics, we introduce \begin{added}an additional precondition $\Gamma \nvdash \Delta$, meaning \emph{$\Gamma \types \Delta$ is not derivable},\end{added} for the preservation proof. That is, we show $\Gamma \types Q:A,\,\Delta$ implies $\Gamma \types Q':A,\,\Delta$ whenever $\Gamma \nvdash \Delta$. This allows us to prove preservation via standard induction on typing rules just as in a single-succedent (traditional) type system, without delving into the complexities of tracking stepwise dependencies across all possible derivations. Finally, we demonstrate that this assumption can be eliminated, restoring the standard preservation theorem without additional premises. This approach provides a purely syntactic method to handle irrelevance in multi-succedent systems, bypassing the need for explicit relevance constraints.

% \begin{mathpar}
%     \infer*[left=App]{
%         \infer*[Left=Abs]{
%             \infer*[Left=Var]{ }{
%                 x:\intty \types \pn{2}:\okty \to \okty,\,x:\intty
%             }
%         }
%         {
%             \types \pn{2}:\okty \to \okty,\,\abs{x}{x} : \intty \to \intty
%         }
%         \infer*[left=Abs]{
%             \infer*[Left=Var]{ }{
%                 x:\intty \types x:\intty,\,(\abs{z}{z}) : \okty
%             }
%         }
%         {
%             \types (\abs{z}{z}) : \okty,\,\abs{x}{x} : \intty \to \intty
%         }
%     }
%     {
%         \types \pn{2}\,(\abs{z}{z}) : \okty,\,\abs{x}{x} : \intty \to \intty
%     }
%     \and
%     \infer*[Left=Abs]{
%         \infer*[Left=Var]{ }{
%             x:\intty,\,\pn{2}:\intty \types x:\intty
%         }
%     }
%     {
%         \types (\abs{z}{z}) : \okty,\,\abs{x}{x} : \intty \to \intty
%     }

% \end{mathpar}
% \begin{mathpar}
%     \infer*[left=App]{
%         \infer*[Left=Abs]{
%             \infer*[Left=Var]{ }{
%                 x:\intty \types x:\intty
%             }
%         }
%         {
%             \types \abs{x}{x} : \intty \to \intty
%         }
%         \infer*[]{
%             \cdots
%         }
%         {
%             \types V : \intty
%         }
%     }
%     {
%         \types \pn{2}\,(\abs{z}{z}) : \okty,\,\abs{x}{x} : \intty \to \intty
%     }
%     \and
%     \infer*[]{
%         \cdots
%     }
%     {
%         \types V : \intty
%     }\types (\abs{z}{z}) : \okty,\,\abs{x}{x} : \intty \to \intty

\begin{theorem}[Progress and Preservation]\label{thm:progress-preservation}
    Suppose all terms in the following are closed:\\
    {\centering
    \begin{tabular}{rp{.9\textwidth}}
        {Preservation:} & Suppose $Q \ped Q'$. If $\Gamma \types Q:A,\, \Delta$, then $\Gamma \types Q':A,\, \Delta$.                 \\
        {Progress:}     & If $\ \types \Delta$, then there is some formula $M:A \in \Delta$ such that, if $M$ is a normal form, then:
        \begin{itemize}
            \item $A \subtype \okty$ implies $M$ is a value
            \item and, $A \subtype \CompPure{\okty}$ implies $M$ is stuck.
        \end{itemize}
    \end{tabular}}
\end{theorem}

\noindent
\begin{added}A corollary of this theorem is that $M:\okty \types$ implies that $M$ either diverges or gets stuck, since $\types M :\CompPure{\okty}$ follows from this judgement by \rlnm{CompR}.\end{added}

Two properties of our system distinguish its progress theorem from conventional formulations. The first is the multi-succedent nature. Unlike preservation, which applies uniformly to all conclusion terms, progress only guarantees that some term in the conclusion satisfies the progress statement. This reflects a reliance on the underlying semantics that justify why the statement holds. The second is its ability to type normal forms in general.  In standard systems, progress rules out non-value normal forms (i.e., stuck terms) by appealing to canonical forms of values, ensuring all non-values are reducible. By contrast, our system does not exclude non-value normal forms. Instead, our progress theorem adopts a different perspective: a normal form of the type $\CompPure{\okty}$ is guaranteed to go wrong, and a normal form a value type is guranteed to not go wrong. This motivated the revealing of another property of our system\del{, which we refer as cut admissibility of normal forms}: any normal-form cannot be given both a type and its complement.

\begin{theorem}[Normal Form Cut Admissibility]
    If $\Gamma \types U :A,\,\Delta$ and $\Gamma \types U:\CompPure{A},\,\Delta$, then $\Gamma \types \Delta$
\end{theorem}

From progress and preservation we obtain the soundness of the type system, that is, if $\types M : \okty$ then $M$ cannot go wrong -- it either diverges or reduces to a value, and if $\types M : \CompPure{\okty}$, then $M$ cannot evaluate -- it either diverges or goes wrong.  As is typical, the system is not \emph{complete}, that is, \begin{added}there exist closed terms $M$ for which neither $\types M:\okty$, nor $\types M : \CompPure{\okty}$.  An example is $\mathsf{if}\ \pn{2} < \pn{3}\ \mathsf{then}\ \pn{1}\ \mathsf{else}\ (\pn{1}\ \pn{1})$, for which deducing a type of $\okty$ would need more precise information about the behaviour of the relation $<$ in the guard than is available in the system\end{added}.

However, a remarkable fact about our system is that it \emph{is} complete for $M$ in normal form.  Given that the type system is designed to allow us to prove properties $A$ of the reduction behaviour of terms $M$, expressed as $M:A$, it seems quite basic that we should want that the system is, at the very least, capable of proving properties of terms in normal form, whose reduction behaviour is completely trivial.  However, this basic property is not satisfied by most traditional type systems, even for values!  A function value whose body happens to be a stuck term, e.g. $\abs{x}{\pn{2}\,x}$, is not usually typable, since the standard \rlnm{Abs} rule would require that we can type the stuck term $\pn{2}\,x$.

\begin{theorem}[Complete Classification of Normal Forms]\label{thm:complete-classification}
    For all normal forms $U$, both of the following are true:
    \begin{quote}
        (i) \ If $U$ is a value then $\ \types U : \okty$. \qquad (ii) \ If $U$ is stuck then $\ \types U : \CompPure{\okty}$
    \end{quote}
\end{theorem}

\begin{added}
    For values, this is perhaps not so interesting, because we are very familiar with the typing of values. For stuck terms, it is more illuminating, because here it is not so clear \emph{a priori} that we have the ``right'' set of typing and subtyping rules.
    Consider a stuck term of shape $\matchtm{V}{\mid_{i=1}^k p_i \mapsto N_i}$, which is stuck by virtue of the fact that value $V$ doesn't match any of the patterns $p_i$.  Completeness requires that we can certify that this term is stuck by a derivation.  To do this, say in the one-sided system (which is equivalent to the two-sided system), there is essentially no choice but to use \rlnm{MatchE}, and hence it must be shown that: from $\forall i.\ \forall \sigma.\ V \neq p_i\sigma$ it follows that $\types V : \CompPure{(\textstyle\bigvee_{i=1}^k p_i[\okty/x \mid x \in \fv(p_i)])}$.  This latter judgement is the single premise of (MatchE). In other words, Theorem 7.6 reveals, amongst other things, that pattern (non-)matching of $V$, i.e. the condition $\forall i.\ \forall \sigma.\ V \neq p_i\sigma$ can be captured precisely by the type system.
\end{added}

% \begin{proof}
%     \input{proofs/preservation/thm-preservation.tex}
% \end{proof}

%\input{progress.tex}
% \begin{lemma}
%     Suppose $\Gamma \types p\sigma\tyfit\theta,\,\Delta$$\Gamma \types \sigma\tyfit\theta,\,\Delta$.
%     \begin{enumerate}
%         \item If $\Gamma \types p\sigma \tyfit p\theta,\,\Delta$
%         \item If $\forall x \in \dom(\theta). \theta(x) \tyneq \botty$, then $p\theta \tyeq \botty$.
%     \end{enumerate}

% \end{lemma}
% \begin{theorem}[Preservation]
%     If $\Gamma \types P:A,\, \Delta$ and $P \ped Q$, then $\Gamma \types Q:A,\, \Delta$.
% \end{theorem}
\section{Conclusion and Related Work}\label{sec:related}

We introduced a new, two-sided type system that overcomes some significant weaknesses of the systems in \cite{ramsay-walpole-popl2024}.  The system supports a wide range of refutation principles, including a type-theoretic of the coimplication from bi-intuitionistic logic.  These were put to use certifying that a number of Erlang-like programs go wrong.  The system is sound and, in a sense, complete for normal forms.

\paragraph{Type inference} We did not consider type inference, but we can observe that, barring the \rlnm{Comp} and \rlnm{SubR} rules, in all of the rules of the one-sided system the terms in the premises are strictly smaller than those in the conclusion.  Hence, the system naturally suggests a standard approach in which subsumption is first absorbed into the other rules, and then a notion of canonical derivation(s) defined, in which fresh type variables are introduced at synthesis points and constrained according to subtyping.  However, since the rules are not syntax directed -- one may freely choose between typings on the right -- there is some challenge to make such a scheme work efficiently.

\subsection{Related Work}

\paragraph{Two-sided type systems} Our work is based on two-sided type systems introduced in \cite{ramsay-walpole-popl2024}.  Compared to the systems in that paper, our system overcomes the three weaknesses listed in Remarks~\ref{rmk:comp-comparison-with-itpde}, \ref{rmk:left-rules-comparison-with-itpde} and \ref{rmk:soundness-comparison-with-itpde}. Thanks to the notion of relevance and the cut theorem, we were able to prove soundness for all judgements, and not only those that satisfy the single-subject constraint. Our system contains both \rlnm{CompL} and \rlnm{CompR} and yet still enjoys both \emph{well-typed programs cannot go wrong} and \emph{ill-typed programs cannot evaluate}.  Our system is made from a small number of well-motivated rules, with the key left rules of \cite{ramsay-walpole-popl2024} all derivable.
Moreover, our \rlnm{Match} rule exploits the two-sidedness of the system and the general soundness result.  However, we don't consider type inference.

\begin{added}
    By encoding atoms as unary constructors, the constrained type system of \cite[Section 6]{ramsay-walpole-popl2024} can handle the first two of Héberts examples, but not the third: when typing an expression of shape $\matchtm{x}{\mid_{i=1}^k p_i \mapsto M_i}$, a type $A$ must be given that can be refuted for both of the branches simultaneously, since there is no capability to refute that a variable $x$ in the scrutinee can match any particular pattern.
    Their system could be extended with such a capability by adding e.g. a notion of type disjointness in the subtype system. This way, assuming a type $x:A$ would be sufficient to refute an assignment $x:B \types$ whenever $A$ and $B$ disjoint, which could be used to ignore a case in the match and thus derive a typing expressing the dependency of just a single branch. Currently, the system of \cite[Section 6]{ramsay-walpole-popl2024} only allows reasoning about type disjointness for concrete values, but not for variables. With this extension, their system would be able to verify the third example.
\end{added}

\paragraph{Typing for incorrectness.}
One half of our system, and the part we have concentrated on in this paper, has as its purpose to reason about incorrectness by refuting type assignments.  Probably the best known type system for reasoning about incorrectness is Lindahl and Sagonas' \emph{success types}, which have been popularised in the Erlang community through the \emph{Dialyzer} tool \cite{lindahl-sagonas-ppdp2006}.
As argued by Sagonas in his Erlang'21 keynote, the Dialyzer tool has been very successful \cite{sagonas-erlang2021}.  It is part of the standard Erlang platform OTP \cite{dialyzer}, and is widely used by the community.  Yet, outside of \begin{added}this community (and those of Erlang's sibling languages)\end{added}, type systems for finding bugs, that is for reasoning about incorrectness, are largely unknown.  One may wonder why the single most popular developer tool for Erlang programmers \cite{nagy-et-al-erlang08} has not \begin{deleted}\del{been}\end{deleted} been replicated elsewhere, for example through success typing for other dynamically typed programming languages?  Or why there is no parallel to the substantial and varied research that we see into new and more powerful type systems for correctness?  One reason could be that there lacks a unified foundation for typing for (in)correctness, and we hope that the present work can contribute there.

Alternative \begin{added}type systems\end{added} for incorrectness, include work by \citet{jakob-thiemann-nasa2015}, who attempt to characterise success types as a kind of refutation in a system in which one constructs a logical formula to describe inputs that guarantee to crash a function, and the \emph{coverage types} of \cite{zhou-et-al-pldi2023}, which embed ``must-style'', underapproximate reasoning designed for checking the coverage of generators in property-based testing.

\begin{added}
    There are also works in which types are used as a language for properties, but the verification that a program satisfies such a property is based on the semantics directly, rather than via a proof (type) system.  For example, \citet{carrott-et-al-ecoop25} and \citet{qian-et-al-oopsla24}, both concern type safety bugs. However, rather than deriving a proof in a type system, they instead reduce the problem to deciding a reachability problem concerning a transformation of the program. Both works show the effectiveness of solving their reachability problem by symbolic execution.
\end{added}

\begin{added}
    One application of our system is to provide more precise type errors.  \citet{webbers-et-al-oopsla24} improve the type error messages from refinement type checking and inference. For each non-provable judgment, an explicit ``refutation'' is produced to explain why the judgement is not provable, but it does not necessarily imply something about the behaviour of the subject (e.g. that it will crash).
\end{added}

\paragraph{Logics for incorrectness.}
Incorrectness logic of \citet{ohearn-popl2019} has sparked a new interest in systems for reasoning about programs that go wrong.  Subsequent work has extended its scope \begin{added}to languages with concurrency\end{added} \citep{raad-et-al-popl2022},  demonstrated its real-world effectiveness \citep{le-et-al-oopsla2022} and spawned alternatives \begin{added}such as the \emph{Outcome Logic} of \citet{zilberstein-et-al-oopsla2023}, which additionally supports reasoning about correctness\end{added}.  A key feature of these logics is \emph{under-approximation} as a means to achieve true positives.  In incorrectness logic, states in a postcondition must be \emph{reachable}.  This leads to the need to reason about termination, but it also allows for an analysis to dynamically drop disjunctions in order to scale well.  The foundation of our system is not under-approximation, but necessity (via complement) -- when $M : A$ is provable, it remains the case that the type $A$ is an \emph{over-approximation} of $M$.  This makes the system closer to the works of \del{on the one hand,} \citet{cousot-et-al-vmcai2011,coutsot-et-al-vmcai2013} and \citet{mackay-et-al-oopsla2022}.  The former considers the problem of inferring \emph{necessary preconditions} in order to detect contract violations that will lead to assertion failures.  The latter uses necessity as a means to specify the \emph{robustness} of module specifications in object-oriented programming.

\paragraph{Type systems with complement.}
A small number of systems in the literature are equipped with a complement, or negation type operator.  However, to the best of our knowledge, none allow for reasoning about incorrectness by proving that programs go wrong.
A complement operator appears in the system of \citet{aiken-et-al-popl1994} and a negation in the work of \citet{parreaux-et-al-oopsla2022}.  Both systems are more sophisticated than ours, with the latter giving a treatment of a realistic programming language and a detailed and insightful account of type inference -- in fact this is the main motivation for the complement in both of these papers.  A complement operator is also present in the line of work on set-theoretic types and semantic subtyping (see e.g. \cite{castagna-et-al-icfp2011,castagna-et-al-popl2022,Castagna2024}) in which the subtyping relation is defined directly as inclusion of the induced sets.  A much simpler system is that of \cite{pearce-vmcai2013}, which lacks function types, and axiomatises a subtyping relation that is shown complete with respect to a set theoretic model.    Finally, in the dependent type system of \citet{unno-et-al-popl2018}, no typable program can go wrong, but the system is so expressive that one can compute directly the complements of arbitrary types without the need for a formal operator.

\paragraph{Coimplication and bi-intuitionistic logic.}
Bi-intuitionistic logic extends intuitionistic logic with a coimplication (also called subtraction, or pseudo-difference), dual to implication.  It was first studied by \citet{rauszer-stlog1974}, but the cut-elimination result was later shown to be flawed by \citet{pinto-uustalu-aratrm2009}.  A version of the Curry-Howard correspondence has been given for a version of bi-intuitionistic logic by \citet{crolard-jlc2004}, but it is for a natural deduction system, in the style of Parigot's $\lambda\mu$-calculus.  \begin{added}Coimplication is at the heart of the work of \citet{choudhury-et-al-popl25}, though in the context of language design. Whereas our paper views a type as a property, they are interested in designing a calculus corresponding to the logical content of the type. In particular, rather than viewing coimplication as a property of functions, they design a calculus in which the proofs of coimplications have their own term construction, the coabstraction.\end{added}

% An incorrectness triple such as $[A]M[B]$ could be seen as analogous to a function typing $M : A \Rightarrow B$ in which every value in the type $B$ can be obtained by applying $M$ to some input in $A$.  By contrast, our necessity arrow requires that no $B$ can be obtained except from an input in $A$.    However, set against this, one cannot simply 

% \input{div-effect}

\begin{acks}
We are grateful for many improvements suggested by the anonymous reviewers, and for useful discussions with our colleagues in the Programming Languages Research Group.  We also especially thank Tomos Sherlock for highlighting an issue with the original version of the proof of Inversion.
\end{acks}

\bibliographystyle{ACM-Reference-Format}
\bibliography{references}

\iftoggle{supplementary}{%
  \clearpage

  \appendix

  \section{Supplementary Materials Related to Subtyping}

This appendix contains proofs of:
\begin{itemize}
  \item The fact that complement is an order-reversing involution.
  \item The soundness of the subtyping relation (Theorem~\ref{thm:subtyping-soundness}).
  \item The equivalence of the original and alternative characterisations of subtyping (Theorem~\ref{thm:subtyping-equivalence}).
\end{itemize}

\subsection{Proof of complement as an order-reversing involution}
We justify the claim that complement acts as an order-reversing involution.
\begin{lemma}\label{thm:ortholattice}
  For all $A$ and $B$: $\CompPure{\CompPure{A}} \tyeq A$ and, if $A \subtype B$, then $\CompPure{B} \subtype \CompPure{A}$.
\end{lemma}
\begin{proof}
  For the first part, we show that $\CompPure{\CompPure{A}} \subtype A$ and $A \subtype \CompPure{\CompPure{A}}$:
  
  \begin{mathpar}
    \infer*[Left=CompL]{
      \infer*[Left=Refl]{ }{\CompPure{A} \subtype \CompPure{A}}
    }
    {
      \CompPure{\CompPure{A}} \subtype A
    }
    \and
    \infer*[Left=CompR]{
      \infer*[Left=Refl]{ }{\CompPure{A} \subtype \CompPure{A}}
    }
    {
      A \subtype \CompPure{\CompPure{A}}
    }
  \end{mathpar}

  For the second part, assume  $A \subtype B$. Then, since $B \subtype \CompPure{\CompPure{B}}$, 
  \begin{mathpar}
    \infer*[Left=CompR]{
      \infer*[Left=Trans]
      { 
        A \subtype B 
        \and
        B \subtype \CompPure{\CompPure{B}}
      }
      {
        A \subtype \CompPure{\CompPure{B}}
      }
    }
    {
      \CompPure{B} \subtype \CompPure{A}
    }
  \end{mathpar}
\end{proof}

\subsection{Proof of Subtyping Soundness, Theorem~\ref{thm:subtyping-soundness}}

Theorem~\ref{thm:subtyping-soundness}:
\begin{quote}
    If $A \subtype B$ then $\mng{A} \subseteq \mng{B}$.
\end{quote}

\begin{proof}
  By induction on the derivation of $A \subtype B$.
  \begin{description}
    \item[\rlnm{Refl}] Immediate.
    \item[\rlnm{Top}] Immediate from the definition of $\topty$ as the set of all closed normal forms.
    \item[\rlnm{Ok}] Suppose $\intty \vee \pairvalty \vee \funty \vee \atomty \subtype \okty$. Then $\mng{\intty \vee \pairvalty \vee \funty \vee \atomty} = \mng{\intty} \cup \mng{\pairvalty} \cup \mng{\funty} \cup \mng{\atomty}$. Now, we just need to check that each of these four sets are subsets of $\mng{\okty}$ i.e they are sets of values, which they are by the definition of the grammar of this language.
    \item[\rlnm{Disj}] Suppose $A \subtype \CompPure{B}$ and $A \neq B,\, A \distype B$.
          If $A$ and $B$ are atoms, then let $\mng{A} = \{\atom{a}\}$ and $\mng{B} = \{\atom{b}\}$ where $\atom{a} \neq \atom{b}$. Then $\atom{a} \in \mng{\topty}$ so $\{\atom{a}\} \subseteq \mng{\topty} \setminus \{\atom{b}\}$ as required.
          If $A,B \in \{\intty,\pairty,\funty,\atomty\}$, by definition of \mng{} for these types, these sets are pairwise disjoint. Thus, for each pair of types, $\mng{A} \subseteq \mng{\CompPure{B}}$.
    \item[\rlnm{Atom}] Suppose $\atom{a} \subtype \atomty$. Then $\mng{\atom{a}} = \{\atom{a}\} \subseteq \mng{\atomty}$.
    \item[\rlnm{PairC}] Suppose $(\CompPure{A},\,\topty) \vee (\okty,\,\CompPure{B}) \subtype \CompPure{(A,\,B)}$.
          Then, by definition $\mng{(\CompPure{A},\,\topty) \vee (\okty,\,\CompPure{B})} = \mng{(\CompPure{A},\,\topty)} \cup \mng{(\okty,\,\CompPure{B})}$,
          where we have
          \begin{align}
            \mng{(\CompPure{A},\,\topty)} & =
            \{(M,\,N) \mid M \in \mng{\CompPure{A}} \cap \mng{\okty}, N \in \mng{\topty} \}
            \cup
            \{(M,\,N) \mid M \in \mng{\CompPure{A}} \cap \mng{\CompPure{\okty}}, N \in \mng*{\topty} \}
          \end{align}
          and, since $\mng{\okty} \cap \mng{\CompPure{\okty}} = \varnothing $, 
          \begin{align}
            \mng{(\okty,\,\CompPure{B})} & =
            \{(P,\,Q) \mid P \in \mng{\okty}, Q \in \mng{\CompPure{B}} \}
            \cup
            \{(P,\,Q) \mid P \in \mng{\okty} \cap \mng{\CompPure{\okty}}, Q \in \mng*{\CompPure{B}} \}      \\
                                         & =  \{(P,\,Q) \mid P \in \mng{\okty}, Q \in \mng*{\CompPure{B}} \}
          \end{align}

          Let $N = (R, S) \in \mng{(\CompPure{A},\,\topty) \vee (\okty,\,\CompPure{B})}$.
          Then either
          \begin{enumerate}
            \item $R \in \mng{\CompPure{A}} \cap \mng{\okty}$ \emph{or}
            \item $R \in \mng{\CompPure{A}} \cap \mng{\CompPure{\okty}}$ \emph{or}
            \item $S \in \mng*{\CompPure{B}}$.
          \end{enumerate}
          In any case, this means $N \notin \mng{(A, \,B)}$.
          Since $N\in \mng{\topty}$, we have $N \in \mng{\topty} \setminus \mng{(A,\,B)} = \mng{\CompPure{(A,\,B)}}$.
    \item[\rlnm{Pair}] Suppose $(A,B) \subtype (A'\,B')$. The induction hypotheses give:
          \begin{enumerate}[({IH}1)]
            \item $\mng{A} \subseteq \mng{A'}$
            \item $\mng{B} \subseteq \mng{B'}$
          \end{enumerate}
          Note that by definition of $\mng*{}$, for types $A$ and $B$, if $\mng{A} \subseteq \mng{B}$, then $\mng*{A} \subseteq \mng*{B}$ also.
          Then, since set intersection is monotonic,
          \begin{align} 
            \mng{(A,\,B)} &=
            \{ (M,\,N) \mid M \in \mng{A} \cap \mng{\okty}, N \in \mng{B}\} \cup \{(M,\,N) \mid M \in \mng{A} \cap \mng{\CompPure{\okty}}, N \in \mng*{B} \}  \\
            &\subseteq 
            \{ (M,\,N) \mid M \in \mng{A'} \cap \mng{\okty}, N \in \mng{B'}\} \cup \{(M,\,N) \mid M \in \mng{A'} \cap \mng{\CompPure{\okty}}, N \in \mng*{B'} \} \\
            &= \mng{(A',\,B')}
          \end{align}
    \item[\rlnm{Fun}] Suppose $A \to B \subtype A' \to B'$. The induction hypotheses give:
          \begin{enumerate}[({IH}1)]
            \item $\mng{A'} \subseteq \mng{A}$
            \item $\mng{B} \subseteq \mng{B'}$
          \end{enumerate}
          Again, by definition of $\mng*{}$, by (IH1) we have $\mng*{A'} \subseteq \mng*{A}$. Similarly, by (IH2) we have $\mng*{B} \subseteq \mng*{B'}$. Then, $\mng{A \to B} = \{ \abs{x}{M} \mid \forall N \in \mng*{A}.\:M[N/x] \in \mng*{B} \} \subseteq \{ \abs{x}{M} \mid \forall N \in \mng*{A'}.\:M[N/x] \in \mng*{B'} \} = \mng{A' \to B'}$.
    \item[\rlnm{UnionL}] Suppose ${\textstyle\bigvee_{i=1}^k B_i} \subtype A$. The induction hypothesis gives:
          \begin{enumerate}[(IH)]
            \item $\mng{B_j} \subseteq \mng{A}  (\forall j\in[1,k])$
          \end{enumerate}
          Then, by basic properties of set inclusion, $\mng{\textstyle\bigvee_{i=1}^k B_i} = \textstyle\bigcup_{i=1}^k \mng{B_i} \subseteq \mng{A}$.
    \item[\rlnm{UnionR}] Suppose $A \subtype {\textstyle\bigvee_{i=1}^k B_i}$. The induction hypothesis gives:
          \begin{enumerate}[(IH)]
            \item $\mng{A} \subseteq \mng{B_j},  j\in[1,k]$
          \end{enumerate}
          Then, by basic properties of set inclusion, $\mng{A} \subseteq \mng{B_j} \subseteq \textstyle\bigcup_{i=1}^k \mng{B_i} = \mng{\textstyle\bigvee_{i=1}^k B_i}$.
    \item[\rlnm{CompL}] Suppose $\CompPure{A} \subtype B$. The induction hypothesis gives:
          \begin{enumerate}[(IH)]
            \item $\mng{\CompPure{B}} \subseteq \mng{A}$
          \end{enumerate}
          Let $N \in \mng{\CompPure{A}} = \mng{\topty} \setminus \mng{A}$. By definition, this means $N \notin \mng{A}$. By the induction hypothesis $\mng{\CompPure{B}} = \mng{\topty} \setminus \mng{B}  \subseteq \mng{A}$, we must also have that  $N \notin \mng{\topty} \setminus \mng{B}$. By $\rlnm{Top}$, $N \in \mng{\topty}$, so we must have $N \in \mng{B}$.
    \item[\rlnm{CompR}] Suppose $A \subtype \CompPure{B}$. The induction hypothesis gives:
          \begin{enumerate}[(IH)]
            \item $\mng{B} \subseteq \mng{\CompPure{A}}$
          \end{enumerate}
          Let $N \in \mng{A}$. By $\rlnm{Top}$, $N \in \mng{\topty}$.
          By the induction hypothesis $\mng{B} \subseteq  \mng{\topty} \setminus \mng{A}$ and that $N \in \mng{A}$, we must have $N \notin \mng{B}$.
          Thus, $N \in \mng{\topty} \setminus \mng{B} = \mng{\CompPure{B}}$ as required.
    \item[\rlnm{Trans}] Suppose $A \subtype B$. The induction hypotheses give:
          \begin{enumerate}[({IH}1)]
            \item $\mng{A} \subseteq \mng{C}$
            \item $\mng{C} \subseteq \mng{B}$
          \end{enumerate}
          Then, by transitivity of set inclusion, $\mng{A} \subseteq \mng{C} \subseteq \mng{B}$.
  \end{description}
\end{proof}

\subsection{Proof of Equivalence of the Subtyping Systems, Theorem~\ref{thm:subtyping-equivalence}}

We begin by showing that the alternative characterisation admits transitivity.

\begin{lemma}[Admissibility of Transitivity]\label{lem:alt-transitivity}
  If $A \subtypeA B$ and $B \subtypeA C$ then $A \subtypeA C$.
\end{lemma}

\begin{proof}
  The proof is by strong induction on the sum of the heights of the two derivations in the antecedent.  We analyse cases on the rules that can be used to conclude two judgements of the required form.  

  To begin, we note the following.  Whenever the judgement on the left is concluded using a rule whose conclusion admits an arbitrary type on the RHS, then $A \subtypeA C$ can be deduced no matter which rule was used to justify the judgement on the right.  
  \begin{description}
    \item[\rlnm{Refl}] In this case, $A=B$ and hence the second premise is already a proof of $A \subtypeA C$.
    % \item[\rlnm{CompTop'}] In this case we may assume some type $D$ and smaller proofs of $\CompPure{D} \subtypeA B$ and $D \subtypeA B$.  It follows from the induction hypothesis that $\CompPure{D} \subtypeA C$ and $D \subtypeA C$, from which the desired conclusion follows by \rlnm{CompTop'}.
    % \item[\rlnm{CompBot'}] In this case, we may assume some type $D$ and smaller proofs of $A \subtypeA \CompPure{D}$ and $A \subtypeA D$.  Then the result follows immediately by \rlnm{CompBot'}.
    \item[\rlnm{UnionL}] In this case, $A$ is of shape $\textstyle\bigvee_{i=1}^k A_i'$ and we may assume smaller proofs of $A_i' \subtypeA B$ for all $i \in [1,k]$.  Then it follows from the induction hypothesis that there are proofs of $A_i' \subtypeA C$ for all $i \in [1,k]$ and thus the result follows by \rlnm{UnionL}.
    \item[\rlnm{CompLTop}] In this case, $A$ is of shape $\botty$ and the result follows immediately by \rlnm{CompLTop}.
    \item[\rlnm{CompLOk}] In this case, $A = \CompPure{\okty}$ and we may assume a smaller proof of $D \subtype B$ where $D$ is one of the types allowed by the side condition.  Then it follows from the induction hypothesis that $D \subtype C$ and so the result follows by \rlnm{CompLOk}.
    \item[\rlnm{CompLUn}] In this case, $A$ is of shape $\CompPure*{\textstyle\bigvee_{i=1}^k A_i'}$ and we may assume a smaller proof of $A_j \subtypeA B$ for some $j$.  Then it follows from the induction hypothesis that $A_j \subtypeA C$ and the conclusion follows from \rlnm{CompLUn}.
    \item[\rlnm{CompLC}] In this case, $A$ is of shape $\CompPure{\CompPure{A'}}$ and we may assume a smaller proof of $A' \subtypeA B$.  It follows from the induction hypothesis that $A' \subtypeA C$ from which the result follows by \rlnm{CompLC}.
  \end{description}

  We note similarly that, whenever the judgement on the right is obtained using one of the rules that admits an arbitrary type on left, then $A \subtypeA C$ can be deduced independently of how the judgement on the left is justified.
  \begin{description}
    \item[\rlnm{Refl}] In this case $B = C$ and so the first premise is actually a proof of $A \subtypeA C$.
    % \item[\rlnm{CompRAtom1}] In this case $C$ is of shape $\CompPure{\atom{a}}$ and it follows from the induction hypothesis that there is some atom $\atom{b}$ with $a \neq b$ and $A \subtypeA \atom{b}$.  Then the result follows by \rlnm{CompRAtom1}.
    \item[\rlnm{Top}] In this case, $C = \topty$ and and the result follows immediately from another application of \rlnm{Top}.
    \item[\rlnm{OkR}] In this case $C=\okty$ and we may assume some type $D$ of those allowed by the side condition, and a smaller proof of $B \subtypeA D$.  It follows from the induction hypothesis that $A \subtypeA D$ and thus the result follows by \rlnm{OkR}.
    \item[\rlnm{CompRInt}, \rlnm{CompRPair}, \rlnm{CompRArr}, \rlnm{CompRAtoms}, \rlnm{CompRAtom2}] These cases follow by the same argument to the previous.
    \item[\rlnm{CompRUn}] In this case, $C$ has shape $\CompPure*{\textstyle\bigvee_{i=1}^k C_i}$ and we may assume a smaller proof of $B \subtype \CompPure{C_i}$ for each $i$.  Then it follows from the induction hypothesis that $A \subtypeA \CompPure{C_i}$ for each $i$, and the result follows by \rlnm{CompRUn}.
    \item[\rlnm{CompRC}] In this case, $C$ has shape $\CompPure{\CompPure{C'}}$ and we may assume a smaller proof of $B \subtypeA C'$.  Then it follows from the induction hypothesis that $A \subtypeA C'$ and the result follows by \rlnm{CompRC}.
    % \item[\rlnm{CompTop'}] In this case, we may assume some type $D$ and smaller proofs of $\CompPure{D} \subtypeA C$ and $D \subtypeA C$.  Then the result follows immediately using \rlnm{CompTop'}.
    % \item[\rlnm{CompBot'}] In this case, we may assume some type $D$ and smaller proofs of $B \subtypeA \CompPure{D}$ and $B \subtypeA D$.  Then it follows from the induction hypothesis that $A \subtypeA \CompPure{D}$ and $A \subtypeA D$ and the result follows by \rlnm{CompBot'}.
    \item[\rlnm{UnionR}] In this case, $C$ has shape $\bigvee_{i=1}^k C_i$ and we may assume a smaller proof of $B \subtype C_j$ for some $j$.  Then it follows from the induction hypothesis that $A \subtype C_j$ and the result follows by \rlnm{UnionR}.
  \end{description}

  Having discharged all those combinations where the rule for the left premise had an arbitrary type $B$ on the RHS in the conclusion and the rule justifying the right premise has an arbitrary type $B$ on the LHS in the conclusion, all that remains to consider those combinations of rules that require agreement of a specific syntactic shape for the pivot type $B$.
  \begin{description}
    \item[\rlnm{CompLAtoms / CompRAtoms}] In this case, $B$ is of shape $\CompPure{\atomty}$ and $C$ is of shape $\CompPure{\atom{a}}$ we may assume some type $D$ allowed by the side condition and a smaller proof of $A \subtypeA D$.  Then $A \subtypeA \CompPure{\atom{a}}$ follows from \rlnm{CompRAtom2}.
    \item[\rlnm{CompRPair} / \rlnm{CompLPair}] In this case, $B$ is of shape $\CompPure{(B',C')}$ and $C$ is of shape $\CompPure{(B'',C'')}$ and we may assume some type $D$ as allowed by the side condition on the left rule, such that $A \subtype D$.  Then the result follows immediately by \rlnm{CompRPair}.
    \item[\rlnm{CompRPairL} / \rlnm{CompLPair}]  In this case, $A$ is of shape $(A',B')$, $B$ is of shape $\CompPure{(C',D')}$ and $C$ is of shape $\CompPure{(C'',D'')}$ and we may assume smaller proofs of $A' \subtypeA \CompPure{C'}$ and $\CompPure{C'} \subtypeA \CompPure{C''}$ and $\CompPure{D'} \subtypeA \CompPure{D''}$.  Then it follows from the induction hypothesis that $A' \subtypeA \CompPure{C''}$ and the result follows by \rlnm{CompRPairL}.
    \item[\rlnm{CompRPairR} / \rlnm{CompLPair}] This case is symmetrical to the previous.
    \item[\rlnm{Pair} / \rlnm{CompRPairL}] In this case, $A$ has shape $(A_1,A_2)$, $B$ has shape $(B_1,B_2)$, $C$ has shape $\CompPure{(C_1,C_2)}$ and we may assume smaller proofs of $A_1 \subtypeA B_1$, $A_2 \subtypeA B_2$ and $B_1 \subtypeA \CompPure{C_1}$.  Then it follows from the induction hypothesis that $A_1 \subtypeA \CompPure{C_1}$ and the result follows from \rlnm{CompRPairL}.
    \item[\rlnm{Pair} / \rlnm{CompRPairR}] This case is symmetrical to the previous.
    \item[\rlnm{CompRArr} / \rlnm{CompLArr}] In this case, $B$ is of shape $\CompPure*{B_1 \to B_2}$ and $C$ is of shape $\CompPure*{B_1' \to B_2'}$ and we may assume some type $D$ of those allowed by the side condition of the first rule, and a smaller proof of $A \subtype D$.  Then the result follows immediately by \rlnm{CompRArr}.
    \item[\rlnm{CompRUn} / \rlnm{CompLUn}] In this case, $B$ is of shape $\CompPure{(\textstyle\bigvee_{i=1}^k B_i)}$ and we may assume smaller proofs of $A \subtypeA B_i$ for each $i \in [1,k]$ and $B_j \subtypeA C$ for some $j \in [1,k]$.  It follows from the induction hypothesis that $A \subtype C$, as required.
    \item[\rlnm{CompRC} / \rlnm{CompLC}] In this case, $B$ is of shape $\CompPure{\CompPure{B'}}$ and we may assume smaller proofs of $A \subtype B'$ and $B' \subtype C$ from which the result follows by the induction hypothesis.
    \item[\rlnm{Pair} / \rlnm{Pair}] In this case, $A$ is of shape $(A_1,A_2)$, $B$ is of shape $(B_1,B_2)$, $C$ is of shape $(C_1,C_2)$ and we may assume smaller proofs of $A_1 \subtypeA B_1$, $A_2 \subtypeA B_2$, $B_1 \subtypeA C_1$ and $B_2 \subtypeA C_2$.  It follows from the induction hypothesis that $A_1 \subtypeA C_1$ and $A_2 \subtypeA C_2$ and so the result follows from \rlnm{Pair}.
    \item[\rlnm{Fun} / \rlnm{Fun}] This case is analogous to the previous.
    \item[\rlnm{CompLPair} / \rlnm{CompLPair}]  This case is analogous to the previous.
    \item[\rlnm{CompLArr} / \rlnm{CompLArr}] This case is analogous to the previous.
    \item[\rlnm{UnionR} / \rlnm{UnionL}] In this case, $B$ has shape $\textstyle\bigvee_{i=1}^k B_i$ and we may assume smaller proofs of $A \subtypeA B_j$ for some $j \in [1,k]$ and $B_i \subtype C$ for all $i \in [1,k]$.  Therefore, the result follows directly from the induction hypothesis.
  \end{description}  
\end{proof}

We can now see that the complement of the alternative characterisation is an order involution in the following sense.
\begin{lemma}
  If $A \subtypeA \CompPure{\CompPure{B}}$ or $\CompPure{\CompPure{A}} \subtypeA B$ then $A \subtypeA B$
\end{lemma}
\begin{proof}
  Suppose $A \subtypeA \CompPure{\CompPure{B}}$ (the other case is analogous).  Then we have $B \subtypeA B$ by reflexivity and hence \begin{added}$\CompPure{\CompPure{B}} \subtypeA B$\end{added} by \rlnm{CompLC}.  Then the result follows by transitivity.
\end{proof}

We next show that the complement operator in the alternative characterisation is also order-reversing.

\begin{lemma}
If $A \subtypeA B$ then $\CompPure{B} \subtypeA \CompPure{A}$.
\end{lemma}
\begin{proof}
  The proof is by induction on the derivation of $A \subtypeA B$.
  \begin{description}
    \item[\rlnm{Refl}] In this case the result is immediate.
    % \item[\rlnm{BotTop}] In this case, $A = \CompPure{\topty}$ and $B = \topty$.  Then note that $\CompPure{\topty} \subtypeA \topty$ by \rlnm{BotTop} and hence $\CompPure{\topty} \subtypeA \CompPure{\CompPure{\topty}}$ by \rlnm{CompRC}.
    \item[\rlnm{Top}] In this case, the result follows immediately by \rlnm{CompLTop}.
    \item[\rlnm{OkR}] In this case, $B = \okty$.  It follows from the induction hypothesis that there is some type $C$ as required by the side condition, such that $\CompPure{C} \subtypeA \CompPure{A}$ and hence the result follows by \rlnm{CompLOk}.
    \item[\rlnm{Atom}] In this case, $A$ is of shape $\atom{a}$ and $B = \atomty$.  Then, since $\CompPure{\atomty} \subtypeA \CompPure{\atomty}$ by \rlnm{Refl}, $\CompPure{\atomty} \subtypeA \CompPure{\atom{a}}$ by \rlnm{CompRAtom2}.
    \item[\rlnm{CompRAtom1}] In this case, $A$ is of shape $\atom{b}$ and $B$ of shape $\CompPure{\atom{a}}$ and we may assume $a \neq b$.  Then also by \rlnm{CompRAtom1}, $\atom{a} \subtypeA \CompPure{\atom{b}}$ and the result follows by \rlnm{CompLC}.
    \item[\rlnm{CompRInt}, \rlnm{CompRPair}, \rlnm{CompRArr}, \rlnm{CompRAtoms}, \rlnm{CompRAtom2}] These cases are similar, so we will only argue the first.  In this case, $B = \CompPure{\intty}$ and it follows from the induction hypothesis that there is some type $D$ of those allowed by the side condition, such that $\CompPure{D} \subtype \CompPure{A}$.  Whichever type $D$ was chosen, it follows that $\intty \subtypeA \CompPure{D}$ by one of the rules in this case.  Therefore, by transitivity, $\intty \subtypeA \CompPure{A}$ and $\CompPure{\CompPure{\intty}} \subtypeA \CompPure{A}$ by \rlnm{CompLC}.
    \item[\rlnm{CompRPairL}, \rlnm{CompRPairR}] We will just argue the first, since the second is similar.  In this case, $A$ is of shape $(A_1,A_2)$ and $B$ of shape $\CompPure{(B_1,B_2)}$ and it follows from the induction hypothesis that $\CompPure{\CompPure{B_1}} \subtypeA \CompPure{A_1}$.  Then $B_1 \subtypeA \CompPure{A_1}$ by the order involution property of complement.  By \rlnm{CompRPairL} we can then conclude $(B_1,B_2) \subtypeA \CompPure*{A_1,A_2}$ and, finally, $\CompPure{\CompPure{(B_1,B_2)}} \subtypeA \CompPure{(A_1,A_2)}$ follows by \rlnm{CompLC}.
    \item[\rlnm{CompRUn}, \rlnm{CompLUn}] These cases are symmetrical, so we only argue the former.  In this case, $B$ is of shape $\CompPure{(\bigvee_{i=1}^k B_i')}$ and it follows from the induction hypothesis that $\CompPure{\CompPure{B_i'}} \subtypeA \CompPure{A}$ for all $i \in [1,k]$.  Then, by the order involution property, $B_i' \subtypeA \CompPure{A}$ for all $i \in [1,k]$.  Hence, $\textstyle\bigvee_{i=1}^k B_i \subtypeA \CompPure{A}$ by \rlnm{UnionL}.  Finally, the desired result is obtained from this by \rlnm{CompLC}.
    \item[\rlnm{CompRC}, \rlnm{CompLC}] These cases are symmetrical, so we only argue the former.  In this case, $B$ is of shape $\CompPure{\CompPure{B'}}$ and it follows from the induction hypothesis that $\CompPure{B'} \subtype \CompPure{A}$.  Then $\CompPure{\CompPure{\CompPure{B'}}} \subtypeA \CompPure{A}$ follows by \rlnm{CompLC}.
    % \item[\rlnm{CompTop'}, \rlnm{CompBot'}] These cases are symmetrical so we only argue the former.  In this case, it follows from the induction hypothesis that there is some type $C$ such that $\CompPure{B} \subtype \CompPure{\CompPure{C}}$ and $\CompPure{B} \subtypeA \CompPure{C}$.  Then, by the order involution property, we obtain $\CompPure{B} \subtype C$ and so the result follows by \rlnm{CompBot'}.
    \item[\rlnm{Pair}, \rlnm{Fun}] These cases are similar, so we only argue the former.  In this case, $A$ is of shape $(A_1,A_2)$ and $B$ of shape $(B_1,B_2)$, and it follows from the induction hypothesis that $\CompPure{B_1} \subtype \CompPure{A_1}$ and $\CompPure{B_2} \subtype \CompPure{A_2}$.  Then the result follows by \rlnm{CompLPair}.
    \item[\rlnm{UnionR}, \rlnm{UnionL}] These cases are symmetrical, so we only argue the former.  In this case, $B$ is of shape $\textstyle\bigvee_{i=1}^k B_i'$ and it follows from the induction hypothesis that there is some $j$ such that $\CompPure{B_j'} \subtypeA \CompPure{A}$.  Then the result is obtained by \rlnm{CompLUn}.
    \item[\rlnm{CompLPair}, \rlnm{CompLArr}] These cases are analogous, so we only argue the former.  In this case, $A$ is of shape $\CompPure{(A_1,A_2)}$ and $B$ of shape $\CompPure{(B_1,B_2)}$ and it follows from the induction hypothesis that $\CompPure{\CompPure{B_1}} \subtypeA \CompPure{\CompPure{A_1}}$ and $\CompPure{\CompPure{B_2}} \subtypeA \CompPure{\CompPure{A_2}}$.  Then it follows from the order involution property that $B_1 \subtypeA A_1$ and $B_2 \subtypeA A_2$, and by \rlnm{Pair}, $(B_1,\,B_2) \subtypeA (A_1,A_2)$.  Finally, the result follows by \rlnm{CompLC} and \rlnm{CompRC}.
    \item[\rlnm{CompLOk}] In this case, $A$ is of shape $\CompPure{\okty}$ and it follows from the induction hypothesis that there is some type $\CompPure{C}$ of the set allowed by the side condition, such that $\CompPure{B} \subtypeA \CompPure{\CompPure{C}}$.  Then $\CompPure{B} \subtype C$ by the order involution property, and $\CompPure{B} \subtypeA \okty$ follows from \rlnm{OkR}.  Finally, the result can be obtained by \rlnm{CompRC}.
    \item[\rlnm{CompLAtoms}]. In this case, $A$ is of shape $\CompPure{\atomty}$ and $B$ is of shape $\CompPure{\atom{a}}$ the result can be obtained from \rlnm{Atom} by \rlnm{CompLC} and \rlnm{CompRC}.
  \end{description}
\end{proof}

This makes it easy to recover \rlnm{CompL} and \rlnm{CompR}, i.e. to show that these are admissible in the alternative system.
Having recovered \rlnm{Trans}, \rlnm{CompL} and \rlnm{CompR}, it is straightforward to see that the two systems are equivalent.

Theorem~\ref{thm:subtyping-equivalence}:
\begin{quote}
  $A \subtype B$ iff $A \subtypeA B$
\end{quote}

\begin{proof}
  In the forward direction, the proof is by induction on $A \subtype B$.
  \begin{description}
    % \item[\rlnm{Top}] In this case, $B$ is $\topty$.  By \rlnm{BotTop}, $\CompPure{\topty} \subtypeA \topty$, and then by \rlnm{CompTop} $A \subtypeA \topty$.
    \item[\rlnm{Ok}] In this case, $A$ is $\intty \vee \pairvalty \vee \funty \vee \atomty$ and $B$ is $\okty$.  Then, since each of the disjuncts $X$ can be shown a subtype of $\okty$ by \rlnm{Refl} and \rlnm{OkR}, the result follows by \rlnm{UnionL}.
    \item[\rlnm{Disj}] In this case, $B$ is of shape $\CompPure{C}$, $A \neq C$ and either $A$ and $C$ are atoms, or they are both drawn from the set $\{\intty,\,\pairty,\,\funty,\,\atomty\}$.  For each such pair of $A$ and $C$, the required conclusion can be derived using one of \rlnm{CompRAtom1}, \rlnm{CompRInt}, \rlnm{CompRPair}, \rlnm{CompRArr} or \rlnm{CompRAtoms}.
    \item[\rlnm{PairC}] In this case, $A$ is of shape $(\CompPure{A'},\topty) \vee (\okty,\CompPure{B'})$ and $B$ is $\CompPure{(A',B')}$.  By \rlnm{Refl} $\CompPure{A'} \subtypeA \CompPure{A'}$ and so $(\CompPure{A'},\topty) \subtypeA \CompPure{(A',B')}$ by \rlnm{CompRPairL}.  A similar argument shows $(\okty,\CompPure{B'}) \subtypeA \CompPure{(A',B')}$.  Hence, the required conclusion follows by \rlnm{UnionL}.
    % \item[\rlnm{Comp}] In this case, $A$ is $\topty$ and $B$ is of shape $C \vee \CompPure{C}$.  Then $A \subtypeA A \vee \CompPure{A}$ and $\CompPure{A} \subtypeA A \vee \CompPure{A}$ both follow by \rlnm{UnionR}, so that $\topty \subtypeA A \vee \CompPure{A}$ follows by \rlnm{CompTop'}.
    \item[\rlnm{CompR}] In this case, $B$ is of shape $\CompPure{C}$.  It follows from the induction hypothesis that $C \subtypeA \CompPure{A}$.  Then, by the order-reversing property of complement, $\CompPure{\CompPure{A}} \subtypeA \CompPure{C}$ and the result follows by the involution property.
    \item[\rlnm{CompL}] This is symmetrical to the previous case.
    \item[\rlnm{Refl}, \rlnm{Top} \rlnm{Atom}, \rlnm{Pair}, \rlnm{Fun}, \rlnm{UnionL}, \rlnm{UnionR}] These rules are identical in both systems.  
  \end{description}
  In the backward direction, the proof is by induction on $A \subtypeA B$.
  \begin{description}
    \item[\rlnm{OkR}] In this case, $B = \okty$.  It follows from the induction hypothesis that there is some $C \in \{\intty,\pairvalty,\funty,\atomty\}$ such that $A \subtype C$.  Then $C \subtype \intty \vee \pairvalty \vee \funty \vee \atomty$ follows by \rlnm{UnionR} and $A \subtype \okty$ follows by transitivity.
    \item[\rlnm{CompRAtom1}] In this case $A$ is of shape $\atom{b}$ and $B$ is of shape $\CompPure{\atom{a}}$.  Then the result follows from \rlnm{Disj}.
    \item[\rlnm{CompRAtom2}] In this case $B$ is of shape $\CompPure{\atom{a}}$ and it follows from the induction hypothesis that there is some type $C$ allowed by the side condition and such that $A \subtype C$.  In all cases, $A \subtype \CompPure{\atomty}$ follows by \rlnm{Disj}, $\CompPure{\atomty} \subtype \CompPure{\atom{a}}$ follows from \rlnm{Atom} by the order-reversing property of complement, and then the result follows by transitivity.
    \item[\rlnm{CompRInt}, \rlnm{CompRPair}, \rlnm{CompRArr}, \rlnm{CompRAtoms}] These cases are similar so we only show the first.  In this case, $B = \CompPure{\intty}$ and it follows from the induction hypothesis that there is some $C \in \{\pairty,\funty,\atomty\}$ such that $A \subtype C$.  Then the result follows from \rlnm{Disj} and transitivity. 
    \item[\rlnm{CompRPairL}, \rlnm{CompRPairR}] These cases are similar, so we only show the former.  In this case, $A$ is of shape $(C,D)$ and $B$ is of shape $\CompPure{(C',D')}$ and the induction hypothesis yields $C \subtype \CompPure{C'}$.  Since $D \subtype \topty$ by \rlnm{Top}, $(C,D) \subtype (\CompPure{C'},\topty)$ by \rlnm{Pair}.  Then $(C,D) \subtype (\CompPure{C'},\topty) \vee (\okty,\CompPure{D'})$ by \rlnm{UnionR} and transitivity.  Finally, the result is obtained by \rlnm{PairC} and transitivity.
    \item[\rlnm{CompRUn}, \rlnm{CompLUn}] These cases are symmetrical, so we only show the former.  In this case, $B$ is of shape $\CompPure*{\textstyle\bigvee_{i=1}^k B_i}$ and it follows from the induction hypothesis that $A \subtype \CompPure{B_i}$ for all $i$.  Then it follows from \rlnm{CompR} that $\CompPure{\CompPure{B_i}} \subtype \CompPure{A}$ for all $i$, and therefore that $B_i \subtype \CompPure{A}$ for all $i$.  Consequently, $\bigvee_{i=1}^k B_i \subtype \CompPure{A}$ and the result follows by \rlnm{CompR}.
    \item[\rlnm{CompRC}, \rlnm{CompLC}] These cases are symmetrical, so we only show the former.  In this case, $B$ is of shape $\CompPure{\CompPure{D}}$ and it follows from the induction hypothesis that $A \subtype D$.  Then it follows from the involution property of complement (\cref{thm:ortholattice}) that $A \subtype \CompPure{\CompPure{D}}$.
    % \item[\rlnm{CompTop}] In this case, it follows from the induction hypothesis that $\CompPure{B} \subtype B$.  Since, $B \subtype B$ by \rlnm{Refl}, it follows from \rlnm{UnionR} that $B \vee \CompPure{B} \subtype B$.  Hence, the result follows from \rlnm{Top}, \rlnm{Comp} and transitivity.
    % \item[\rlnm{CompBot}] In this case, it follows from the induction hypothesis that $A \subtype \CompPure{A}$ and so, by a symmetrical argument to the above, we obtain $\CompPure{B} \subtype \CompPure{A}$ from which the result follows by the order-reversing property of complement (\cref{thm:ortholattice}).
    \item[\rlnm{CompLPair}, \rlnm{CompLArr}] These cases are similar, so we show only the former.  In this case, $A$ is of shape $\CompPure*{C_1,C_2}$ and $B$ is of shape $\CompPure*{D_1,D_2}$ and it follows from the induction hypothesis that $\CompPure{C_1} \subtype \CompPure{D_1}$ and $\CompPure{C_2} \subtype \CompPure{D_2}$.  Then since complement is order-reversing, $D_1 \subtype C_1$ and $D_2 \subtype C_2$.  Hence, $(D_1,D_2) \subtype (C_1,C_2)$ by \rlnm{Pair} and the result follows by the order-reversing property of complement (\cref{thm:ortholattice}).
    \item[\rlnm{CompLOk}] In this case, $A = \CompPure{\okty}$ and it follows from the induction hypothesis that there is some $C \in \{\intty, \funty, \atomty, \pairvalty\}$ such that $\CompPure{C} \subtype B$.  Then $\CompPure{B} \subtype C$ by \rlnm{CompL} and $C \subtype \okty$ by \rlnm{Ok}, so $\CompPure{B} \subtype \okty$ by transitivity.  Then \rlnm{CompL} gives $\CompPure{\okty} \subtype B$ as required.
    \item[\rlnm{CompLAtoms}] In this case $A$ is of shape $\CompPure{\atomty}$, $B$ is of shape $\CompPure{\atom{a}}$.  Then the result follows from \rlnm{Atom} by the order-reversing property of complement.
    \item[\rlnm{Refl}, \rlnm{Top}, \rlnm{Atom}, \rlnm{Pair}, \rlnm{Fun}, \rlnm{UnionL}, \rlnm{UnionR}] These rules are identical in both systems.
  \end{description} 
\end{proof}

  \section{Derivability of the Refutation Rules, Theorem~\ref{thm:left-rules-derivable}}

Theorem~\ref{thm:left-rules-derivable}:
\begin{quote}
  All the rules of Figure~\ref{fig:simple-two-sided-left} are derivable.
\end{quote}

\begin{proof}
First observe that the swapping rules are derivable:
\begin{mathpar}
  \infer*[left=SubL]{
    \infer*[left=CompL]{
        \Gamma \types M:\CompPure{A} \Delta
    }
    {
      \Gamma,\,M:\CompPure*{\CompPure{A}} \types \Delta
    }
  }
  {
    \Gamma,\,M:A \types \Delta
  }
  \and
  \infer*[left=SubR]{
    \infer*[left=CompR]{
        \Gamma,\,M:\CompPure{A} \types \Delta
    }
    {
      \Gamma \types M:\CompPure*{\CompPure{A}},\,\Delta
    }
  }
  {
    \Gamma \types M:A,\,\Delta
  }
\end{mathpar}
We derive each of the remaining rules in turn (inline).
\begin{description}
  \item[\rlnm{BotL}] This judgement follows from \rlnm{Top} by \rlnm{CompL}.
  \item[\rlnm{NumL}] This judgement follows from \rlnm{Num} by \rlnm{CompL}.
  \item[\rlnm{FunL}] This judgement follows from \rlnm{Abs} and \rlnm{Top} by \rlnm{CompL}.
  \item[\rlnm{AtomL}] This judgement follows from \rlnm{Atom} by \rlnm{CompL}.
  \item[\rlnm{PairL1}] This judgement follows from \rlnm{Pair} and \rlnm{Top} by \rlnm{CompL}
  \item[\rlnm{PairLL}] Suppose $\Gamma,\,M_1:A_1 \types \Delta$ (the other case is similar).  We have $\Gamma \types M_1:\CompPure{A_1},\,\Delta$ by \rlnm{CompR} and $\Gamma \types M_2:\topty,\,\Delta$ by \rlnm{Top}.  Then, $\Gamma \types (M_1,M_2) : (\CompPure{A_1},\topty),\,\Delta$ by \rlnm{Pair}.  Since $(A_1,\topty) \subtype \CompPure{(A_1,A_2)}$, we obtain $\Gamma \types (M_1,M_2) : \CompPure{(A_1,A_2)}$ by \rlnm{SubL} and finally the result is obtained by \rlnm{SwapL}.
  \item[\rlnm{PrjL1}] Suppose $\Gamma \types M:(\CompPure{A},\,B),\,\Delta$.  It follows from \rlnm{Prj$1$} that $\Gamma \types \pi_1(M) : \CompPure{A},\,\Delta$ and the result follows from \rlnm{SwapL}.
  \item[\rlnm{PrjL2}] Suppose $\Gamma \types M:(A,\,\CompPure{B}),\,\Delta$.  It follows from \rlnm{Prj$2$} that $\Gamma \types \pi_2(M) : \CompPure{B},\,\Delta$ and the result follows from \rlnm{SwapL}.
  \item[\rlnm{AppL}] Suppose $\Gamma \types M : B \onlyto A,\,\Delta$ and $\Gamma,\,N:B \types \Delta$.  Then it follows from \rlnm{CompR} that $\Gamma \types N:\CompPure{B},\,\Delta$ and hence, by \rlnm{App}, that $\Gamma \types M\,N:\CompPure{A}$.  Then the result follows from \rlnm{SwapL}.
  \item[\rlnm{FixL}] Suppose $\Gamma,\,M:A \types x:A,\,\Delta$.  Then by \rlnm{CompL} and \rlnm{CompR}, $\Gamma,\,x:\CompPure{A} \types M:\CompPure{A},\,\Delta$.  Then $\Gamma \types \fixtm{x}{M} : \CompPure{A},\,\Delta$ by \rlnm{Fix}, and the result follows from \rlnm{SwapL}.
  \item[\rlnm{MatchL}] Suppose $\Gamma \types M : \bigvee_{i=1}^k p_i\theta_i,\,\Delta$ and, for each $i \in [1,k]$, $\Gamma,\,M:p_i\theta_i,\,N_i:A \types \CompPure{\theta_i},\,\Delta$.  Then, for each $i\in[1,k]$, $\Gamma,\,\theta_i,\,M:p_i\theta_i \types N_i : \CompPure{A},\,\Delta$ by \rlnm{SwapL} and \rlnm{CompR}.  Hence, it follows from \rlnm{Match} that $\Gamma \types \matchtm{M}{|_{i=1}^k p_i \mapsto N_i} : \CompPure{A},\,\Delta$ and the result can be obtained by \rlnm{SwapL}.
\end{description}
\end{proof}
  \section{The Equivalence of the Two- and One-Sided Type Systems, Theorem~\ref{thm:one-two-equivalence}}

\begin{figure}\vspace{1ex}
  \input{figures/simple-typing-rules.tex}
  \caption{One-sided Typing.}\label{fig:simple-typing-rules}
\end{figure}

The full typing rules of the one-sided system are given in Figure~\ref{fig:simple-typing-rules}.  We begin by showing that the one-sided system admits a version of subtyping on the left (only for variables, since we do not allow terms on the left in the one-sided system), a version of complement on the left (again, only for variables) and a version of swap on the left (likewise).

\begin{lemma}
  The following rules are admissible in the system of Figure~\ref{fig:simple-typing-rules}:
  \begin{mathpar}
    \infer[SubL]{
      \Gamma,\,x:B \types \Delta
    }
    {
      \Gamma,\,x:A \types \Delta
    }\;A \subtype B
    \and
    \infer[CompL$_0$]{
      \Gamma \types x:A,\,\Delta
    }
    {
      \Gamma,\,x:\CompPure{A} \types \Delta
    }\;x \notin \dom(\Gamma)
    \and
    \infer[SwapL$_0$]{
      \Gamma \types x:\CompPure{A},\,\Delta
    }
    {
      \Gamma,\,x:A \types \Delta
    }\;x \notin \dom(\Gamma)
  \end{mathpar}
\end{lemma}
\begin{proof}
  We start with \rlnm{SubL}, since this is needed in the proof for \rlnm{CompL$_0$}.  The proof is by induction on $\Gamma,\,x:B \types \Delta$.
  \begin{description}
    \item[\rlnm{Var}] Let $A \subtype B$.  There are two subcases.  If $x:B$ is the principal formula of the rule, so that $\Delta$ is of shape $x:B,\,\Delta'$ then we have $\Gamma,\,x:A \types x:A,\,\Delta'$ by \rlnm{Var} and then $\Gamma,\,x:A \types x:B,\,\Delta'$ by \rlnm{Sub}.  Otherwise, there is some other $y:B \in \Gamma \cap \Delta$ and so $\Gamma,\,x:A \types \Delta$ follows by \rlnm{Var}.
    \item[Otherwise] In all other cases, the result follows immediately from the induction hypotheses.
  \end{description}
  For the admissibility of \rlnm{CompL$_0$} we need to generalise the result to relax the condition on $\Gamma$.  We instead prove that:
  \begin{quotation}
    $\Gamma \types x:A,\,\Delta$ and $(\forall B \subtype A.\:x:B \notin \Gamma)$ implies $\Gamma,\,x:\CompPure{A} \types \Delta$
  \end{quotation} by induction on the derivation of $\Gamma \types x:A,\,\Delta$.
  \begin{description}
    \item[\rlnm{Var}] Suppose $\forall B \subtype A.\,x:B \notin \Gamma$ so that, in particular, $x:A \notin \Gamma$.  Then it cannot be that $x:A$ is the principal formula of the rule instance, and there is some other $y:B \in \Gamma \cap \Delta$ which is.  Consequently, $\Gamma,\,x:\CompPure{A} \types \Delta$ follows by \rlnm{Var}.
    \item[\rlnm{Top},\,\rlnm{Num},\,\rlnm{Atom}]  For each of these axioms (R), $x:A$ cannot be the principal formula of the rule, and so the conclusion follows by (R).
    \item[\rlnm{Comp}] Suppose $\forall B \subtype A.\,x:B \notin \Gamma$.  There are three subcases.
          \begin{itemize}
            \item If $x:A$ is the principal formula of the rule, then $A$ is of shape $\CompPure{C}$ and the rule hypothesis gives $\Gamma,\,x:C \types \Delta$.  Then $\Gamma,\,x:\CompPure*{\CompPure{C}} \types \Delta$ follows by \rlnm{SubL}, which is just $\Gamma,\,x:\CompPure{A} \types \Delta$, as required.
            \item Otherwise, the principal formula is some $y:\CompPure{C} \in \Delta$ and $\Delta$ has shape $y:\CompPure{C},\,\Delta'$.  If $y=x$ and $C \subtype A$, then this judgement has shape $\Gamma \types x:A,\,x:\CompPure{C},\,\Delta'$ and our goal is to show $\Gamma,\, x:\CompPure{A} \types x:\CompPure{C},\,\Delta'$.  Since $C \subtype A$, also $\CompPure{A} \subtype \CompPure{C}$, and the goal may be proven directly:
                  \begin{mathpar}
                    \infer*[left=Sub]{
                      \infer*[left=Var]{
                      }
                      {
                        \Gamma,\,x:\CompPure{A} \types x:\CompPure{A},\,\Delta'
                      }
                    }
                    {
                      \Gamma,\, x:\CompPure{A} \types x:\CompPure{C},\,\Delta'
                    }
                  \end{mathpar}
            \item Otherwise the principal formula is some $y:\CompPure{C} \in \Delta$ and $\Delta$ has shape $y:\CompPure{C},\,\Delta'$ but $y \neq x$ or $C \not\subtype A$.  Hence, there is no typing $x:B \in (\Gamma,\,y:C)$ with $B \subtype A$ and we may apply the induction hypothesis to obtain: $\Gamma,\,y:C,\,x:\CompPure{A} \types \Delta'$.  Thus, $\Gamma,\,x:\CompPure{A} \types \Delta$ follows by another application of \rlnm{Comp}.
          \end{itemize}
    \item[\rlnm{Sub}]  Suppose $\forall B \subtype A.\,x:B \notin \Gamma$.  There are two subcases.
          \begin{itemize}
            \item If $x:A$ is the principal formula of the rule instance, then the induction hypothesis gives some $C \subtype A$ such that, if $\forall B \subtype C.\,x:B \notin \Gamma$, then $\Gamma,\,x:\CompPure{C} \types \Delta$.  To establish the antecedent, suppose that $B \subtype C$ and $x:B \in \Gamma$ and obtain a contradiction with the assumption, since then we would have some $x:B \in \Gamma$ with $B \subtype C \subtype A$.  Therefore, we obtain $\Gamma,\,x:\CompPure{C} \types \Delta$ from the induction hypothesis.  Since $C \subtype A$, $\CompPure{A} \subtype \CompPure{C}$ and thus $\Gamma,\,x:\CompPure{A} \types \Delta$ follows by \rlnm{SubL}.
            \item Otherwise, $\Delta$ is of shape $M:C,\,\Delta'$ and it follows from the induction hypothesis that there is some $D$ such that $D \subtype C$ and $\Gamma,\,x:\CompPure{A} \types M:D,\,\Delta'$.  Then the desired $\Gamma,\,x:\CompPure{A} \types \Delta$ follows by \rlnm{Sub}.
          \end{itemize}
    \item[Otherwise] For every other rule (R), $x:A$ cannot be the principal formula and so the result follows immediately from the induction hypotheses and another application of (R).
  \end{description}
  Finally, we can then derive \rlnm{SwapL$_0$}.  Suppose $x \notin \dom(\Gamma)$ and note that $A \subtype \CompPure*{\CompPure{A}}$:
  \begin{mathpar}
    \infer*[left=SubL]{
      \infer*[left=CompL$_0$]{
        \Gamma \types x:\CompPure{A},\,\Delta
      }
      {
        \Gamma,\,x:\CompPure{\CompPure{A}} \types \Delta
      }
    }
    {
      \Gamma,\,x:A \types \Delta
    }
  \end{mathpar}
\end{proof}

Now we may prove the main equivalence theorem.

Theorem~\ref{thm:one-two-equivalence}:
\begin{quote}
  Let us write $\Gamma \typesA \Delta$ for provability in the system of Figure~\ref{fig:simple-two-sided-rules} and $\Gamma \typesB \Delta$ for provability in the system of Figure~\ref{fig:simple-two-sided-left}.  Then both of the following:
  \begin{enumerate}[(i)]
    \item If $\Gamma \typesA \Delta$ then $\typesB \CompPure{\Gamma},\,\Delta$.
    \item If $\Gamma \typesB \Delta$ then $\Gamma \typesA \Delta$.
  \end{enumerate}
\end{quote}

\begin{proof}
  Part (i) is proven by induction on $\Gamma \typesA \Delta$.
  \begin{description}
    \item[\rlnm{Var}] In this case, $\Delta$ is of shape $x:A,\,\Delta'$ and $\Gamma$ is of shape $\Gamma',\,x:A$.  Hence, $x:A \typesB \CompPure{\Gamma'},\,x:A,\,\Delta'$ by \rlnm{Var} and $\typesB \CompPure{\Gamma},\,\Delta$ by \rlnm{Comp}.
    \item[\rlnm{Top}, \rlnm{Num}, \rlnm{Atom}] In these cases, the result is immediate.
    \item[\rlnm{SubL}] In this case, $\Gamma$ is of shape $\Gamma',\,M:A$.  It follows from the induction hypothesis that there is some type $A \subtype B$ such that $\typesB \CompPure{\Gamma'},\,M:\CompPure{B},\,\Delta$ and the desired conclusion follows from \rlnm{Sub}.
    \item[\rlnm{CompL}] In this case, $\Gamma$ is of shape $\Gamma',\,M:\CompPure{A}$.  It follows from the induction hypothesis that $\typesB \CompPure{\Gamma'},\,M:\CompPure{\CompPure{A}},\,\Delta$, and the conclusion follows by \rlnm{Sub} since $\CompPure{\CompPure{A}} \subtype A$.
    \item[\rlnm{CompR}] In this case, $\Delta$ is of shape $M:\CompPure{A},\,\Delta'$ and it follows from the induction hypothesis that $\typesB \CompPure{\Gamma},\, M:\CompPure{A},\,\Delta'$, as required.
    \item[\rlnm{Abs}] In this case, $\Delta$ is of shape $\abs{x}{M}:A \to B,\,\Delta'$.  We may assume that $x$ does not occur in $\Gamma$ or $\Delta'$. The induction hypothesis gives:
          \begin{enumerate}[(IH)]
            \item $\typesB \CompPure{\Gamma},\,x:\CompPure{A},\,M:B,\,\Delta'$
          \end{enumerate}
          Then, $x:A \typesB \CompPure{\Gamma},\,M:B,\,\Delta'$ by \rlnm{SwapL$_0$} and the result can be obtained by \rlnm{Abs}.
    \item[\rlnm{Fix}] In this case, $\Delta$ is of shape $\fixtm{x}{M}:A,\,\Delta'$.  We may assume that $x$ does not occur in $\Gamma$ or $\Delta'$.  Then the induction hypothesis gives:
          \begin{enumerate}[(IH)]
            \item $\typesB \CompPure{\Gamma},\,x:A,\,M:A,\,\Delta'$
          \end{enumerate}
          Then, $x:A \typesB \CompPure{\Gamma},\,M:A,\,\Delta'$ follows by \rlnm{SwapL$_0$} and the result can then be obtained using \rlnm{Fix}.
    \item[\rlnm{Match}] In this case, $\Delta$ is of shape $\matchtm{M}{|_{i=1}^k p_i \mapsto N_i}:A,\,\Delta'$ and we may assume that no variable in any of the patterns $p_i$ occurs in $\Gamma$ or $\Delta'$.  The induction hypotheses give some substitutions $\theta_i$ whose domain is exactly $\fv(p_i)$ and such that:
          \begin{enumerate}[({IH}1)]
            \item $\typesB \CompPure{\Gamma},\,M:\bigvee_{i=1}^k p_i\theta_i,\,\Delta'$
            \item for each $i \in [1,k]$: $\typesB \CompPure{\Gamma},\,\CompPure{\theta_i},\,M:\CompPure*{p_i\theta_i},\,N_i:A,\,\Delta'$
          \end{enumerate}
          If follows by $\Sigma_{i=1}^k |\fv(p_i)|$ applications of \rlnm{SwapL$_0$} that, for each $i \in [1,k]$, $\theta_i \typesB \CompPure{\Gamma},\,M:\CompPure*{p_i\theta_i},\,N_i:A,\,\Delta'$.  Therefore, the result can be derived by \rlnm{Match}.
    \item[Otherwise] For all other rules \rlnm{R} of the first system, the result follows immediately from the induction hypotheses and the corresponding rule \rlnm{R} of the second system.
  \end{description}
  Part (ii) can be proven by induction on $\Gamma \typesB \Delta$.
  \begin{description}
    \item[\rlnm{Comp}] In this case, $\Delta$ has shape $x:\CompPure{A},\,\Delta'$ the induction hypothesis gives $\Gamma,\,x:A \typesA \Delta'$ from which we obtain the result by \rlnm{CompR}.
    \item[\rlnm{OpE}] In this case, $\Delta$ has shape $M \otimes N : \CompPure{\okty},\,\Delta'$ and it follows from the induction hypothesis that $\Gamma \typesA M:\CompPure{\intty},\,N:\CompPure{\intty},\,\Delta'$.  Then, by two applications of \rlnm{SwapL}, we obtain $\Gamma,\,M:\intty,\,N:\intty \typesA \Delta'$ from which the conclusion follows by \rlnm{OpE}.
    \item[\rlnm{AppE}] In this case, $\Delta$ has shape $M\,N:\CompPure{\okty},\,\Delta'$ and it follows from the induction hypothesis that $\Gamma \typesA M:\CompPure{\funty},\,\Delta'$.  By \rlnm{SwapL}, $\Gamma,\,M:\funty \types \Delta'$ and the conclusion follows by \rlnm{AppE}.
    \item[\rlnm{MatchE}] In this case, $\Delta$ has shape $\matchtm{M}{|_{i=1}^k P_i \mapsto N_i} : \CompPure{\okty},\,\Delta'$ and it follows from the induction hypothesis that $\Gamma \typesA M:\CompPure*{\bigvee_{i=1}^k p_i\theta_\okty},\,\Delta'$.  By \rlnm{SwapL}, $\Gamma,\,M:\bigvee_{i=1}^k p_i\theta_\okty \typesA \Delta'$ and the conclusion follows by \rlnm{MatchE}.
    \item[\rlnm{PrjE$_i$}]  In this case, $\Delta$ has shape $\pi_i(M):\CompPure{\okty},\,\Delta'$ and it follows from the induction hypothesis that $\Gamma \typesA M:\CompPure{\pairty},\,\Delta'$.  By \rlnm{SwapL}, $\Gamma,\,M:\pairty \typesA \Delta'$ and the conclusion follows by \rlnm{PrjE$_i$}.
    \item[Otherwise] For every other rule (R) of the system of Figure~\ref{fig:simple-typing-rules}, the conclusion follows immediately from the induction hypotheses and an application of (R) in the system of Figure~\ref{fig:simple-two-sided-rules}.
  \end{description}
\end{proof}
  \section{Progress and Preservation}
The proofs in this appendix relate to the one-sided system of Figure~\ref{fig:simple-typing-rules}.

\subsection{Preservation}
\begin{lemma}[Weakening]
    \label{lem:weakening}
    If $\Gamma \types \Delta$ is provable, then for any term $Q$ and type $B$, it is also provable that $\Gamma \types Q:B,\,\Delta$.
\end{lemma}
\begin{proof}
    It suffices to show that for any proof $\delta$ of $\Gamma \types \Delta$, we can always construct some proof $\delta'$ of $\Gamma \types Q:B,\, \Delta$ (with arbitrary $Q:B$). We proceed by induction on the height of $\delta$.
    \begin{indproof}
        \indcase{($\height{\delta}=0$)} Here, $\delta$ must be an instance of an axiom rule (i.e., $\textsc{Var}$, $\textsc{Top}$, $\textsc{Num}$, or $\textsc{Atom}$). Thus $\delta$ must be of one of the structures shown in \cref{fig:weakening-r-base-delta}. For every $\delta$ that proves $\Gamma \types \Delta$, we can construct $\delta'$ by inserting $Q:B$ into the right context of the underlying rule instance of $\delta$, such that $\delta'$ proves $\Gamma \types Q:B,\, \Delta$. \cref{fig:weakening-r-base-delta'} shows the corresponding structures of $\delta'$ for all respective structures of $\delta$.
        \begin{figure}[!htbp]
            \centering
            \begin{subfigure}[c]{0.3\textwidth}
                \centering
                \[
                    \begin{array}{c}
                        \infer*[left=\textsc{Var}]{ }{\smashunderbrace{\Gamma',\,x:A}{\Gamma} \types \smashunderbrace{\Principal{x:A},\,\Delta'}{\Delta}} \\[0.6em]
                        \infer*[left=\textsc{Top}]{ }{\Gamma \types \smashunderbrace{\Principal{M:\top},\,\Delta'}{\Delta}}                               \\[0.6em]
                        \infer*[left=\textsc{Num}]{ }{\Gamma \types \smashunderbrace{\Principal{n : \intty},\,\Delta'}{\Delta}}                           \\[0.6em]
                        \infer*[left=\textsc{Atom}]{ }{\Gamma \types \smashunderbrace{\Principal{\atomtm{a} : \atom{a}},\,\Delta'}{\Delta}}
                    \end{array}
                \]
                \subcaption{Proof $\delta$ of $\Gamma \types \Delta$.}
                \label{fig:weakening-r-base-delta}
            \end{subfigure}
            \hspace{0.02\textwidth}
            \raisebox{0.6\height}{$\xrightarrow{\text{Add } Q:B \text{ to the right context}}$}
            \hspace{0.02\textwidth}
            \begin{subfigure}[c]{0.3\textwidth}
                \centering
                \[
                    \begin{array}{c}
                        \infer*[left=\textsc{Var}]{ }{\smashunderbrace{\Gamma',\,x:A}{\Gamma} \types Q:B,\, \smashunderbrace{\Principal{x:A},\,\Delta'}{\Delta}} \\[0.6em]
                        \infer*[left=\textsc{Top}]{ }{\Gamma \types Q:B,\,\smashunderbrace{\Principal{M:\top},\,\Delta'}{\Delta}}                                \\[0.6em]
                        \infer*[left=\textsc{Num}]{ }{\Gamma \types Q:B,\,\smashunderbrace{\Principal{n : \intty},\,\Delta'}{\Delta}}                            \\[0.6em]
                        \infer*[left=\textsc{Atom}]{ }{\Gamma \types Q:B,\,\smashunderbrace{\Principal{\atomtm{a} : \atom{a}},\,\Delta'}{\Delta}}
                    \end{array}
                \]
                \subcaption{Proof $\delta'$ of $\Gamma \types Q:B,\, \Delta$.}
                \label{fig:weakening-r-base-delta'}
            \end{subfigure}
            %\caption{$\exists \delta.\,(\Gamma \types_{\delta} \Delta) \Rightarrow \exists \delta'.\,(\Gamma \types_{\delta'} Q:B,\,\Delta) \quad (\text{when }\height{\delta}=0)$}
            \label{fig:weakening-r-base}
        \end{figure}
        \par Thus, whenever $\Gamma\,\types \Delta$ has a proof of height $0$, $\Gamma\,\types Q:B,\, \Delta$ also has a proof of height $0$.
        \indcase{($\height{\delta}>0$)} Suppose for every proof $\delta$ of $\Gamma \types \Delta$ with height $\leq k\, (k\geq 0)$, we can construct a proof $\delta'$ of $\Gamma \types Q:B,\, \Delta$. Now consider $\height{\delta}=k+1$, then $\delta$ must have the structure:
        \begin{equation*}
            \tag{Structure of $\delta$}
            \label{proof:weakening-r-inductive-delta}
            \infer*[left=\textup{\textsc{R}},right=\normalsize\,($n>0$)]{
                \Gamma,\, \Constituent{\Gamma_{c}^{i}} \types_{\delta_i} \Constituent{\Delta_{c}^{i}},\, \Delta_{r} \quad (0 < i \leq n)
            }{
                \Gamma \types \smashunderbrace{\Principal{P:A},\, \Delta_{r}}{\Delta}
            }
        \end{equation*}
        where \textup{\textsc{R}} is the rule applied at the root step, $P:A$ is the principal formula of the root application and $\Gamma^i_c,\,\Delta^i_c$ are the constituent formulas of the $i^{th}$ premise, that is, those that are distinguished in the rule. Since each immediate subproof $\delta_i$ has height $\height{\delta_i} \leq \height{\delta}-1 = (k+1)-1=k$, it follows by the induction hypothesis that:
        \begin{equation}
            \label{proof:weakening-r-inductive-premises}
            \forall Q\,\forall B.\,\exists \delta'_{i} \text{ such that } \Gamma,\,\Constituent{\Gamma_{c}^{i}} \types_{\delta'_i} \Constituent{\Delta_{c}^{i}},\, Q:B,\, \Delta_{r} \quad (0 < i \leq n)
        \end{equation}
        By applying \textsc{R} to the judgements in \cref{proof:weakening-r-inductive-premises}, we combine all these modified subproofs $\delta'_1,\,\ldots,\, \delta'_n$  and construct:
        \begin{equation*}
            \tag{Structure of $\delta'$}
            \label{proof:weakening-r-inductive-delta'}
            \infer*[left=\textup{\textsc{R}},right=\normalsize\,($n>0$)]{
                \Gamma,\, \Constituent{\Gamma_{c}^{i}} \types_{\delta'_i} \Constituent{\Delta_{c}^{i}},\, Q:B,\,\Delta_{r} \quad (0 < i \leq n)
            }{
                \Gamma \types Q:B,\,\smashunderbrace{\Principal{P:A},\, \Delta_{r}}{\Delta}
            }
        \end{equation*}
        Let $\delta'$ be this proof. It is immediate that $\delta'$ proves $\Gamma \types Q:B,\, \Delta$, which completes the inductive step.
        \par By combining the base and inductive step, we conclude that for every proof (regardless of its height) of $\Gamma \types \Delta$, there exists a corresponding proof of $\Gamma \types Q:B,\, \Delta$ where $Q:B$ is arbitrary. Thus the provability of $\Gamma \types \Delta$ guarantees the provability of $\Gamma \types Q:B,\,\Delta$ for any term $Q$ and type $B$, thereby proving the weakening lemma.
    \end{indproof}
\end{proof}

% \Celia{To decide if to add $\Gamma \nvdash \Delta$ here}
\begin{lemma}[Substitution]
    \label{lem:substitution}
    %\clnotes{chose to include $\Gamma \nvdash \Delta$}
    Suppose $\Gamma,\, x:A \vdash M:B,\, \Delta$, where $x \notin \fv(\Delta)$. If $\Gamma \vdash N:C,\,\Delta$ and $C \subtyping A$,
    then $\Gamma \vdash \Subst{M}{N}{x}:B,\,\Delta$.
\end{lemma}
\begin{proof}
    By induction on the structure of $M$.
\end{proof}

\begin{definition}[Subtyping on Pattern Type Substitutions]
    \label{def:subtyping-on-subst}
    The subtyping relation $\subtyping$ is extended to pattern type substitutions \textit{pointwise}:
    \[
        \theta \subtyping \theta' \quad \defeq \quad \dom(\theta') = \dom(\theta) \;\land\; \forall x\in\dom(\theta').\, \theta(x) \subtyping \theta'(x).
    \]
\end{definition}

\begin{lemma}
    \label{lem:subst-subtyping-on-pattern}
    Let $p$ be a pattern, and let $\theta,\,\theta'$ be pattern type substitutions with $\dom(\theta)=\dom(\theta')=\fv(p)$. Then:
    \[
        p\theta \subtyping p\theta' \quad\Leftrightarrow\quad \theta \subtyping \theta'
    \]
\end{lemma}
\begin{proof}
    Both directions of $\Leftrightarrow$ are proved by induction on on the shape of $p$. The induction for $\Leftarrow$ direction is straightforward. Here we only show the pair case for the $\Rightarrow$ direction.
    \begin{indproof}
        \indcase{$p$ is of shape $(p_1,p_2)$}
        Then from $p\theta \subtyping p\theta'$, we have:
        \[
            (p_1\theta_1, p_2\theta_2) \subtyping (p_1\theta'_1, p_2\theta'_2)
        \]
        , where $\theta_1$, $\theta'_1$ and $\theta_2$, $\theta'_2$ are the restriction of $\theta$ and $\theta'$ on $\fv(p_1)$ and $\fv(p_2)$ respectively.
        By \cref{lem:subtyping} it follows that  $p_1\theta \subtyping p_1\theta'_1$ and $p_2\theta \subtyping p_2\theta'_2$. Then by the induction hypothosis we have $\theta_1 \subtyping \theta'_1$ and $\theta_2 \subtyping \theta'_2$. Since $\dom(\theta)=\dom(\theta')=\fv(p_1) \cup \fv(p_2)$, $\theta \subtyping \theta'$ follows directly.
    \end{indproof}
\end{proof}
% \paragraph{Remark} \cref{def:subtyping-on-subst} demonstrates a dual variance behaviour, similar to that of function subtyping:
% \begin{itemize}
%     \item \textbf{Contravariant domain}: $\dom(\theta') \subseteq \dom(\theta)$ requires $\theta$ to support \textit{at least} all variables in $\theta'$.
%     \item \textbf{Covariant assignments}: $\theta(x) \subtyping \theta'(x)$ requires $\theta$ to assign \textit{potentially more precise} types than $\theta'$.
% \end{itemize}
% This mirrors function subtyping: $A \to B \subtyping A' \to B'$ if $A' \subtyping A$ (contravariant input) and $B \subtyping B'$ (covariant output). As a result, a function of type $A \to B$ (analogous to $\theta$) can be used wherever $A' \to B'$ (analogous to $\theta'$) is expected.

\begin{definition}[Coherence of Term and Pattern Type Substitutions under Contexts]
    \label{def:coherence-subst-ctx}
    % \clnotes{chose to include $\Gamma \nvdash \Delta$}
    Let $\sigma$ be a term substitution (mapping term variables to terms), $\theta$ a pattern type substitution (mapping term variables to types), and $\Gamma$, $\Delta$ left and right contexts respectively.  We define the \textit{coherence relation} $\Gamma \types \sigma\tyfit\theta,\, \Delta$ as:
    \[
        \Gamma \types \sigma\tyfit\theta,\, \Delta \quad \defeq \quad
        %\Gamma \nvdash \Delta \;\land\;
        \dom(\theta) = \dom(\sigma)
        \;\land\;
        \forall x\in \dom(\theta).\,\Gamma \types \sigma(x):\theta(x),\,\Delta
    \]
\end{definition}

\paragraph{Remark}
The implication of $\Gamma \types \sigma\tyfit\theta,\, \Delta$ depends critically on the contexts:
\begin{enumerate}
    \item When $\Gamma \nvdash \Delta$, the relation $\Gamma \vdash \sigma\tyfit\theta, \Delta$ says $\sigma$ and $\theta$ share the same domain, and for each variable $x$ in the domain, the type $\theta(x)$ assigned to $\sigma(x)$ is \textit{meaningful} under the given contexts.
    \item By contrast, when $\Gamma \vdash \Delta$, the relation $\Gamma \vdash \sigma\tyfit\theta, \Delta$ is satisfied vacuously for any $\sigma$ and $\theta$ with $\dom(\sigma) = \dom(\theta)$. This follows directly from \cref{lem:weakening} (\textsc{Weakening}).  For example, when $\Gamma \vdash \Delta$ we can have $\Gamma \types [(\abs{y}{y})/x]\tyfit[\intty/x],\, \Delta$.
\end{enumerate}
\begin{lemma}[Simultaneous Substitution]
    %\clnotes{chose to include $\Gamma \nvdash \Delta$}
    \label{lem:simultaneous-substitution}
    Suppose $\Gamma,\,\theta' \types N:A,\,\Delta$, where $\dom(\theta') \cap \fv(\Delta)=\varnothing$.
    %(which implies $\Gamma \nvdash \Delta$)
    If $\theta \subtyping \theta'$ and $\Gamma \types \sigma\tyfit\theta,\,\Delta$, then $\Gamma \types N\sigma:A,\,\Delta$.
\end{lemma}
\begin{proof}
    By \cref{lem:substitution} (\textsc{Substitution}) and induction on the size of $\dom(\theta')$.
\end{proof}

\begin{lemma}[Irrelevant Assumption]
  \label{lem:irrelevance}
  If $\Gamma,\,x:A \types \Delta$ and $x \notin \fv(\Delta)$, then $\Gamma \types \Delta$.
\end{lemma}
\begin{proof}
The proof is by induction on the derivation.
\begin{description}
  \item[\rlnm{Var}] In this case, the assumption implies that the distinguished typing $x:A$ is not principal.  Therefore, $\Gamma \types \Delta$ follows by \rlnm{Var}.
  \item[\rlnm{Top},\rlnm{Num},\rlnm{Atom}] In these cases, $x:A$ cannot be the principal formula and so the result follows immediately by application of the same rule. 
  \item[\rlnm{Abs},\rlnm{Fix}] These cases are similar, so we do only the former.  In this case, $\Delta$ is of shape $\abs{y}{M}:B \to C,\,\Delta'$ with the principal formula distinguished, and we may assume $y \neq x$.  Assume $x \notin \fv(\Delta)$, so in particular $x \notin \fv(M)$.  Then it follows from the induction hypothesis that $\Gamma,\,y:B \types M:C,\,\Delta'$, and the result follows by \rlnm{Abs}.
  \item[\rlnm{Sub}] In this case, the result follows immediately from the induction hypothesis and \rlnm{Sub}.
  \item[\rlnm{Comp}] In this case, $\Delta$ has shape $y:\CompPure{B},\,\Delta'$ and it follows from the induction hypothesis that $\Gamma,\,y:B \types \Delta'$ from which the result follows by \rlnm{Comp}.
  \item[\rlnm{Pair},\rlnm{BinOp},\rlnm{RelOp},\rlnm{App},\rlnm{BinOpE},\rlnm{RelOpE},\rlnm{PairE},\rlnm{AppE},\rlnm{MatchE}] Each of these cases follows immediately from the induction hypotheses and application of the same rule.
  \item[\rlnm{Match}] In this case, $\Delta$ has shape $\matchtm{M}{|_{i=1}^k p_i \mapsto N_i} : B,\,\Delta'$ and we may assume that the free variables of each $p_i$ are disjoint from $x$.  Hence, for each $i$, $x$ is not in the domain of $\theta_i$.  Thus, we may apply the induction hypotheses in each premise to obtain $\Gamma \types M : \bigvee_{i=1}^k p_i\theta_i,\,\Delta'$ and, for each $i$, $\Gamma,\,\theta_i \types M:\CompPure*{p_i\theta_i},\,N_i:B,\,\Delta'$.  Then the result follows by \rlnm{Match}.   
\end{description}
\end{proof}

\begin{lemma}[Irrelevant Conclusions]
  If $\Gamma \types \Delta_1,\,\Delta_2$ and $\fv(\Delta_1) \cap \fv(\Delta_2) = \emptyset$ and $\Gamma \nvdash \Delta_2$ then already $\Gamma \types \Delta_1$. 
\end{lemma}
\begin{proof}
  The proof is by induction on the derivation.
  \begin{description}
    \item[\rlnm{Var},\rlnm{Top},\rlnm{Num},\rlnm{Atom}] Suppose $\fv(\Delta_1) \cap \fv(\Delta_2) = \emptyset$ and $\Gamma \nvdash \Delta_2$.  By the latter, it must be that the principal formula of the rule is in $\Delta_1$, from which the result follows by \rlnm{Var}.
    \item[\rlnm{Abs},\rlnm{Fix}] These cases are similar, so we do only the former. Suppose $\fv(\Delta_1) \cap \fv(\Delta_2) = \emptyset$ and $\Gamma \nvdash \Delta_2$.  There are two subcases.  
    \begin{itemize}
      \item If the principal formula is in $\Delta_1$, then $\Delta_1$ has shape $\abs{x}{M}:A \to B,\,\Delta_1'$ and we may assume, in particular, that $x \notin \fv(\Gamma) \cup \fv(\Delta_2)$.  Then the premise may be split as $M:B,\,\Delta_1'$ and $\Delta_2$.  It follows that $\fv(M:B,\,\Delta_1') \cap \fv(\Delta_2) = \emptyset$ and, by the contrapositive of \cref{lem:irrelevance} (\rlnm{Irrelevant Assumption}), that $\Gamma,\,x:A \nvdash \Delta_2$.  Therefore, the induction hypothesis yields $\Gamma,\,x:A \types M:B,\,\Delta_1'$ and the result follows by \rlnm{Abs}.
      \item Otherwise, $\Delta_2$ is of shape $\abs{x}{M}:A \to B,\,\Delta_2'$ and we may assume that $x \notin \fv(\Gamma) \cup \fv(\Delta_1)$.  Then the premise may be split as $\Delta_1$ and $M:B,\,\Delta_2'$.  It follows that $\fv(\Delta_1) \cap \fv(M:B,\,\Delta_2') = \emptyset$ and we must have $\Gamma,\,x:A \nvdash M:B,\,\Delta_2'$ since, otherwise, $\Gamma \types \Delta_2$ would follow by \rlnm{Abs}.  Therefore, the result follows immediately from the induction hypothesis.
    \end{itemize}
  \item[\rlnm{Comp}] There are two subcases.
    \begin{itemize}
      \item If the principal formula is in $\Delta_1$, then $\Delta_1$ has shape $x:\CompPure{A},\Delta'$.  Since $\Gamma \nvdash \Delta_2$ and $x \notin \fv(\Delta_2)$, it follows from the contrapositive of Lemma \cref{lem:irrelevance} \rlnm{Irrelevant Assumption} that $\Gamma,\,x:A \nvdash \Delta_2$ and thus $\Gamma,\,x:A \types \Delta'$ follows from the induction hypothesis.  From this, $\Gamma \types \Delta_1$ follows by \rlnm{Comp}.
      \item Otherwise, $\Delta_2$ is of shape $x:\CompPure{A},\Delta'$.  Since $\Gamma \nvdash \Delta_2$, it follows from the contrapositive of \rlnm{Comp} that $\Gamma,\,x:A \nvdash \Delta'$.  Hence, $\Gamma,\,x:A \types \Delta_1$ follows from the induction hypothesis.  Since $x \notin \fv(\Delta_1)$, $\Gamma \types \Delta_1$ follows by Lemma~\cref{lem:irrelevance} \rlnm{Irrelevant Assumption}.
    \end{itemize}
  \item[\rlnm{Pair},\rlnm{BinOp},\rlnm{RelOp},\rlnm{App},\rlnm{BinOpE},\rlnm{RelOpE},\rlnm{PairE},\rlnm{AppE},\rlnm{MatchE}] These cases are similar, so we consider only the first.  There are two subcases.
    \begin{itemize}
      \item If the principal formula is in $\Delta_1$, then $\Delta_1$ has shape $(M,N):(A,B),\Delta'$.  It follows immediately from the induction hypotheses that $\Gamma \types M:A,\Delta'$ and $\Gamma \types N:B,\Delta'$, and the result then follows by \rlnm{Pair}.
      \item Otherwise, the principal formula is in $\Delta_2$, which therefore has shape $(M,N):(A,B),\Delta'$.  It follows from the contrapositive of \rlnm{Pair} that either $\Gamma \nvdash M:A,\Delta'$ or $\Gamma \nvdash N:B,\Delta'$.  In either case, it follows immediately from the respective induction hypothesis that $\Gamma \types \Delta_1$.
    \end{itemize}
  \item[\rlnm{Match}] There are two subcases.
    \begin{itemize}
      \item If the principal formula is in $\Delta_1$, then $\Delta_1$ has shape $\matchtm{M}{|_{i=1}^k p_i \mapsto N_i}:A,\Delta'$ with $\fv(p_i) \cap \fv(\Delta_2) = \emptyset$.  Hence, the first induction hypothesis yields $\Gamma \types M:\bigvee_{i=1}^k p_i \theta_i,\Delta'$.  By the disjointness of the pattern variables, we have that, for each $i$, the domain of $\theta_i$ is disjoint from the free variables of $\Delta_2$.  Therefore, it follows by repeated application the contrapositive of Lemma~\ref{lem:irrelevance} \rlnm{Irrelevant Assumption} that $\Gamma,\,theta_i \nvdash \Delta_2$.  Consequently, the induction hypothesis yields, for each $i$, $\Gamma,\,\theta_i \types M:\CompPure*{p_i\theta_i},N_i:A,\Delta'$.  Then the result follows by \rlnm{Match}.
      \item Otherwise, $\Delta_2$ has shape $\matchtm{M}{|_{i=1}^k p_i \mapsto N_i}:A,\Delta'$ with $\fv(p_i) \cap \fv(\Delta_1) = \emptyset$.  It follows from the contrapositive of \rlnm{Match} that either $\Gamma \nvdash M : \bigvee_{i=1}^k p_i\theta_i,\,\Delta'$ or $\Gamma,\,\theta_i \nvdash M:\CompPure*{p_i\theta_i},N_i:A,\Delta'$ for some $i$.  In the first case, the result follows from the induction hypothesis.  In the second case, it follows from the induction hypothesis that, for the witnessing $i$, $\Gamma,\,\theta_i \types \Delta_1$.  Then, due to the disjointness of the pattern variables (domain of $\theta_i$), repeated application of Lemma~\cref{lem:irrelevance} \rlnm{Irrelevant Assumption} yields the desired result.
    \end{itemize}
  \end{description}
\end{proof}

\begin{lemma}[Inversion]\label{lem:inversion}
    Let $\Gamma \types Q:A,\, \Delta$ with $Q$ closed and $A \tyneq \topty$. Then, \textbf{either} $\Gamma \types \Delta$ is itself provable, \textbf{or} $Q$ satisfies the following syntax-directed conditions:
    \begin{enumerate}
        \item If $Q$ is of shape $\atomtm{a}$, then $\atomtm{a} \subtyping A$.
        \item If $Q$ is of shape $\pn{n}$, then $\intty \subtyping A$.
        \item If $Q$ is of shape $\abs{x}{M}$, then $\exists B,\,C$ such that $B \to C \subtyping A$ and $\Gamma, x:B \types M:C,\, \Delta$.
        \item If $Q$ is of shape $\app{M}{N}$, then either:
              \begin{enumerate}[\protect\circled{\arabic*}]
                  \item $\exists B, C$ such that $C\subtyping A$, $\Gamma \types M:B \to C,\,\Delta$ and $\Gamma \types N:B,\,\Delta$
                  \item $\CompPure{\okty}\subtyping A$ and $\Gamma \types M : \CompPure{\funty},\,\Delta$.
              \end{enumerate}
              % \item If $Q$ is of shape $n$, then $\intty \subtyping A$.
              % \item If $Q$ is of shape $\atom{a}$ then $\atom{a} \subtyping A$.
        \item If $Q$ is of shape $(M_1,M_2)$, then either:
              \begin{enumerate}[\protect\circled{\arabic*}]
                  \item $\exists A_1,\,A_2$ such that $(A_1,A_2) \subtyping A$ and $\Gamma \types M_1: A_1,\,\Delta$ and $\Gamma \types M_2 : A_2,\,\Delta$
                  \item $\CompPure{\okty} \subtyping A$ and $\Gamma \types M_1: \CompPure{\okty},\,M_2:\CompPure{\okty},\,\Delta$.
              \end{enumerate}
              % \item If $Q$ is of shape $\pi_{1}(M)$, then either:
              %       \begin{enumerate}[\protect\circled{\arabic*}]
              %           \item $\Gamma \types M:(A, \okty),\,\Delta$
              %           \item $\CompPure{\okty}\subtyping A$ and $\Gamma \types M : \CompPure{\pairvalty},\,\Delta$.
              %       \end{enumerate}
              % \item If $Q$ is of shape $\pi_{2}(M)$, then either:
              %       \begin{enumerate}[\protect\circled{\arabic*}]
              %           \item $\Gamma \types M:(\okty,A),\,\Delta$
              %           \item $\CompPure{\okty}\subtyping A$ and $\Gamma \types M : \CompPure{\pairvalty},\,\Delta$.
              %       \end{enumerate}
        \item If $Q$ is of shape $\fixtm{x}{M}$, then $\exists A'$ such that $A' \subtyping A$ and $\Gamma,\,x:A' \types M:A',\,\Delta$.
        \item If $Q$ is of shape $\opPrim{M}{N}$, then either:
              \begin{enumerate}[\protect\circled{\arabic*}]
                  \item $\intty \subtyping A$ and $\Gamma \types M : \intty,\,\Delta$ and $\Gamma \types N : \intty,\,\Delta$
                  \item $\CompPure{\okty} \subtyping A$ and $\Gamma \types M: \CompPure{\intty},\,N:\CompPure{\intty},\,\Delta$.
              \end{enumerate}
        \item If $Q$ is of shape $\relOpPrim{M}{N}$, then either:
              \begin{enumerate}[\protect\circled{\arabic*}]
                  \item $\boolty \subtyping A$ and $\Gamma \types M : \intty,\,\Delta$ and $\Gamma \types N : \intty,\,\Delta$
                  \item $\CompPure{\okty} \subtyping A$ and $\Gamma \types M: \CompPure{\intty},\,N:\CompPure{\intty},\,\Delta$.
              \end{enumerate}
        \item If $Q$ is of shape $\matchtm{M}{\mid_{i=1}^k p_i \mapsto N_i}$ then either:
              \begin{enumerate}[\protect\circled{\arabic*}]
                  \item $\Gamma \types M:\textstyle\bigvee_{i=1}^k p_i\theta_i,\,\Delta$ and $\forall i\in[1,k].\,\Gamma,\,\theta_i \types M:\CompPure*{p_i\theta_i},\,N_i:A,\,\Delta$ and, for all $i \in [1,k]$ and $x \in \dom(\theta_i)$, $\theta_i(x) \subtype \okty$.
                  \item $\CompPure{\okty}\subtyping A$ and $\Gamma \types M:\CompPure*{\textstyle\bigvee_{i=1}^k p_i\thetaok{p_i}},\,\Delta$.
              \end{enumerate}
    \end{enumerate}
\end{lemma}
\begin{proof}
    Suppose $\Gamma \types Q:A,\, \Delta$ with $Q$ closed, $A \tyneq \topty$, and suppose $\Gamma \nvdash \Delta$.  Then it follows from \rlnm{Irrelevant Conclusions} that already there is a proof of $\Gamma \types Q : A$.  We prove that $Q$ satisfies conditions 1 -- 9, with $\Delta$ empty, by induction on the structure of that proof.
    \begin{itemize}
      \item If the proof is concluded by \rlnm{Var}, $Q$ must be a variable, which contradicts the assumption that $Q$ was closed.
      \item If the proof is concluded by \rlnm{Sub}, we note that the induction hypothesis applies to the premise, of shape $\types Q : D$ with $D \subtype A$, and this gives that $Q:D$ satisfies the conditions given, with $\Delta$ empty.  Then observe that all the conditions require only that the given type $A$ is a supertype (of some type that depends on the particular case), and thus, since the syntax-directed conditions hold, with $\Delta$ empty, for $Q:D$ they also hold for $Q:A$, since $D \subtype A$.  
      \item Otherwise, the conditions, with $\Delta$ empty, are immediate by inspection of the rule.
    \end{itemize}
    Finally, if the conditions 1--9 hold with $\Delta$ empty, then they also hold for the $\Delta$ given in the statement by repeated application of \rlnm{Weakening}.

\end{proof}

\begin{lemma}
    \label{lem:fun-to-div}
    For any types $A$ and $B$, $\abs{x}{\divtm} \in \mng{A \to B}$.
\end{lemma}
\begin{proof}
    By the definition of $\mng*{\cdot}$, for any type $B$, $\divtm \in \mng*{B}$ holds vacuously.
    Moreover, as $\fv(\divtm)=\varnothing$, we have $\Subst{\divtm}{N}{x} \eqsyntac \divtm$ for any term $N$. It follows that $\Subst{\divtm}{N}{x} \in \mng*{B}$ for any type $B$ and any term $N$. By \cref{def:type-meanings}, this guarantees that $\abs{x}{\divtm} \in \mng{A \to B}$ for any types $A$ and $B$.
\end{proof}

\begin{lemma}
    \label{lem:all-values-typable}
    Let $V$ be a value.  Then $\types V : \okty$, in particular:
    \begin{itemize}
        \item If $V$ is a numeral, then $\types V : \intty$.
        \item If $V$ is an atom $\atom{a}$, then $\types V : \atom{a}$.
        \item If $V$ is an abstraction, then $\types V : \funty$.
        \item If $V$ is a pair, then $\types V : \pairvalty$.
    \end{itemize}
\end{lemma}
\begin{proof}
    Straightforward induction on $V$.
\end{proof}

\begin{lemma}[Cut Admissibility for Values]
    \label{lem:cut}
    Let $V$ be a value, and $A$ and $B$ types with $A \subtyping \CompPure{B}$. If both $\Gamma \types V:A,\,\Delta$ and $\Gamma \types V:B,\,\Delta$ hold, then $\Gamma \types \Delta$ follows.
\end{lemma}
\begin{proof}
    First we consider the case that one of $A$ or $B$ is subtype equivalent to $\topty$, let us suppose it is $A$ because the other case is symmetrical -- $A \subtype \CompPure{B}$ iff $B \subtype \CompPure{A}$.  Then it follows that $\CompPure{B} \tyeq \topty$.  In such a case, we can always find another type $A' \tyneq \topty$ such that $A' \subtype \CompPure{B}$ and $\Gamma \types V:A',\,\Delta$, and thus we may conduct the proof assuming wlog $A \tyneq \topty$.  To see this, observe that, by Lemma~\ref{lem:all-values-typable}, subtyping and weakening, we obtain $\Gamma \types V : \okty,\,\Delta$ and, of course, $\okty \subtype \topty \tyeq  \CompPure{B}$.

    So, assume $A \tyneq \topty$, $B \tyneq \topty$, $\Gamma \types V:A,\,\Delta$ and $\Gamma \types V:B,\,\Delta$, while $\Gamma \nvdash \Delta$ (i.e., the negation of \cref{lem:cut}). We show that this always leads to a contradiction by induction on the shape of $V$.
    \begin{indproof}

        \indcase{$V$ is of shape $\atomtm{a}$}
        Since $\Gamma\nvdash\Delta$, applying \cref{lem:inversion} (\textsc{Inversion}) to $\Gamma \types \atomtm{a}:A,\,\Delta$ and $\Gamma \types \atomtm{a}:B,\,\Delta$, we obtain $\atom{a} \subtyping A$ and $\atom{a} \subtyping B$. By transitivity of subtyping, from $A \subtyping \CompPure{B}$ and $\atom{a} \subtyping A$, we derive $\atom{a} \subtyping \CompPure{B}$.
        \par Now, applying \cref{thm:subtyping-soundness} (\textsc{Soundness of Subtyping}) to both $\atom{a} \subtyping B$ and $\atom{a} \subtyping \CompPure{B}$, we conclude $\mng{\atom{a}} \subseteq \mng{B}$ and $\mng{\atom{a}} \subseteq \mng{\CompPure{B}}$. This implies $\mng{\atom{a}} \subseteq \mng{B} \cap \mng{\CompPure{B}}$. However, by \cref{def:type-meanings}, $\mng{B}\cap\mng{\CompPure{B}}=\mng{B}\cap\mng{\topty} \setminus \mng{B}=\varnothing$. Thus $\mng{\atom{a}} \subseteq \varnothing$. But since $\mng{\atom{a}} = \makeset{\atomtm{a}}\neq\varnothing$ (again by \cref{def:type-meanings}), we arrive at a contradiction.

        \indcase{$V$ is of shape $\pn{n}$}
        Since $\Gamma\nvdash\Delta$, applying \cref{lem:inversion} (\textsc{Inversion}) to $\Gamma \types \pn{n}:A,\,\Delta$ and $\Gamma \types \pn{n}:B,\,\Delta$ yields $\intty \subtyping B$ and $\intty \subtyping \CompPure{B}$. This leads to a contradiction by an argument analogous to the case for $V \eqsyntac \atomtm{a}$.

        \indcase{$V$ is of shape $(W_1,W_2)$}
        Here we have $\Gamma \types (W_1,W_2):A,\,\Delta$ and $\Gamma \types (W_1,W_2):B,\,\Delta$, where both $W_1$ and $W_2$ are values.
        Since $\Gamma \nvdash \Delta$, by \cref{lem:inversion} (\textsc{Inversion}), there exists $A_1,\,A_2$ and $B_1,\,B_2$, such that:
        \[
            \begin{array}{c@{\;\land\;}c@{\;\land\;}c}
                (A_1,A_2) \subtyping A & \Gamma \types W_1:A_1,\,\Delta & \Gamma \types W_2:A_2,\,\Delta \\
                (B_1,B_2) \subtyping B & \Gamma \types W_1:B_1,\,\Delta & \Gamma \types W_2:B_2,\,\Delta
            \end{array}
        \]
        By transitivity of subtyping, from $(A_1, A_2) \subtyping A$ and $A \subtyping \CompPure{B}$, we have $(A_1, A_2) \subtyping \CompPure{B}$. Applying the subtyping rule \textsc{CompR} to this yields $B \subtyping \CompPure{(A_1, A_2)}$. Combining this with $(B_1, B_2) \subtyping B$ via transitivity gives $(B_1, B_2) \subtyping \CompPure{(A_1, A_2)}$. By \cref{lem:subtyping}, this implies either $B_1 \subtyping \CompPure{A_1}$ or $B_2 \subtyping \CompPure{A_2}$.
        \par If $B_1 \subtyping \CompPure{A_1}$, the induction hypothesis applied to $\Gamma \types W_1:A_1,\,\Delta$ and $\Gamma \types W_1:B_1,\,\Delta$ implies $\Gamma \types \Delta$. Similarly, if $B_2 \subtyping \CompPure{A_2}$, the induction hypothesis applied to $\Gamma \types W_2 : A_2, \Delta$ and $\Gamma \types W_2 : B_2, \Delta$ again gives $\Gamma \types \Delta$. In either case, we conclude $\Gamma \types \Delta$, contradicting the assumption that $\Gamma \nvdash \Delta$.

        \indcase{$V$ is $\abs{x}{M}$}
        Here we have $\Gamma \types \abs{x}{M} : A,\, \Delta$ and $\Gamma \types \abs{x}{M} : B,\, \Delta$.
        Since $\Gamma \nvdash \Delta$, by \cref{lem:inversion} (\textsc{Inversion}):
        \[
            \begin{array}{c@{\text{ such that: }}c@{\;\land\;}c}
                \exists A_1, A_2 & A_1 \to A_2 \subtyping A & \Gamma, x:A_1 \types M:A_2, \Delta \\
                \exists B_1, B_2 & B_1 \to B_2 \subtyping B & \Gamma, x:B_1 \types M:B_2, \Delta
            \end{array}
        \]
        By transitivity of subtyping, from $A \subtyping \CompPure{B}$ and $A_1 \to A_2 \subtyping A$, we have $A_1 \to A_2 \subtyping \CompPure{B}$. Applying the subtyping rule \textsc{CompR} to this yields $B \subtyping A_1 \to A_2$. Combining this with $B_1 \to B_2 \subtyping B$ via transitivity gives $B_1 \to B_2 \subtyping \CompPure{(A_1 \to A_2)}$.
        \par By \cref{thm:subtyping-soundness} (\textsc{Soundness of Subtyping}), $B_1 \to B_2 \subtyping \CompPure{(A_1 \to A_2)}$ implies $\mng{B_1 \to B_2} \subseteq \mng{\CompPure{(A_1 \to A_2)}}$. Since $\mng{\CompPure{(A_1 \to A_2)}}=\topty \setminus \mng{A_1 \to A_2}$ (by \cref{def:type-meanings}), this requires $\mng{B_1 \to B_2} \cap \mng{A_1 \to A_2} = \varnothing$.
        \par However, \cref{lem:fun-to-div} states that $\abs{x}{\divtm} \in \mng{C \to D}$ for all types $C$ and $D$. In particular, $\abs{x}{\divtm} \in \mng{B_1 \to B_2} \cap \mng{A_1 \to A_2}$. This contradicts our earlier conclusion that $\mng{B_1 \to B_2} \cap \mng{A_1 \to A_2} = \varnothing$.
    \end{indproof}
\end{proof}

\begin{lemma}
    \label{lem:type-of-value}
    Let $V$ be a value. If $\Gamma \types V:A,\,\Delta$,
    then $\exists B\,(B\subtyping \okty \land B\subtyping A)$ such that $\Gamma \types V:B,\,\Delta$.
\end{lemma}
\begin{proof}
    Note, that if $A \tyeq \topty$ then the result follows immediately by Lemma~\ref{lem:all-values-typable}.  Otherwise, we start with a simple case analysis on whether $\Gamma \types \Delta$.
    \begin{itemize}
        \item Suppose $\Gamma \types \Delta$. Then it follows by \cref{lem:weakening} (\textsc{Weakening}) that $\Gamma \types V:\botty,\,\Delta$. By \cref{thm:ortholattice}, we have both $\botty \subtyping \okty$ and $\botty \subtyping A$. Thus when taking $B \eqsyntac \botty$, we have $\Gamma \types V:B,\,\Delta$ where $B \subtyping \okty \land B\subtyping A$.
        \item Otherwise, $\Gamma \nvdash \Delta$. In this case, we proceed by induction on the shape of $V$.
              \begin{indproof}
                  \indcase{$V$ is of shape $\atomtm{a}$}
                  Since $\Gamma \nvdash \Delta$, applying \cref{lem:inversion} (\textsc{Inversion}) to $\Gamma \types \atomtm{a}:A,\,\Delta$ yields $\atomtm{a} \subtyping A$. Moreover, $\Gamma \types \atomtm{a}:\atomtm{a},\,\Delta$ holds by the typing rule \textsc{Atom}. Thus, let $B \eqsyntac \atom{a}$, we have $\Gamma \types \atom{a}:B,\,\Delta$ where $B \subtyping \okty \land B\subtyping A$.

                  \indcase{$V$ is of shape $\pn{n}$}
                  Since $\Gamma \nvdash \Delta$, applying \cref{lem:inversion} (\textsc{Inversion}) to $\Gamma \types \pn{n}:A,\,\Delta$ yields $\intty \subtyping A$. Moreover, $\Gamma \types \pn{n}:\intty,\,\Delta$ holds by the typing rule \textsc{Num}. Thus, let $B \eqsyntac \intty$, we have $\Gamma \types \pn{n}:B,\,\Delta$ where $B \subtyping \okty \land B\subtyping A$.

                  \indcase{$V$ is of shape $(W_1,W_2)$}
                  Since $\Gamma \nvdash \Delta$, applying {\cref{lem:inversion}} (\textsc{Inversion}) to $\Gamma \types (W_1,W_2):A,\,\Delta$, we obtain:
                  \begin{circledlist}
                      \item $\exists A_1,\,A_2$ such that $(A_1,A_2) \subtyping A$ and $\Gamma \types W_1: A_1,\,\Delta$ and $\Gamma \types W_2 : A_2,\,\Delta$
                      \item $\CompPure{\okty} \subtyping A$ and $\Gamma \types W_1: \CompPure{\okty},\,W_2:\CompPure{\okty},\,\Delta$.
                  \end{circledlist}

                  As $W_1$ and $W_2$ are values, applying the induction hypothesis to each component, we get:
                  \begin{align*}
                      \exists B_1\,(B_1 \subtyping \okty \land B_1 \subtyping A_1) \text{ such that } \Gamma \types W_1 : B_1, \Delta \\
                      \exists B_2\,(B_2 \subtyping \okty \land B_2 \subtyping A_2) \text{ such that } \Gamma \types W_2 : B_2, \Delta
                  \end{align*}
                  Then $\Gamma \types W_i : \okty,\,\Delta$ for $i \in \makeset{1,2}$ by the typing rule \textsc{Sub}, and $\Gamma \types (W_1, W_2) : (B_1, B_2), \Delta$ by the typing rule \textsc{Pair}.
                  Assume \circled{2} hold.
                  \par Applying rule \textsc{Pair} to $B_1 \subtyping \okty$ and $B_2 \subtyping \okty$ yields $(B_1, B_2) \subtyping (\okty, \okty)$.
                  Since $\pairty \eqsyntac (\okty, \okty)$ by definition, we have $(B_1, B_2) \subtyping \pairvalty$.
                  %   \ToDo{To be adjusted according to the revised version of subtyping rules}
                  Moreover, since $\pairvalty \subtyping \intty \vee \pairvalty \vee \funty \vee \atomty$ (by the subtyping rule \textsc{Ok}) and $\intty \vee \pairvalty \vee \funty \vee \atomty \subtyping \okty$ (by the subtyping rule \textsc{UnionR}), transitivity gives $\pairvalty \subtyping \okty$. Conbining this with $(B_1, B_2) \subtyping \pairvalty$ gives $(B_1, B_2) \subtyping \okty$.
                  Assume \circled{2} is true.
                  \par Meanwhile, applying the subtyping rule \textsc{Pair} to $B_1 \subtyping A_1$ and $B_2 \subtyping A_2$ yields $(B_1, B_2) \subtyping (A_1, A_2)$. Since we also have $(A_1, A_2) \subtyping A$ from inversion, transitivity gives $(B_1, B_2) \subtyping A$. Therefore, when taking $B \eqsyntac (B_1, B_2)$, we have $B \subtype \okty \land B \subtype A$ and $\Gamma \types (W_1,W_2):B,\,\Delta$.

                  \indcase{$V$ is of shape $\abs{x}{M}$}
                  Since $\Gamma \nvdash \Delta$, applying {\cref{lem:inversion}} (\textsc{Inversion}) to $\Gamma \types (W_1,W_2):A,\,\Delta$, we obtain:
                  \[
                      \exists B',\,C' \text{ such that } B' \to C' \subtyping A \text{ and }\Gamma,\,x:B' \types M:C',\,\Delta
                  \]
                  \par Applying the typing rule \textsc{Abs} to $\Gamma,\,x:B' \types M:C',\,\Delta$ yields $\Gamma \types \abs{x}{M}:B' \to C',\,\Delta$. From \cref{thm:ortholattice}, we have $\botty \subtyping B'$ and $C' \subtyping \topty$. By the subtyping rule \textsf{Fun}, it follows that $B' \to C' \subtyping \funty$ (recall that $\textsf{Fun} \eqsyntac \botty \to \topty$ by definition). By transitivity, combining $B' \to C' \subtyping \funty$ with $\funty \subtyping \okty$, $B' \to C'\subtyping A$ gives $B' \to C' \subtyping \okty$.
                  Since we also know $B' \to C'\subtyping A$ from inversion, when taking $B \eqsyntac B' \to C'$, we have $\Gamma \types (W_1,W_2):B,\,\Delta$ where $B \subtype \okty \land B \subtype A$.
              \end{indproof}
    \end{itemize}
\end{proof}

\begin{lemma}[Canonical Forms of Value Types]
    \label{lem:canonical-forms-value}
    Let $V$ be a (closed) value, and let $\Gamma$ and $\Delta$ be contexts such that $\Gamma \nvdash \Delta$.
    \begin{enumerate}
        \item If $\Gamma \types V: \intty,\,\Delta$, then $V \eqsyntac \pn{n}$ for some $n \in \mathbb{Z}$.
        \item If $\Gamma \types V: \atom{a},\,\Delta$, then $V \eqsyntac \atomtm{a}$.
        \item If $\Gamma \types V: \pairty,\,\Delta$, then $V \eqsyntac (V_1,V_2)$ for some values $V_1$ and $V_2$, and $\Gamma \types V: \pairvalty,\,\Delta$.
        \item If $\Gamma \types V: A \to B,\,\Delta$, then $V \eqsyntac \abs{x}{M}$ from some term $M$.
    \end{enumerate}
\end{lemma}
\begin{proof}
    We present the proof (1) and (3) of \cref{lem:canonical-forms-value}; the remaining (2) and (4) follow by similar reasoning.
    \settowidth{\titleindent}{(2)}
    \begin{mylist}
        \item[(1)] Assume that $\forall n \in \mathbb{Z}.\,V \neqsyntac \pn{n}$ and $\Gamma \types V: \intty,\,\Delta$ where $\Gamma \nvdash \Delta$ (i.e., the negation of \cref{lem:canonical-forms-value}.(1)). We show that this always leads to a contradiction, by induction on the structure of the closed value $V$.
        \begin{indproof}
            \indcase{$V$ is of shape $\atomtm{a}$}
            By \cref{lem:inversion} (Inversion) applied to $\Gamma \types \atomtm{a}:\intty$ with $\Gamma \nvdash \Delta$, we obtain $\atom{a} \subtyping \intty$. Then \cref{thm:subtyping-soundness} implies $\mng{\atom{a}} \subseteq \mng{\intty}$. However, this contradicts \cref{def:type-meanings} since $\atomtm{a} \in \mng{\atom{a}}$ while $\atomtm{a} \in \mng{\intty}$.

            \indcase{$V$ is of shape $(V_1,V_2)$}
            By \cref{lem:inversion} (\textsc{Inversion}) on $\Gamma \types (V_1,V_2):\intty, \,\Delta$ with $\Gamma \nvdash \Delta$, either:
            \begin{circledlist}
                \item $\exists A_1,\,A_2$ such that $(A_1,A_2) \subtyping \intty$ and $\Gamma \types V_1: A_1,\,\Delta$ and $\Gamma \types V_2 : A_2,\,\Delta$
                \item $\CompPure{\okty} \subtyping \intty$ and $\Gamma \types V_1: \CompPure{\okty},\,V_2:\CompPure{\okty},\,\Delta$.
            \end{circledlist}
            \settowidth{\titleindent}{\textsf{Case}\circled{2}.}
            \begin{mypar}
                \item[Case \circled{1}.]
                By \cref{thm:subtyping-soundness}, $(A_1, A_2) \subtyping \intty$ implies $\mng{(A_1, A_2)} \subtyping \mng{\intty}$. However, by \cref{def:type-meanings}, $\mng{(A_1, A_2)}$ contains pairs only, while $\mng{\intty}$ contains $\pn{n}\, (n \in \mathbb{Z})$ only. This contradicts $\mng{(A_1, A_2)} \subtyping \mng{\intty}$, which would require pairs belong to $\mng{\intty}$.
                \item[Case \circled{2}.] By \cref{thm:subtyping-soundness}, $\CompPure{\okty} \subtyping \intty$ implies $\mng{\CompPure{\okty}} \subtyping \mng{\intty}$. However, by \cref{def:type-meanings} $\atomtm{a} \in \mng{\CompPure{\okty}}$ while $\atomtm{a}
                    \notin \mng{\intty}$, contradicting $\mng{\CompPure{\okty}} \subtyping \mng{\intty}$.
            \end{mypar}
            \indcase{$V$ is of shape $\abs{x}{M}$}
            Applying \cref{lem:inversion} (\textsc{Inversion}) to \begin{added}$\Gamma \types \abs{x}{M}:\intty, \,\Delta$\end{added} with $\Gamma \nvdash \Delta$, we obtain:
            \begin{equation*}
                \exists B,\,C \text{ such that } B \to C \subtyping \intty \land \Gamma, x:B \types M:C,\, \Delta
            \end{equation*}
            By \cref{thm:subtyping-soundness}, $B \to C \subtyping \intty$ implies $\mng{B \to C} \subtyping \mng{\intty}$. However, while \cref{lem:fun-to-div} gives $\abs{x}{\divtm} \in \mng{B \to C}$, \cref{def:type-meanings} gives $\abs{x}{\divtm} \notin \mng{\intty}$, contradicting $\mng{B \to C} \subtyping \mng{\intty}$.
        \end{indproof}
        Since all possible cases lead to contradictions, we conclude that the original assumption must be false. Therefore, if $\Gamma \types V : \intty$ and $\Gamma \nvdash \Delta$, then $V$ must be of the form $\pn{n}$ for some integer $n$.
        \item[(3)] Assume that $V \neqsyntac (V_1,V_2)$ for any values $V_1$, $V_2$ and $\Gamma \types V: \pairty,\,\Delta$ where $\Gamma \nvdash \Delta$ (i.e., the negation of \cref{lem:canonical-forms-value}.(3)). We show that this always leads to a contradiction, by induction on the structure of the closed value $V$.
        \begin{indproof}
            \indcase{$V$ is of shape $\pn{n}$ for some $n \in \mathbb{Z}$}
            Applying \cref{lem:inversion} (\textsc{Inversion}) to $\Gamma \types \pn{n}:\pairty, \,\Delta$ with $\Gamma \nvdash \Delta$ yields $\pn{n} \subtyping \pairty$. Since $\pairty \eqsyntac (\topty,\topty)$, by \cref{thm:subtyping-soundness}, it follows that $\mng{\intty} \subtyping \mng{(\topty,\topty)}$. This would repuire $\pn{n} \in \mng{(\topty,\topty)}$, contradicting \cref{def:type-meanings}.

            \indcase{$V$ is of shape $\atomtm{a}$}
            Applying \cref{lem:inversion} (\textsc{Inversion}) to $\Gamma \types \atomtm{a}:\pairty, \,\Delta$ with $\Gamma \nvdash \Delta$ yields $\atomtm{a} \subtyping \pairty$. Since $\pairty \eqsyntac (\topty,\topty)$, by \cref{thm:subtyping-soundness}, it follows that $\mng{\atom{a}} \subtyping \mng{(\topty,\topty)}$. This would repuire $\atomtm{a} \in  \mng{(\topty,\topty)}$, contradicting \cref{def:type-meanings}.

            \indcase{$V$ is of shape $\abs{x}{M}$}
            Applying \cref{lem:inversion} (\textsc{Inversion}) to \begin{added}$\Gamma \types \abs{x}{M}:\pairty, \,\Delta$\end{added} with $\Gamma \nvdash \Delta$, we obtain:
            \begin{equation*}
                \exists B,\,C \text{ such that } B \to C \subtyping \pairty \land \Gamma, x:B \types M:C,\, \Delta
            \end{equation*}
            By \cref{thm:subtyping-soundness} with $\pairty \eqsyntac (\topty,\topty)$, $B \to C \subtyping \pairty$ implies $\mng{B \to C} \subtyping \mng{(\topty,\topty)}$. But \cref{lem:fun-to-div} gives $\abs{x}{\divtm} \in \mng{B \to C}$, which would require $\abs{x}{\divtm} \in \mng{(\topty,\topty)}$. This contradicts \cref{def:type-meanings}.
        \end{indproof}
        Since all possible cases lead to contradictions, we conclude that the original assumption must be false. Thus, if $\Gamma \types V:\pairty,\,\Delta$ and $\Gamma \nvdash \Delta$, then $V \eqsyntac (V_1, V_2)$ for some values $V_1$ and $V_2$. Then, by \cref{lem:inversion} applied to $\Gamma \types (V_1, V_2):\pairty,\,\Delta$ with $\Gamma \nvdash \Delta$, either:
        \begin{circledlist}
            \item $\exists A_1,\,A_2$ such that $(A_1,A_2) \subtyping \pairty$ and $\Gamma \types V_1: A_1,\,\Delta$ and $\Gamma \types V_2 : A_2,\,\Delta$
            \item $\CompPure{\okty} \subtyping \pairty$ and $\Gamma \types S': \CompPure{\okty},\,M:\CompPure{\okty},\,\Delta$.
        \end{circledlist}
        By \cref{thm:subtyping-soundness} (\textsc{Soundness of Subtyping}), $\CompPure{\okty} \subtyping \pairty$ implies $\mng{\CompPure{\okty}} \subtyping {\pairty}$, which contradicts \cref{def:type-meanings}. Thus case \circled{2} is impossible, and case \circled{1} must be true.
        By \cref{lem:type-of-value}, $\Gamma \types V_1: A_1,\,\Delta$ and $\Gamma \types V_2 : A_2,\,\Delta$ implies $\Gamma \types V_1: \okty,\,\Delta$ and $\Gamma \types V_2 : \okty,\,\Delta$. Applying the typing rule \textsc{Pair} then gives $\Gamma \types (V_1, V_2):(\okty,\okty),\,\Delta$. Given $V \eqsyntac (V_1, V_2)$ and $\pairvalty \eqsyntac (\okty,\okty)$,we conclude $\Gamma \types V:\pairvalty,\,\Delta$.
    \end{mylist}
\end{proof}

\begin{lemma}
    \label{lem:value-substituition-to-pattern-yield-value}
    For any pattern $p$, $p\thetaok{p_i} \subtype \okty$ holds.
\end{lemma}
\begin{proof}
    By induction on the shape of pattern $p$.
\end{proof}

\begin{lemma}
    \label{lem:value-not-matching-pattern}
    Let $V$ be a closed value and $p$ a pattern, we have: if $V \neqsyntac p\sigma$ for any $\sigma$ then $\types V:\CompPure{(p\thetaok{p})}$.
\end{lemma}
\begin{proof}
    By induction on the shape of $V$.
    \begin{indproof}
        \indcase{$V$ is of shape $\atomtm{a}$}
        Given $V \neqsyntac p\sigma$, there must be $p \neqsyntac x$ (otherwise $V \eqsyntac p\sigma$ by taking $\sigma \eqsyntac[\atomtm{a}/x]$) and $p \neqsyntac \atomtm{a}$ (otherwise $V \eqsyntac p\sigma$ for all $\sigma$). Thus, it has to be $p \eqsyntac (p_1,p_2)$ for some patterns $p_1$ and $p_2$. Consequently, $(p\thetaok{p}) \eqsyntac (p_1,p_2)\,\thetaok{p} \eqsyntac (p_1\thetaok{p},p_2\thetaok{p})$. By \cref{lem:value-substituition-to-pattern-yield-value}, we have $p_1\thetaok{p} \subtyping \okty$ and $p_2\thetaok{p} \subtyping \okty$. Thus the subtyping rule \textsc{Pair} gives:
        \[
            (p_1\thetaok{p},p_2\thetaok{p}) \subtyping (\okty,\okty)
        \]
        ,i.e., $(p\thetaok{p}) \subtyping \pairvalty$. \Cref{thm:ortholattice} then gives $\CompPure{\pairvalty} \subtyping \CompPure{(p\thetaok{p})}$.
        \par By \cref{lem:all-values-typable}, we have $\types \atomtm{a} :\atom{a}$. Since $\atom{a} \subtyping \CompPure{\pairvalty}$ (by the subtyping rule \textsc{Disj}), transitivity then gives $\atom{a} \subtyping \CompPure{(p\thetaok{p})}$. Therefore, applying the typing rule \text{Sub} to $\types \atomtm{a} :\atom{a}$ with $\atom{a} \subtyping \CompPure{(p\thetaok{p})}$, we get $\types V :\CompPure{(p\thetaok{p})}$.

        \indcase{$V$ is of shape $\pn{n}$}
        By an argument analogous to the case where $V$ is of shape $\atomtm{a}$.

        \indcase{$V$ is of shape $\abs{x}{M}$}
        By an argument analogous to the case where $V$ is of shape $\atomtm{a}$.

        \indcase{$V$ is of shape $(W_1,W_2)$}
        If $p$ is not of shape $(p_1, p_2)$, the lemma follows by a similar reasoning with other cases. Otherwise, $p\sigma \eqsyntac (p_1, p_2)\,\sigma \eqsyntac (p_1\sigma,p_2\sigma)$. Since $\forall \sigma.\,(W_1,W_2) \neqsyntac (p_1\sigma,p_2\sigma)$, either $\forall \sigma.\, W_1 \neqsyntac p_1\sigma$ or $\forall \sigma.\, W_2 \neqsyntac p_2\sigma$. In the following we prove for the former case, and the latter case follows by an analogous reasoning.
        \par If $\forall \sigma.\, W_1 \neqsyntac p_1\sigma$, the induction hypothesis applied to $W_1$ gives $\types W_1 : \CompPure{p_1\thetaok{p_1}}$, and \cref{lem:all-values-typable} gives $\types W_2 :\okty$. Applying the typing rule \textsc{Pair} then yields $\types (W_1,W_2): (\CompPure{p_1\thetaok{p_1}}, \okty)$.
        \par Next we show $(\CompPure{p_1\thetaok{p_1}}, \okty) \subtyping \CompPure{(p_1,p_2)\thetaok{p}}$. For any type $B$, by the subtyping rule \textsc{PairC} we have:
        \[
            (\CompPure{p_1\thetaok{p_1}}, \topty) \lor (\okty, \CompPure{B}) \subtyping \CompPure{(p_1\thetaok{p_1},B)}
        \]
        and by the subtyping rule \textsc{UnionR} and \textsc{Top} we have:
        \[(\CompPure{p_1\thetaok{p_1}}, \okty) \subtyping (\CompPure{p_1\thetaok{p_1}}, \topty) \lor (\okty, \CompPure{B})\]
        Thus by transitivity, for any type $B$:
        \begin{equation}
            \label{no-label-1}
            (\CompPure{p_1\thetaok{p_1}}, \okty) \subtyping \CompPure{(p_1\thetaok{p_1},B)}
        \end{equation}
        Moreover, as $\fv(p_1) \subseteq \fv(p)$, $\thetaok{p_1}$ is the restriction of $\thetaok{p}$ on $\fv(p_1)$. Thus $p\thetaok{p} \eqsyntac (p_1, p_2)\,\thetaok{p} \eqsyntac (p_1\thetaok{p}, p_2\thetaok{p}) \eqsyntac (p_1\thetaok{p_1}, p_2\thetaok{p})$, i.e.,
        \[
            p\thetaok{p} \eqsyntac ({p_1\thetaok{p_1}},p_2\thetaok{p})
        \]
        Let $B \eqsyntac  p_2\thetaok{p}$ in \cref{no-label-1} to get:
        \begin{equation}
            \label{no-label-2}
            (\CompPure{p_1\thetaok{p_1}}, \okty) \subtyping \CompPure{(p_1\thetaok{p_1}, p_2\thetaok{p})} \equiv \CompPure{p\thetaok{p}}
        \end{equation}
        Finally, applying the typing rule \textsc{Sub} to $\types (W_1,W_2): (\CompPure{p_1\thetaok{p_1}}, \okty)$ with \cref{no-label-2}, we conclude that $\types V : \CompPure{p\thetaok{p}}$.
    \end{indproof}
\end{proof}

% \begin{lemma}[Normal Forms]
%     The (closed) normal forms may equivalently be described by the following grammar:
%     \[
%         \begin{array}{lll}
%             U \;\Coloneqq\; \pn{n} & \mid \atomtm{a} \mid \abs{x}{M} \mid (U_1,\,U_2) &                                                                                  \\
%                                    & \mid (U,\,M)                                     & \textup{\textsf{when $U$ is not of shape $V$}}                                   \\
%                                    & \mid U \otimes M \mid \relOpPrim{U}{M}                                & \textup{\textsf{when $U$ is not of shape $\pn{n}$}}                              \\
%                                    & \mid \pn{n} \otimes U \mid \relOpPrim{\pn{n}}{U}                           & \textup{\textsf{when $U$ is not of shape $\pn{n}$}}                              \\
%             %    & \mid \pi_i(U)                                    & \textup{\textsf{when $U$ is not of shape $(V,W)$}}                               \\
%                                    & \mid U\,M                                        & \textup{\textsf{when $U$ is not of shape $\abs{x}{M}$}}                          \\
%                                    & \mid \matchtm{U}{\mid_{i=1}^k p_i \mapsto N_i}   & \textup{\textsf{when $\forall i \in [1,k].\ \forall \sigma.\ U \neq p_i\sigma$}}
%         \end{array}
%     \]
% \end{lemma}
% \begin{proof}
%     By structural inspection of terms and analysis of the reduction relation.
% \end{proof}

For the following, it is useful to have a more explicit description of the possible shapes of stuck terms.

\begin{lemma}[Stuck Terms]
    \label{lem:stuck-term}
    The (closed) stuck terms may equivalently be described by the following grammar:
    \[
        \begin{array}{lll}
            S,\,T \;\Coloneqq\; (S,\,M) & \mid (V,\,S) \mid \opPrim{S}{M} \mid \relOpPrim{S}{M} \mid S\,M \mid \opPrim{\pn{n}}{S} \mid \relOpPrim{\pn{n}}{S} &                                                                       \\
                                        & \mid \opPrim{V}{M} \mid \relOpPrim{V}{M} \mid \opPrim{\pn{n}}{V} \mid \relOpPrim{\pn{n}}{V}                        & \text{when $V$ is not of shape $\pn{n}$}                              \\
            % & \mid    \pi_i(U)                                 & \text{when $U$ is not of shape $(V,W)$}                               \\
                                        & \mid  V\,M                                                                                                         & \text{when $V$ is not of shape $\abs{x}{M}$}                          \\
                                        & \mid   \matchtm{S}{|_{i=1}^k p_i \mapsto N_i}                                                                                                                                              \\
                                        & \mid   \matchtm{V}{\mid_{i=1}^k p_i \mapsto N_i}                                                                   & \text{when $\forall i \in [1,k].\ \forall \sigma.\ V \neq p_i\sigma$}
        \end{array}
    \]
\end{lemma}

\begin{lemma}
    \label{lem:type-of-stuck-term}
    For any stuck term $S$, we can derive the typing judgement $\types S : \CompPure{\okty}$.
\end{lemma}
\begin{proof}
    By induction on the shape of $S$.

    \begin{indproof}

        \indcase{$S$ is of shape $(S', M)$}
        Applying the induction hypothesis to $S'$ yields $\types S : \CompPure{\okty}$. By \cref{lem:weakening} (\textsc{Weakening}), it follows that $\types S : \CompPure{\okty},\, M:\CompPure{\okty}$, to which we apply the typing rule \textsc{PairE} to derive $\types (S', M): \CompPure{\okty}$.

        \indcase{$S$ is of shape $(V, S')$}
        By an analogous argument with the case where $S$ is of shape $(S', M)$.

        \indcase{$S$ is of shape $\opPrim{S'}{M}$}
        Applying the induction hypothesis to $S'$ yields $\types S': \CompPure{\okty}$. Moreover, from the subtyping rule \textsc{Ok}, we have $\intty \vee \pairvalty \vee \funty \vee \atomty \subtype \okty$. We also have $\intty \subtyping \intty \vee \pairvalty \vee \funty \vee \atomty$ from the subtyping rule \textsc{UnionR}. Then, transitivity gives $\intty \subtyping \okty$. By Lemma~\ref{thm:ortholattice}, this implies $\CompPure{\okty} \subtyping \CompPure{\intty}$.
        \par Consequently, applying the typing rule \textsc{Sub} to $\types S': \CompPure{\okty}$ gives $\types S': \CompPure{\intty}$. By \cref{lem:weakening} (\textsc{Weakening}), it follows that $\types S': \CompPure{\intty},\,M: \CompPure{\intty}$, to which we apply the typing rule \textsc{OpE} to derive $\types \opPrim{S'}{M}: \CompPure{\okty}$.

        \indcase{$S$ is of shape $\opPrim{V}{M}$, where $\forall n \in \mathbb{Z}.\,V \neqsyntac \pn{n}$}
        It follows from Lemma~\ref{lem:all-values-typable} that $\types V : \CompPure{\intty}$.  By \rlnm{Weakening}, $\types V: \CompPure{\intty},\,M:\CompPure{\intty}$. Thus, applying the typing rule \textsc{OpE} yields $\types \opPrim{V}{M} : \CompPure{\okty},\,\Delta$ as desired.

        % By the typing rule \textsc{Top}, the typing judgement $\types \opPrim{V}{M} : \topty$ always holds. Moreover, note that $\Gamma \types \Delta$ is never derivable when both $\Gamma$ and $\Delta$ are empty. Thus, we can apply \cref{lem:inversion} (\textsc{Inversion}) to $\types \opPrim{V}{M} : \topty$ (with $\cdot \nvdash \cdot$) to obtain that either:
        % \begin{circledlist}
        %     \item $\intty \subtyping A$ and $\types V : \intty$ and $\types M : \intty$
        %     \item $\CompPure{\okty} \subtyping A$ and $\types V: \CompPure{\intty},\,M:\CompPure{\intty}$.
        % \end{circledlist}
        % \par Assume \circled{1} holds, then $\types V : \intty,\,\Delta$. By \cref{lem:canonical-forms-value}, it follows that $V \eqsyntac \pn{n}$ for some $n \in \mathbb{Z}$. Since our assumption requires $\forall n \in \mathbb{Z}.\,V \neqsyntac \pn{n}$, \circled{1} is impossible. Consequently, \circled{2} must hold.
        % \par From \circled{2}, $\types V: \CompPure{\intty},\,M:\CompPure{\intty}$. Thus, applying the typing rule \textsc{OpE} yields $\types \opPrim{V}{M} : \CompPure{\okty},\,\Delta$ as desired.

        \indcase{$S$ is of shape $\opPrim{\pn{n}}{S'}$}
        By an analogous argument with the case where $S$ is of shape $\opPrim{S'}{M}$.

        \indcase{$S$ is of shape $\opPrim{\pn{n}}{V}$, where $V \neqsyntac \pn{n}$ for any $n \in \mathbb{Z}$}
        By an analogous argument with the case where $S$ is of shape $\opPrim{V}{M}$ and $\forall n \in \mathbb{Z}.\,V \neqsyntac \pn{n}$.

        \indcase{$S$ is of shape $\app{S'}{M}$}
        By the induction hypothesis, $\types S' : \CompPure{\okty}$. Moreover, from the subtyping rule \textsc{Ok}, we have $\intty \vee \pairvalty \vee \funty \vee \atomty \subtype \okty$. We also have $\funty \subtyping \intty \vee \pairvalty \vee \funty \vee \atomty$ from the subtyping rule \textsc{UnionR}. Then, transitivity gives $\funty \subtyping \okty$. By the subtyping rule \textsc{CompR}, this implies $\CompPure{\okty} \subtyping \CompPure{\funty}$. Thus, applying the typing rule \textsc{App} to $\types S' : \CompPure{\funty}$ gives $\types \app{S'}{M} : \CompPure{\okty}$.

        \indcase{$S$ is of shape $\app{V}{M}$, where $V$ is not of shape $\abs{x}{M}$}
        Then it follows from Lemma~\ref{lem:all-values-typable} that $\types V: \CompPure{\funty}$, to which we apply the typing rule \textsc{App} to conclude $\types \app{V}{M}:\CompPure{\okty}$.
        % By the typing rule \textsc{Top}, the typing judgement $\types \app{V}{M}: \topty$ always holds. By \cref{lem:inversion} (\textsc{Inversion}) on $\Gamma \types \app{V}{M}: A,\,\Delta$, we obtain that either:
        % \begin{circledlist}
        %     \item $\exists B, C$ such that $C\subtyping A$, $\types V:B \to C$ and $\types M:B$
        %     \item $\CompPure{\okty} \subtyping A$ and $\types V : \CompPure{\funty}$.
        % \end{circledlist}
        % \par Assume \circled{1} holds, then $\types V:B \to C$. By \cref{lem:canonical-forms-value}.(4), this requires $V \eqsyntac \abs{x}{M$} for some term $M$. Since this contradicts our assumption that $V$ is not an abstraction, \circled{1} is impossible. Consequently, \circled{2} must hold.
        %     \par From \circled{2}, $\types V: \CompPure{\funty}$, to which we apply the typing rule \textsc{App} to conclude $\types \app{V}{M}:\CompPure{\okty}$.

        \indcase{$S$ is of shape $\matchtm{S'}{|_{i=1}^k p_i \mapsto N_i$}}
        By the induction hypothesis, $\types S' : \CompPure{\okty}$ holds. From \cref{lem:value-substituition-to-pattern-yield-value}, $p_i\thetaok{p_i} \subtyping \okty$ holds for any $i \in [1,k]$. Consequently, \textsc{UnionL} gives $\textstyle\bigvee_{i=1}^k p_i\thetaok{p_i} \subtyping \okty$. Thus Lemma~\ref{thm:ortholattice} gives $\CompPure{\okty} \subtyping \CompPure*{\textstyle\bigvee_{i=1}^k p_i\thetaok{p_i}}$.
        \par Therefore, applying the typing rule \textsc{Sub} to $\types S' : \CompPure{\okty}$ yields $\types S': \CompPure{\textstyle\bigvee_{i=1}^k p_i\thetaok{p_i}}$. By the typing rule \textsc{MatchE}, it follows that $\types \matchtm{S'}{|_{i=1}^k p_i \mapsto N_i} : \CompPure{\okty}$.

        \indcase{$S$ is of shape $\matchtm{V}{|_{i=1}^k p_i \mapsto N_i$}, where $\forall i.\, \sigma.\, V \neqsyntac p_i\sigma$}
        By \cref{lem:value-not-matching-pattern}, since $\forall i.\, \sigma.\, V \neqsyntac p_i\sigma$, it follows that $\forall i.\, \types V:\CompPure{(p_i\thetaok{p_i})}$ for any $i$.  If $k=1$ then the result follows immediately by subtyping and \rlnm{MatchE}.  Otherwise, $k>1$ and, by disjointness of patterns, it must be that each pattern is either a pair or an atom.   Consideration of the proof of Lemma~\ref{lem:type-of-value} shows that there is some type $B$ (e.g. for $V$ a numeral, this can be taken to be $\intty$, for $V$ a function this may be taken to be any arrow, and so on) such that $\types V:B$ and $B \subtype \CompPure{(p_i\thetaok{p_i})}$.  Therefore $B \subtype \CompPure*{\bigvee_{i=1}^k p_i\thetaok{p_i}}$ and thus $\types V : \CompPure*{\bigvee_{i=1}^k p_i\thetaok{p_i}}$.  Therefore, the result is obtained by \rlnm{MatchE}.

    \end{indproof}

\end{proof}

\begin{lemma}
    \label{lem:subtyping}
    The subtyping relation satisfies the following invariants:
    \begin{enumerate}
        \item If $A \to B \subtyping A' \to B'$, then $A' \subtyping A$ and $B \subtyping B'$.
        \item If $A \subtyping \bigvee_{i=1}^k B_i$, and $A$ is \emph{not} a complement or union, then $\exists j\,(1 \leq j \leq k)$ such that $A\subtyping B_j$.
        \item If $\textstyle\bigvee_{i=1}^k B_{i} \subtyping A$, then $\forall j\,(1 \leq j \leq k)$, $B_j \subtyping A$.
        \item If $(A,B) \subtyping (A',B')$, then  $A \subtyping A'$ and $B\subtyping B'$.
        \item If $(A,B) \subtyping \CompPure{(A',B')}$, then $A\subtyping \CompPure{A'}$ or $B\subtyping \CompPure{B'}$.
        \item If $(A,B) \subtyping \okty$, then $A\subtyping \okty$ and $B\subtyping \okty$.
    \end{enumerate}
\end{lemma}

\begin{proof}
    By the equivalence of $\subtyping$ and $\subtypeA$, for the proof we will use $\subtypeA$ and the rules from \cref{fig:sub-alt}.
    \begin{enumerate}
        \item We consider which rule could have been used to derive $A \to B \subtypeA A' \to B'$.
              \begin{description}
                  \item[\rlnm{Refl}] We have $A \to B \subtypeA A \to B$, and by \rlnm{Refl}, $A \subtypeA A$ and $B \subtypeA B$ as required.
                  \item[\rlnm{Fun}] Immediate from the premises.
                  \item[Otherwise] The conclusion could not have been derived by any other subtyping rule.
              \end{description}
        \item By inspection of the subtyping rules, we see the only rule that could have been used to derive the conclusion is \rlnm{UnionR}, whose premise is exactly the result required.
        \item We proceed by induction on the derivation of $\textstyle\bigvee_{i=1}^k B_{i} \subtypeA A$.
              \begin{description}
                  \item[\rlnm{Refl}] In this case, we have $\bigvee_{i=1}^k B_i \subtypeA \bigvee_{i=1}^k B_i$. By \rlnm{Refl}, $\forall j \in [1,k], B_j \subtypeA B_j$. Then, by applying \rlnm{UnionR} $k$ times, we have $\forall j \in [1,k], B_j \subtypeA \bigvee_{i=1}^k B_{i}$ as required.
                  \item[Otherwise] For every other rule (R) that could have derived the conclusion, the result follows immediately from the induction hypotheses and $k$ applications of (R).
              \end{description}
        \item We consider which rule could have been used to derive $(A,B) \subtypeA (A', B')$.
              \begin{description}
                  \item[\rlnm{Refl}] By \rlnm{Refl}, $A \subtypeA A$ and $B \subtypeA B$ as required.
                  \item[\rlnm{Pair}] Immediate from the premises.
                  \item[Otherwise] The conclusion could not have been derived by any other subtyping rule.
              \end{description}
        \item We consider which rule could have been used to derive $(A,B) \subtypeA \CompPure{(A',B')}$.
              \begin{description}
                  \item[\rlnm{CompRPair}] Not possible since $(A,B)$ is disjoint from $\intty, \funty, \atomty$, so the premise cannot be true.
                  \item[\rlnm{CompRPairL}] Immediate from premise.
                  \item[\rlnm{CompRPairR}] Immediate from premise.
                  \item[Otherwise] The conclusion could not have been derived by any other subtyping rule.
              \end{description}
        \item The only rule that could have been used to derive $(A,B) \subtypeA \okty$ is \rlnm{OkR} with the premise $(A,B) \subtypeA \pairvalty$. The only rule we could have used to derive the premise is either \rlnm{Refl}, in which case $A = \okty, B = \okty$ (since $\pairvalty = (\okty, \okty)$), so it follows by \rlnm{Refl} that $A \subtypeA \okty, B \subtypeA \okty$; or \rlnm{Pair}, in which case the result is immediate from the premises.
    \end{enumerate}
\end{proof}

For the following, it is useful to have a more explicit description of normal forms:

\begin{lemma}[Normal Forms]
    The (closed) normal forms may equivalently be described by the following grammar:
    \[
        \begin{array}{lll}
            U \;\Coloneqq\; \pn{n} & \mid \atomtm{a} \mid \abs{x}{M} \mid (U_1,\,U_2) &                                                                                  \\
                                   & \mid (U,\,M)                                     & \textup{\textsf{when $U$ is not of shape $V$}}                                   \\
                                   & \mid U \otimes M \mid \relOpPrim{U}{M}           & \textup{\textsf{when $U$ is not of shape $\pn{n}$}}                              \\
                                   & \mid \pn{n} \otimes U \mid \relOpPrim{\pn{n}}{U} & \textup{\textsf{when $U$ is not of shape $\pn{n}$}}                              \\
            %    & \mid \pi_i(U)                                    & \textup{\textsf{when $U$ is not of shape $(V,W)$}}                               \\
                                   & \mid U\,M                                        & \textup{\textsf{when $U$ is not of shape $\abs{x}{M}$}}                          \\
                                   & \mid \matchtm{U}{\mid_{i=1}^k p_i \mapsto N_i}   & \textup{\textsf{when $\forall i \in [1,k].\ \forall \sigma.\ U \neq p_i\sigma$}}
        \end{array}
    \]
\end{lemma}
\begin{proof}
    By structural inspection of terms and analysis of the reduction relation.
\end{proof}

\begin{lemma}[Cut Admissibility of Normal Forms]
    \label{lem:cut-normal-form}
    Let $U$ be a closed normal form, and $A$ and $B$ types with $A \subtyping \CompPure{B}$. If both $\Gamma \types U:A,\,\Delta$ and $\Gamma \types U:B,\,\Delta$ hold, then $\Gamma \types \Delta$ follows.
\end{lemma}
\begin{proof}
    First we observe that it suffices to consider the case in which $A \tyneq \topty$ and $B \tyneq \topty$.  To see this, suppose that one of $A$ or $B$ is subtype equivalent to $\topty$, let us suppose it is $A$ because the other case is symmetrical.  Then it follows that $\CompPure{B} \tyeq \topty$.  By Lemma~\ref{lem:all-values-typable} and Lemma~\ref{lem:type-of-stuck-term}, we obtain $\types U : A'$ where $A'$ is either the type $\okty$ or the type $\CompPure{\okty}$.  Since $\okty \subtype \topty \tyeq \CompPure{B}$ and $\CompPure{\okty} \subtype \topty \tyeq \CompPure{B}$, this $A'$ can be used in place of $A$.

    So, assume $A \tyneq \topty$, $B \tyneq \topty$, $\Gamma \types U:A,\,\Delta$ and $\Gamma \types U:B,\,\Delta$, while $\Gamma \nvdash \Delta$ (i.e., the negation of \cref{lem:cut}). We show that this always leads to a contradiction by induction on the shape of $U$.
    \begin{indproof}

        \indcase{$U$ is of shape $\pn{n}$}
        Since $\Gamma\nvdash\Delta$, applying \cref{lem:inversion} (\textsc{Inversion}) to $\Gamma \types \pn{n}:A,\,\Delta$ and $\Gamma \types \pn{n}:B,\,\Delta$ yields $\intty \subtyping B$ and $\intty \subtyping \CompPure{B}$. This leads to a contradiction by an argument analogous to the case for $U \eqsyntac \atomtm{a}$.

        \indcase{$U$ is of shape $\atomtm{a}$}
        Since $\Gamma\nvdash\Delta$, applying \cref{lem:inversion} (\textsc{Inversion}) to $\Gamma \types \atomtm{a}:A,\,\Delta$ and $\Gamma \types \atomtm{a}:B,\,\Delta$, we obtain $\atom{a} \subtyping A$ and $\atom{a} \subtyping B$. By transitivity of subtyping, from $A \subtyping \CompPure{B}$ and $\atom{a} \subtyping A$, we derive $\atom{a} \subtyping \CompPure{B}$.
        \par Now, applying \cref{thm:subtyping-soundness} (\textsc{Soundness of Subtyping}) to both $\atom{a} \subtyping B$ and $\atom{a} \subtyping \CompPure{B}$, we conclude $\mng{\atom{a}} \subseteq \mng{B}$ and $\mng{\atom{a}} \subseteq \mng{\CompPure{B}}$. This implies $\mng{\atom{a}} \subseteq \mng{B} \cap \mng{\CompPure{B}}$. However, by \cref{def:type-meanings}, $\mng{B}\cap\mng{\CompPure{B}}=\mng{B}\cap\mng{\topty} \setminus \mng{B}=\varnothing$. Thus $\mng{\atom{a}} \subseteq \varnothing$. But since $\mng{\atom{a}} = \makeset{\atomtm{a}}\neq\varnothing$ (again by \cref{def:type-meanings}), we arrive at a contradiction.

        \indcase{$U$ is of shape $(M_1,M_2)$}
        Here we have $\Gamma \types (M_1,M_2):A,\,\Delta$ and $\Gamma \types (M_1,M_2):B,\,\Delta$. By \cref{lem:inversion} (\textsc{Inversion}) on $\Gamma \types (M_1,M_2):A,\,\Delta$ with $\Gamma \nvdash \Delta$, either:
        \begin{circledlist}
            \item $\exists A_1,\,A_2$ such that $(A_1,A_2) \subtyping A$ and $\Gamma \types M_1: A_1,\,\Delta$ and $\Gamma \types M_2 : A_2,\,\Delta$
            \item $\CompPure{\okty} \subtyping A$ and $\Gamma \types M_1: \CompPure{\okty},\,M_2:\CompPure{\okty},\,\Delta$.
        \end{circledlist}
        Similarly, by \cref{lem:inversion} (\textsc{Inversion}) on $\Gamma \types (M_1,M_2):B,\,\Delta$ with $\Gamma \nvdash \Delta$, either:
        \begin{itemize}[align=left,
                labelwidth=\circlelabel,
                leftmargin={\circlelabel + 1.54em\relax},
                topsep=\mytopsep,
                itemsep=\myitemsep,]
            \item[\circled{i}] $\exists B_1,\,B_2$ such that $(B_1,B_2) \subtyping B$ and $\Gamma \types M_1: B_1,\,\Delta$ and $\Gamma \types M_2 : B_2,\,\Delta$
            \item[\circled{ii}] $\CompPure{\okty} \subtyping B$ and $\Gamma \types M_1: \CompPure{\okty},\,M_2:\CompPure{\okty},\,\Delta$.
        \end{itemize}
        Moreover, since $U \eqsyntac (M_1,M_2)$ for some terms $M_1$ and $M_2$, for $U$ to be a normal form, there must be either
        \settowidth{\titleindent}{(1)}
        \begin{mylist}
            \item[(1)] $M_1 \eqsyntac U_1$ and $M_2 \eqsyntac U_2$ (i.e., $U \eqsyntac (U_1,U_2)$), for some normal forms $U_1$ and $U_2$.
            \item[(2)] $M_1 \eqsyntac S$ and $M_2 \eqsyntac M$ (i.e., $U \eqsyntac (S,M)$), for some stuck term $S$ and some term $M$.
        \end{mylist}
        \settowidth{\titleindent}{\textsf{Case}\circled{2}\,\circled{2}}
        \begin{mypar}
            \item[Case \circled{1}\,\circled{i}.] By transitivity, from $(A_1, A_2) \subtyping A$ and $A \subtyping \CompPure{B}$, we have $(A_1, A_2) \subtyping \CompPure{B}$. Apply \cref{thm:ortholattice} to $(B_1, B_2) \subtyping B$ gives $\CompPure{B} \subtyping \CompPure{(B_1,B_2)}$. Combining this with $(A_1, A_2) \subtyping \CompPure{B}$ via transitivity gives $(A_1, A_2) \subtyping \CompPure{(B_1, B_2)}$. By \cref{lem:subtyping}, this implies either $A_1 \subtyping \okty \land A_2 \subtyping \CompPure{B_2}$ or $A_1 \subtyping \CompPure{B_1}$. Then we proceed by further case analasis on the shape of $M_1$ and $M_2$.
            \settowidth{\titleindent}{(1)}
            \begin{mylist}
                \item[(1)] If $M_1 \eqsyntac U_1$ and $M_2 \eqsyntac U_2$ (i.e., $U \eqsyntac (U_1,U_2)$), for some normal forms $U_1$ and $U_2$:
                \settowidth{\titleindent}{\textendash}
                \begin{mylist}[\textendash]
                    \item If $A_1 \subtyping \okty \land A_2 \subtyping \CompPure{B_2}$, the induction hypothesis on \(U_1\) (with \(\Gamma \vdash U_2 : A_2, \Delta\) and \(\Gamma \vdash U_2 : B_2, \Delta\)) gives \(\Gamma \vdash \Delta\).
                    \item If $A_1 \subtyping \CompPure{B_1}$, the induction hypothesis on \(U_1\) (with \(\Gamma \vdash U_1 : A_1, \Delta\) and \(\Gamma \vdash U_1 : B_1, \Delta\)) also gives \(\Gamma \vdash \Delta\).
                \end{mylist}
                In both cases, this contradicts the assumption that $\Gamma \nvdash \Delta$.
                \item[(2)] $M_1 \eqsyntac S$ and $M_2 \eqsyntac M$ (i.e., $U \eqsyntac (S,M)$), for some stuck term $S$ and some term $M$. By \cref{lem:type-of-stuck-term} we have $\types S: \CompPure{\okty}$. Then \textsc{Weakening} gives $\Gamma \types S: \CompPure{\okty},\,\Delta$.
                \settowidth{\titleindent}{\textendash}
                \begin{mylist}[\textendash]
                    \item If $A_1 \subtyping \okty \land A_2 \subtyping \CompPure{B_2}$, applying the typing rule \textsc{Sub} to $\Gamma \types A_1: \okty,\,\Delta$ gives $\Gamma \types S: \okty,\,\Delta$. Thus, the induction hypothesis applied to $S$ (with $\Gamma \types S: \okty,\,\Delta$ and $\Gamma \types S: \CompPure{\okty},\,\Delta$) gives \(\Gamma \vdash \Delta\).
                    \item If $A_1 \subtyping \CompPure{B_1}$, applying the induction hypothesis to \(U_1\) (with \(\Gamma \vdash U_1 : A_1, \Delta\) and \(\Gamma \vdash U_1 : B_1, \Delta\)) gives \(\Gamma \vdash \Delta\) directly.
                \end{mylist}
                Again, both cases contradict the assumption that $\Gamma \nvdash \Delta$.
            \end{mylist}
            % {\color{purple}
            % This is the problematic case.
            % Consider $(S,M) \eqsyntac (\pn{3}(\abbv{id}),\divtm)$, it is possible to typing it with:
            % \begin{mathpar}
            %     \infer[Pair]{
            %         \Gamma \types \pn{3}(\abbv{id}) : \CompPure{\okty},\,\Delta \\
            %         \Gamma \types \divtm : \intty,\,\Delta
            %     }
            %     {
            %         \Gamma \types {(\pn{3}(\abbv{id}),\,\divtm) : (\CompPure{\okty},\,\intty)},\,\Delta
            %     }
            %     \and
            %     \infer[Pair]{
            %         \Gamma \types \pn{3}(\abbv{id}) : \CompPure{\okty},\,\Delta \\
            %         \Gamma \types \divtm : \CompPure{\intty},\,\Delta
            %     }
            %     {
            %         \Gamma \types {(\pn{3}(\abbv{id}),\,\divtm) : (\CompPure{\okty},\,\CompPure{\intty})},\,\Delta
            %     }
            % \end{mathpar}
            % Even though $(\CompPure{\okty},\,\intty) \subtyping \CompPure{(\CompPure{\okty},\,\CompPure{\intty})}$, it is not necessary that $\Gamma \types \Delta$.
            % }

            \item[Case \circled{1}\,\circled{ii}.]
            By transitivity of subtyping, from $(A_1,A_2)\subtyping A$ and $A \subtyping \CompPure{B}$, we have $(A_1, A_2) \subtyping \CompPure{B}$. Applying the subtyping rule \textsc{CompL} to $\CompPure{\okty} \subtyping B$ yields $\CompPure{B} \subtyping \okty$. Combining $(A_1, A_2) \subtyping \CompPure{B}$ and $\CompPure{B} \subtyping \okty$ via transitivity gives $(A_1, A_2) \subtyping \okty$. By \cref{lem:subtyping}, this implies $A_1 \subtyping \okty$ and $A_2 \subtyping \okty$.
            \settowidth{\titleindent}{(1)}
            \begin{mylist}
                \item[(1)] If $M_1 \eqsyntac U_1$ and $M_2 \eqsyntac U_2$ (i.e., $U \eqsyntac (U_1,U_2)$), for some normal forms $U_1$ and $U_2$:
                \par With $A_1 \subtyping \okty$ and $A_2 \subtyping \okty$, the typing rule \textsc{Sub} applied to $\Gamma \types U_1: A_1,\,\Delta$ and $\Gamma \types U_2 : A_2,\,\Delta$ gives $\Gamma \types U_1: \okty,\,\Delta$ and $\Gamma \types U_2 : \okty,\,\Delta$. By \cref{lem:weakening} (\textsc{Weakening}) on $\Gamma \types U_1: \okty,\,\Delta$, we have:
                \begin{equation}
                    \label{proof:cut-normal-form-pair-1ii1-eq1}
                    \Gamma \types U_1: \okty,\,U_2:\CompPure{\okty},\,\Delta
                \end{equation}
                Moreover, from \circled{ii}, we also have
                \begin{equation}
                    \label{proof:cut-normal-form-pair-1ii1-eq2}
                    \Gamma \vdash U_1 : \CompPure{\okty}, U_2 : \CompPure{\okty}, \Delta
                \end{equation}
                Thus, by the induction hypothesis (with $A = \okty$ and $B = \CompPure{\okty}$) on $U_1$, \cref{proof:cut-normal-form-pair-1ii1-eq1} and \cref{proof:cut-normal-form-pair-1ii1-eq2} implies $\Gamma \types U_2 : \CompPure{\okty},\, \Delta$. Since we have already shown $\Gamma \vdash U_2 : \okty, \Delta$, by the induction hypothesis on $U_2$ (again with $A = \okty$ and $B = \CompPure{\okty}$), it follows that \(\Gamma \vdash \Delta\). This contradicts our assumption that $\Gamma \nvdash \Delta$.
                \item[(2)] If $M_1 \eqsyntac S$ and $M_2 \eqsyntac M$ (i.e., $U \eqsyntac (S,M)$), for some stuck term $S$ and some term $M$:
                \par Applying \cref{lem:type-of-stuck-term} to $S$ gives $\types S: \CompPure{\okty}$, which implies $\Gamma \types S: \CompPure{\okty},\,\Delta$ by \textsc{Weakening}. Meanwhile, applying the typing rule \textsc{Sub} to $\Gamma \types S: A_1,\,\Delta$ with $A_1 \subtyping \okty$ gives $\Gamma \types S: \okty,\,\Delta$. Thus, the induction hypothesis on $S$ gives $\Gamma \nvdash \Delta$
            \end{mylist}
            Therefore, both cases contradict the assumption that $\Gamma \nvdash \Delta$.
            \item[Case \circled{2}\,\circled{i}.] This case follows symmetrically from \textit{Case} \circled{1}\,\circled{ii} via $A \subtyping \CompPure{B} \Leftrightarrow B \subtyping \CompPure{A}$.
            \item[Case \circled{2}\,\circled{ii}.] Applying the subtyping rule \textsc{CompL} to $\CompPure{\okty} \subtyping B$ gives $\CompPure{B} \subtyping \okty$. By transitivity, chaining $\CompPure{\okty} \subtyping A$, $A \subtyping \CompPure{B}$, and $\CompPure{B} \subtyping \okty$ yields $\CompPure{\okty} \subtyping \okty$. By \cref{thm:subtyping-soundness} (\textsc{Soundness of Subtyping}), this implies $\mng{\CompPure{\okty}} \subseteq \mng{\okty}$.
            \par Now, combining $\mng{\CompPure{\okty}} \subseteq \mng{\okty}$ and $\mng{\CompPure{\okty}} = \mng{\topty} \setminus \mng{\okty}$ (from \cref{def:type-meanings}), we derive $\mng{\CompPure{\okty}} \subseteq \mng{\okty} \cap (\mng{\topty} \setminus \mng{\okty})= \varnothing$. Thus $\mng{\CompPure{\okty}} = \mng{\topty} \setminus \mng{\okty} = \varnothing$. This requires all normal forms to be values (i.e., $\mng{\topty} \subseteq \mng{\okty}$), contradicting the existence of stuck terms.
        \end{mypar}

        \indcase{$U$ is of shape $M_1 \otimes M_2$}
        Here we have $\Gamma \types M_1 \otimes M_2:A,\,\Delta$ and $\Gamma \types M_1 \otimes M_2:B,\,\Delta$. By \cref{lem:inversion} (\textsc{Inversion}) on $\Gamma \types M_1 \otimes M_2:A,\,\Delta$ with $\Gamma \nvdash \Delta$, either:
        \begin{circledlist}
            \item $\intty \subtyping A$ and $\Gamma \types M_1 : \intty,\,\Delta$ and $\Gamma \types M_2 : \intty,\,\Delta$
            \item $\CompPure{\okty} \subtyping A$ and $\Gamma \types M_1: \CompPure{\intty},\,M_2:\CompPure{\intty},\,\Delta$.
        \end{circledlist}
        Similarly, by \cref{lem:inversion} (\textsc{Inversion}) on $\Gamma \types M_1 \otimes M_2:B,\,\Delta$ with $\Gamma \nvdash \Delta$, either:
        \begin{itemize}[align=left,
                labelwidth=\circlelabel,
                leftmargin={\circlelabel + 1.54em\relax},
                topsep=\mytopsep,
                itemsep=\myitemsep,]
            \item[\circled{i}] $\intty \subtyping B$ and $\Gamma \types M_1 : \intty,\,\Delta$ and $\Gamma \types M_2 : \intty,\,\Delta$
            \item[\circled{ii}] $\CompPure{\okty} \subtyping B$ and $\Gamma \types M_1: \CompPure{\intty},\,M_2:\CompPure{\intty},\,\Delta$
        \end{itemize}
        In cases \circled{1}\,\circled{ii} and \circled{2}\,\circled{i}, $\Gamma \types \Delta$ follows immediately from the induction hypothesis.
        Moreover, since $U \eqsyntac M_1 \otimes M_2$ for some terms $M_1$ and $M_2$, for $U$ to be a normal form, there must be either
        \settowidth{\titleindent}{(1)}
        \begin{mylist}
            \item[(1)] $M_1 \eqsyntac U_1$ and $U_1$ is not a numeral.
            \item[(2)] $M_1 \eqsyntac \pn{n}$ for some $n$ and $M_2 \eqsyntac U_2$ and $U_2$ is not a numeral.
        \end{mylist}
        It follows from (1) via Lemmas~\ref{lem:all-values-typable} and \ref{lem:type-of-stuck-term} that $\Gamma \types M_1 : \CompPure{\intty},\,\Delta$ and similarly from (2) that $\Gamma \types M_2 : \CompPure{\intty},\,\Delta$.  Therefore, in both cases, consideration of \circled{1}\,\circled{i} yields $\Gamma \types \Delta$ by the induction hypothesis.
        Therefore, all that remains is to consider \circled{2}\,\circled{ii}.  In this case, we have that $\CompPure{\okty} \subtype A \subtype \CompPure{B} \subtype \okty$, which is impossible by Theorem~\ref{thm:subtyping-soundness}.

        % \indcase{$U$ is of shape $\pi_1(M)$}
        % Here we have $\Gamma \types \pi_i(M):A,\,\Delta$ and $\Gamma \types \pi_i(M):B,\,\Delta$. By \cref{lem:inversion} (\textsc{Inversion}) on $\Gamma \types \pi_i(M):A,\,\Delta$ with $\Gamma \nvdash \Delta$, either:
        % \begin{circledlist}
        %     \item $\Gamma \types M:(A, \okty),\,\Delta$
        %     \item $\CompPure{\okty}\subtyping A$ and $\Gamma \types M : \CompPure{\pairvalty},\,\Delta$.
        % \end{circledlist}
        % Similarly, by \cref{lem:inversion} (\textsc{Inversion}) on $\Gamma \types \pi_i(M):B,\,\Delta$ with $\Gamma \nvdash \Delta$, either:
        % \begin{itemize}[align=left,
        %         labelwidth=\circlelabel,
        %         leftmargin={\circlelabel + 1.54em\relax},
        %         topsep=\mytopsep,
        %         itemsep=\myitemsep,]
        %     \item[\circled{i}] $\Gamma \types M:(B, \okty),\,\Delta$
        %     \item[\circled{ii}] $\CompPure{\okty}\subtyping B$ and $\Gamma \types M : \CompPure{\pairvalty},\,\Delta$
        % \end{itemize}
        % In cases \circled{1}\,\circled{i} and \circled{2}\,\circled{ii}, $\Gamma \types \Delta$ follows immediately from the induction hypothesis.  So, the only interesting case is \circled{1}\circled{ii}.  In this case, it follows from \circled{ii} that $A \subtype \okty$ and therefore, from \circled{1} that $\Gamma \types M : \pairvalty,\,\Delta$.  Therefore, $\Gamma \types \Delta$ follows from the induction hypothesis.

        % \indcase{$U$ is of shape $\pi_2(M)$}
        % This case follows as the previous.

        \indcase{$U$ is of shape $\matchtm{U'}{\mid_{i=1}^k p_i \mapsto N_i}$ and $\forall i. \sigma.\,U' \neq p_i\sigma$}
        Here we have $\Gamma \types \matchtm{U'}{\mid_{i=1}^k p_i \mapsto N_i}:A,\,\Delta$ and $\Gamma \types \matchtm{U'}{\mid_{i=1}^k p_i \mapsto N_i}:B,\,\Delta$. By \cref{lem:inversion} (\textsc{Inversion}) on $\Gamma \types \matchtm{U'}{\mid_{i=1}^k p_i \mapsto N_i}:A,\,\Delta$ with $\Gamma \nvdash \Delta$, either:
        \begin{circledlist}
            \item $\Gamma \types U':\textstyle\bigvee_{i=1}^k p_i\theta_i,\,\Delta$ and $\forall i\in[1,k].\,\Gamma,\,\theta_i \types U':\CompPure*{p_i\theta_i},\,N_i:A,\,\Delta$ and, for all $i \in [1,k]$ and $x \in \dom(\theta_i)$, $\theta_i(x) \subtype \okty$
            \item $\CompPure{\okty}\subtyping A$ and $\Gamma \types U':\CompPure*{\textstyle\bigvee_{i=1}^k p_i\thetaok{p_i}},\,\Delta$.
        \end{circledlist}
        Similarly, by \cref{lem:inversion} (\textsc{Inversion}) on $\Gamma \types \matchtm{U'}{\mid_{i=1}^k p_i \mapsto N_i}:B,\,\Delta$ with $\Gamma \nvdash \Delta$, either:
        \begin{itemize}[align=left,
                labelwidth=\circlelabel,
                leftmargin={\circlelabel + 1.54em\relax},
                topsep=\mytopsep,
                itemsep=\myitemsep,]
            \item[\circled{i}] $\Gamma \types U':\textstyle\bigvee_{i=1}^k p_i\theta_i,\,\Delta$ and $\forall i\in[1,k].\,\Gamma,\,\theta_i \types U':\CompPure*{p_i\theta_i},\,N_i:B,\,\Delta$ and, for all $i \in [1,k]$ and $x \in \dom(\theta_i)$, $\theta_i(x) \subtype \okty$
            \item[\circled{ii}] $\CompPure{\okty}\subtyping B$ and $\Gamma \types U':\CompPure*{\textstyle\bigvee_{i=1}^k p_i\thetaok{p_i}},\,\Delta$
        \end{itemize}
        In case \circled{2}\,\circled{ii}, we obtain $\CompPure{\okty} \subtype \okty$ and thus $\Gamma \types \Delta$ follows from the induction hypothesis.  In cases \circled{1}\,\circled{ii} and \circled{2}\,\circled{i}, it follows that $\textstyle\bigvee_{i=1}^k p_i\theta_i \subtype \okty$ and so we may apply the induction hypothesis with $U'$ to obtain $\Gamma \types \Delta$.  Finally, in case \circled{1}\,\circled{i}, for each $i \in [1,k]$, we may apply the induction hypothesis to $N_i$ to obtain $\Gamma,\,\theta_i \types U':\CompPure*{p_i\theta_i},\,\Delta$.  Applying the induction hypothesis again to $U'$ yields $\Gamma,\,\theta_i \types \Delta$.  Finally, since every type in the range of $\theta_i$ is a subtype of $\okty$, every type in $\okty$ is inhabited by a closed value and $\dom(\theta_i) \cap \fv(\Gamma,\,\Delta) = \varnothing$, it follows from Lemma~\ref{lem:simultaneous-substitution} that $\Gamma \types \Delta$.

        \indcase{$U$ is of shape $\abs{x}{M}$}
        Here we have $\Gamma \types \abs{x}{M} : A,\, \Delta$ and $\Gamma \types \abs{x}{M} : B,\, \Delta$.
        Since $\Gamma \nvdash \Delta$, by \cref{lem:inversion} (\textsc{Inversion}):
        \[
            \begin{array}{c@{\text{ such that: }}c@{\;\land\;}c}
                \exists A_1, A_2 & A_1 \to A_2 \subtyping A & \Gamma, x:A_1 \types M:A_2, \Delta \\
                \exists B_1, B_2 & B_1 \to B_2 \subtyping B & \Gamma, x:B_1 \types M:B_2, \Delta
            \end{array}
        \]
        By transitivity of subtyping, from $A \subtyping \CompPure{B}$ and $A_1 \to A_2 \subtyping A$, we have $A_1 \to A_2 \subtyping \CompPure{B}$. Applying the subtyping rule \textsc{CompR} to this yields $B \subtyping A_1 \to A_2$. Combining this with $B_1 \to B_2 \subtyping B$ via transitivity gives $B_1 \to B_2 \subtyping \CompPure{(A_1 \to A_2)}$.
        \par By \cref{thm:subtyping-soundness} (\textsc{Soundness of Subtyping}), $B_1 \to B_2 \subtyping \CompPure{(A_1 \to A_2)}$ implies $\mng{B_1 \to B_2} \subseteq \mng{\CompPure{(A_1 \to A_2)}}$. Since $\mng{\CompPure{(A_1 \to A_2)}}=\topty \setminus \mng{A_1 \to A_2}$ (by \cref{def:type-meanings}), this requires $\mng{B_1 \to B_2} \cap \mng{A_1 \to A_2} = \varnothing$.
        \par However, \cref{lem:fun-to-div} states that $\abs{x}{\divtm} \in \mng{C \to D}$ for all types $C$ and $D$. In particular, $\abs{x}{\divtm} \in \mng{B_1 \to B_2} \cap \mng{A_1 \to A_2}$. This contradicts our earlier conclusion that $\mng{B_1 \to B_2} \cap \mng{A_1 \to A_2} = \varnothing$.
        % \Celia{consider if we want a seperate lemma for $\abs{x}{\divtm}$}
        % \par To derive a contradiction, consider the term $\abs{x}{\divtm}$. By \cref{def:type-meanings}, we have:
        % \settowidth{\titleindent}{\textendash}
        % \begin{itemize}[labelwidth=\titleindent,align=left,leftmargin=\titleindent+1.5em,topsep=\mytopsep,itemsep=\myitemsep,label=\textendash]
        %     \item $\mng{C \to D} = \{\abs{x}{M} \mid \forall N \in \mng*{C}.\:M[N/x] \in \mng*{D}\}$
        %     \item For any type $D$, $\Subst{\divtm}{N}{x} \eqsyntac \divtm \in \mng*{D}$.
        % \end{itemize}
        % It follows that $\abs{x}{\divtm} \in \mng{C \to D}$ for any type $C$, $D$. In particular, $\abs{x}{\divtm} \in \mng{A_1 \to A_2}$ and $\abs{x}{\divtm} \in \mng{B_1 \to B_2}$. But this contradicts $\mng{B_1 \to B_2} \cap \mng{A_1 \to A_2} = \varnothing$.
    \end{indproof}
\end{proof}
% Note: a corollary of this is that, if $\Gamma \types V : \botty, \Delta$, then $\Gamma \types \Delta$.
Note: a corollary of \cref{lem:cut} is as follows.

\begin{corollary}
    \label{cor:bot-elim}
    Let $V$ be a value. If $\Gamma \types V : \botty, \Delta$, then $\Gamma \types \Delta$.
\end{corollary}
\begin{proof}
    We have $\Gamma \types V : \botty, \Delta$ from the assumption.
    Meanwhile, $\Gamma \types V : \topty, \Delta$ holds for any $V$ by the typing rule \textsc{Top}.
    By \cref{thm:ortholattice}, $\botty$ is a subtype of all types. In particular we have $\botty \subtyping \CompPure{\topty}$. Thus applying \cref{lem:cut} with $A = \botty$ and $B = \topty$ gives $\Gamma \types \Delta$.
\end{proof}

%% Note: for the conclusion of this lemma to hold, either \Gamma \nvdash \Delta or A \tyneq \botty is needed. Or, it is also possible to alter the conclusion to: $\exists B\,(B\subtyping \okty \land B\subtyping A)$ such that $\Gamma \types V:B,\,\Delta$, and call \cref{lem:value-cut} to derive that if B\tyeq\botty, then \Gamma \types \Delta when B \tyneq \botty is needed later in \cref{lem:preservation}.

% \Celia{Condition(1)(2) and (3) are not disjoint.}
\begin{lemma}
    \label{lem:pattern-subst-value}
    Let $p$ be a pattern, and let $\sigma$ be a term substitution, such that $p\sigma$ is a value with $\fv(p) = \dom(\sigma)$. If $\Gamma \types p\sigma:A,\,\Delta$ where $\Gamma \nvdash \Delta$, then there exists a type substitution $\theta$ with $\dom(\theta)=\fv(p)$ such that:
    \begin{enumerate}
        \item %$\dom(\theta) = \fv(p)$ and 
              $\forall x \in \dom(\theta). \theta(x) \subtyping \okty$
        \item $\Gamma \types \sigma\tyfit\theta,\,\Delta$
        \item $\Gamma \types p\sigma:p\theta,\,\Delta$ where $p\theta \subtyping \okty \land p\theta \subtyping A$
    \end{enumerate}
\end{lemma}
\begin{proof}
    By induction on the shape of $p$.
    \begin{indproof}
        \indcase{$p$ has shape $y$}
        In this case, $p\sigma \eqsyntac y\sigma$, $\fv(p)=\makeset{y}$ and $\Gamma \types y\sigma:A,\,\Delta$. Since $y\sigma$ is a value, by \cref{lem:type-of-value}, $\exists B\,(B \subtyping A \land B \subtyping \okty)$ such that $\Gamma \types y\sigma:B,\,\Delta$. Define the substitution $\theta = [y \mapsto B]$, then $\dom(\theta) = \{y\} = \fv(p)$. We confirm that $\theta$ satisfies:
        \settowidth{\titleindent}{(1)}
        \begin{mylist}
            \item[(1)] $\forall x \in \dom(\theta). \theta(x) \subtyping \okty$:
            \par Immediate from the definition of $\theta$ and $B \subtyping \okty$
            \item[(2)] $\Gamma \types \sigma\tyfit\theta,\,\Delta$:
            \par Substituting $y\sigma \eqsyntac \sigma(y)$ and $\theta(y)=B$ into $\Gamma \types y\sigma:B,\,\Delta$ yields $\Gamma \types \sigma(y):\theta(y),\,\Delta$. Moreover, $\dom(\sigma)=\dom(\theta)=\makeset{y}$. Thus all conditions required by \cref{def:coherence-subst-ctx} (\textsc(Coherence)) are satisfied.
            \item[(3)] $\Gamma \types p\sigma:p\theta,\,\Delta$ where $p\theta \subtyping \okty \land p\theta \subtyping A$:
            \par Since $B\subtyping A \land B\subtyping \okty$, by substituting $p\sigma \eqsyntac y\sigma$ and $p\theta \eqsyntac B$ into $\Gamma \types y\sigma : B, \Delta$, the condition follows.
        \end{mylist}
        \indcase{$p$ has shape $\atomtm{a}$}
        In this case, $p\sigma \eqsyntac \atomtm{a}\sigma$, $\fv(p)=\varnothing$ and $\Gamma \types \atomtm{a}\sigma:A,\,\Delta$. Let $\theta$ be the empty substitution so that $\dom(\theta)=\fv(p)=\varnothing$. We check that $\theta$  satisfies:
        \begin{mylist}
            \item[(1)] $\forall x \in \dom(\theta).\,\theta(x)\subtyping\okty$:
            \par Immediate from $\dom(\theta)=\varnothing$.
            \item[(2)] $\Gamma \types \sigma\tyfit\theta,\,\Delta$:
            \par Immediate from $\dom(\sigma)=\dom(\theta)=\varnothing$, by \cref{def:coherence-subst-ctx}.
            \item[(3)] $\Gamma \types p\sigma:p\theta,\,\Delta$ where $p\theta \subtyping \okty \land p\theta \subtyping A$:
            \par By the typing rule \textsc{Atom}, we have $\Gamma \types \atomtm{a}:\atom{a},\,\Delta$. Substituting  $p\sigma \eqsyntac \atomtm{a}\sigma \eqsyntac \atomtm{a}$ and $p\theta \eqsyntac \atom{a}\theta \eqsyntac \atom{a}$ into $\Gamma \types \atomtm{a}:\atom{a},\,\Delta$ gives $\Gamma \types p\sigma:p\theta,\,\Delta$. Since $\Gamma \nvdash \Delta$, applying \cref{lem:inversion} (\textsc{Inversion}) to $\Gamma \types \atomtm{a}\sigma:A,\,\Delta$ yields $\atom{a} \subtyping A$. Finally, since $\atom{a}\subtyping\atomty$ (by the subtyping rule \textsc{Atom}), $\atomty \subtyping \intty \vee \pairty \vee \funty \vee \atomty$ (by the subtyping rule \textsc{UnionR}) and $\intty \vee \pairty \vee \funty \vee \atomty \subtyping \okty$ (by the subtyping rule \textsc{Atom}), we have $p\theta \subtyping \okty$ by transitivity.
        \end{mylist}
        \indcase{$p$ has shape $(p_1, p_2)$} In this case $\Gamma \types (p_1,p_2)\,\sigma:A,\,\Delta$ and $\dom(\sigma)=\fv(p)=\fv(p_1)\cup \fv(p_2)$. Let $\sigma_1 = \sigma |_{\fv(p_1)}$ and $\sigma_2 = \sigma |_{\fv(p_2)}$, then $\dom(\sigma_{1})=\fv(p_1)$ and $\dom(\sigma_{2})=\fv(p_2)$. Since $\Gamma \nvdash \Delta$ and $(p_1,p_2)\,\sigma \eqsyntac (p_1\sigma,p_2\sigma)\eqsyntac (p_1\sigma_1,p_2\sigma_2)$, applying \cref{lem:inversion} (\textsc{Inversion}) to $\Gamma \types (p_1,p_2)\,\sigma:A,\,\Delta$, we obtain:
        \begin{equation}
            \label{proof:pattern-type-substitution-pair-1}
            \exists A_1,\,A_2 \text{ such that } (A_1,A_2)\subtyping A \;\land\; \Gamma \types p_1\sigma_1:A_1,\,\Delta \;\land\; \Gamma \types p_2\sigma_2:A_2,\,\Delta
        \end{equation}
        Furthermore, since $(p_1, p_2)\,\sigma$ is a value, both $p_1\sigma_1$ and $p_2\sigma_2$ must be values. By applying the induction hypothesis to the type judgments in \cref{proof:pattern-type-substitution-pair-1}, there exist $\theta_1$ with $\dom(\theta_1)=\fv(p_1)$ and $\theta_2$ with $\dom(\theta_2)=\fv(p_2)$, such that:
        \settowidth{\titleindent}{(IH.1)}
        \begin{mylist}
            \item[(IH.1)] $\forall x \in \dom(\theta_1).\,\theta_1(x) \subtyping \okty$ and $\forall x \in \dom(\theta_2).\,\theta_2(x) \subtyping \okty$
            \item[(IH.2)] $\Gamma \types \sigma_1\tyfit\theta_1,\,\Delta$ and  $\Gamma \types \sigma_2\tyfit\theta_2,\,\Delta$
            \item[(IH.3)] $\Gamma \types p\sigma_1:p_1\theta,\,\Delta$ where $p_1\theta_1 \subtyping \okty \land p_1\theta_1 \subtyping A_1$ and $\Gamma \types p\sigma_2:p_2\theta,\,\Delta$ where $p_2\theta_2 \subtyping \okty \land p_2\theta_2 \subtyping A_2$
        \end{mylist}
        Define $\theta$ to be the substitution:
        \[
            \theta(x) =
            \begin{cases}
                \theta_1(x) & \text{if } x\in\dom(\theta_1), \\
                \theta_2(x) & \text{if } x\in\dom(\theta_2).
            \end{cases}
        \]
        where $\dom(\theta)=\dom(\theta_1)\cup\dom(\theta_2)$. Then $\dom(\theta)=\dom(\sigma)=\fv(p)$. Now we check that $\theta$ satisfies:
        \settowidth{\titleindent}{(1)}
        \begin{mylist}
            \item[(1)] $\forall x \in \dom(\theta).\,\theta(x) \subtyping \okty$:
            \par By definition of $\theta$ and (IH.1).
            \item[(2)] $\Gamma \types \sigma\tyfit\theta,\,\Delta$:
            \par By \cref{def:coherence-subst-ctx}, (IH.2) implies:
            \begin{equation}
                \label{proof:pattern-subst-value-pair-eq1}
                \forall x \in \dom(\theta_i).\,\Gamma \types \sigma_{i}(x):\theta_{i}(x)\quad i\in\{1,2\}
            \end{equation}
            By definition of $\sigma_1$, $\sigma_2$ and $\theta$, we have:
            \begin{equation}
                \label{proof:pattern-subst-value-pair-eq2}
                \forall x \in \dom(\theta_i).\,\sigma_i(x)\eqsyntac\sigma(x) \text{ and } \theta(x) \eqsyntac \theta_i(x) \quad i\in\{1,2\}
            \end{equation}
            Substituting \cref{proof:pattern-subst-value-pair-eq2} into \cref{proof:pattern-subst-value-pair-eq1} gives:
            \begin{equation}
                \label{proof:pattern-subst-value-pair-eq3}
                \forall x \in \dom(\theta_i).\,\Gamma \types \sigma(x):\theta(x)\quad i\in\{1,2\}
            \end{equation}
            Moreover, as $\dom(\theta)=\dom(\theta_1)\cup\dom(\theta_2)$, for any $x \in \dom(\theta)$, either $x \in \dom(\theta_1)$ or $x \in \dom(\theta_2)$. This combined with \cref{proof:pattern-subst-value-pair-eq3} implies:
            \[
                \forall x \in \dom(\theta).\,\Gamma \types \sigma(x):\theta(x)
            \]
            Since $\dom(\theta)=\dom(\sigma)$, by \cref{def:coherence-subst-ctx}, this shows $\Gamma \types \sigma\tyfit\theta,\,\Delta$ by \cref{def:coherence-subst-ctx}.
            \item[(3)]  $\Gamma \types (p_1,p_2)\,\sigma:(p_1,p_2)\,\theta$, where $(p_1,p_2)\,\theta \subtyping \okty \land (p_1,p_2)\,\theta \subtyping A$
            \par Applying the typing rule \textsc{Pair} to the type judgements in (IH.3), we have:
            \begin{equation}
                \label{proof:pattern-type-substitution-pair-2}
                \Gamma \types (p_1\sigma_1, p_2\sigma_2):(p_1\theta_1,p_2\theta_2),\,\Delta
            \end{equation}
            By definition of $\sigma_1$, $\sigma_2$ and $\theta$, $(p_1,p_2)\,\sigma \eqsyntac (p_1\sigma, p_2\sigma) \eqsyntac (p_1\sigma_1, p_2\sigma_2)$ and $(p_1,p_2)\,\theta \eqsyntac (p_1\theta, p_2\theta) \eqsyntac (p_1\theta_1, p_2\theta_2)$. Substituting $(p_1\sigma_1, p_2\sigma_2) \eqsyntac (p_1,p_2)\,\sigma$ and $(p_1\theta_1, p_2\theta_2) \eqsyntac (p_1,p_2)\,\theta$ into \cref{proof:pattern-type-substitution-pair-2}, we get:
            \[
                \Gamma \types (p_1,p_2)\,\sigma:(p_1,p_2)\,\theta,\,\Delta
            \]
            \par From (IH.3), we also have $p_1\theta_1 \subtyping \okty \land p_1\theta_1 \subtyping A_1$ and $p_2\theta_2 \subtyping \okty \land p_2\theta_2 \subtyping A_2$. Applying the subtyping rule \textsc{Pair} with $p_1\theta_1 \subtyping \okty$ and $p_2\theta_2 \subtyping \okty$ gives $(p_1\theta_1, p_2\theta_2) \subtyping \pairty$. Combining this with $\pairty \subtyping \okty$ via transitivity, we have $(p_1\theta_1, p_2\theta_2) \subtyping \okty$. Similarly, applying the subtyping rule \textsc{Pair} with $p_1\theta_1 \subtyping A_1$ and $p_2\theta_2 \subtyping A_2$ yields $(p_1\theta_1, p_2\theta_2) \subtyping (A_1,A_2)$. Since we also know $(A_1,A_2)\subtyping A$ from inversion, transitivity gives $(p_1\theta_1, p_2\theta_2) \subtyping A$. Finally, Substituting $(p_1\sigma_1, p_2\sigma_2) \eqsyntac (p_1,p_2)\,\sigma$ and $(p_1\theta_1, p_2\theta_2) \eqsyntac (p_1,p_2)\,\theta$ into $(p_1\theta_1, p_2\theta_2) \subtyping \okty$ and $(p_1\theta_1, p_2\theta_2) \subtyping A$, we conclude that $(p_1,p_2)\,\theta \subtyping \okty \land (p_1,p_2)\,\theta \subtyping A$.
        \end{mylist}
        %\par Therfore, $\theta$ satisfies the required condition.
    \end{indproof}
\end{proof}

\begin{lemma}[Preservation for Redex Reduction]
    \label{lem:preservation-redex}
    Let $R$ be a redex. If $\Gamma \types R:A,\, \Delta$ and $R \ped R'$, then $\Gamma \types R':A,\, \Delta$.
\end{lemma}
\begin{proof}
    We start with a simple case analysis on whether $\Gamma \types \Delta$.
    \begin{itemize}
        \item Suppose $\Gamma \types \Delta$. Then by \cref{lem:weakening} (\textsc{Weakening}), it follows immediately that we also have $\Gamma \types R':A,\, \Delta$.
        \item Otherwise, $\Gamma \nvdash \Delta$. In this case, we proceed by case analysis on the structure of the redex $R$.
              \begin{indproof}
                  \indcase{\textsc{Beta}}
                  Suppose $R \ped R'$ has shape $(\abs{x}{M})\,N \ped M[N/x]$. Then $R \eqsyntac (\abs{x}{M})N$, $R'\eqsyntac M[N/x]$ and $\Gamma \types (\abs{x}{M})N:A,\, \Delta$. As $\Gamma \nvdash \Delta$, by \cref{lem:inversion} (\textsc{Inversion}) on $\Gamma \types (\abs{x}{M})N:A,\, \Delta$, either:
                  \begin{circledlist}%[inline]
                      \item $\exists B,\, C \text{ such that } C \subtyping A$ and $\Gamma \types N:B,\,\Delta$ and $\Gamma \types \abs{x}{M}:B \to C,\,\Delta$;
                      \item or, $\CompPure{\okty}\subtyping A$ and $\Gamma \types \abs{x}{M} : \CompPure{\funty},\,\Delta$.
                  \end{circledlist}
                  %%% ver. 1
                  %   \color{ACMPurple}
                  %   \par Assume \circled{2} holds. Since $\Gamma \nvdash \Delta$, applying \cref{lem:inversion} (\textsc{Inversion}) to $\Gamma \types \abs{x}{M} : \CompPure{\funty},\,\Delta$ yields:
                  %   \begin{equation*}
                  %       \exists B',\,C' \text{ such that } B' \to C' \subtyping \CompPure{\funty} \,\land\, \Gamma,\, x:B' \types M:C',\,\Delta
                  %   \end{equation*}
                  %   It follows immediately that $\Gamma \types \abs{x}{M}:B'\to C',\,\Delta$ by the typing rule \textsc{Abs}. Moreover, since $B' \to C' \subtyping \botty \to \topty$ (by the subtyping rule \textsc{Fun}) and $\funty \eqsyntac \botty \to \topty$ (by definition), we have $B' \to C' \subtyping \funty$. Thus, by applying the typing rule \text{Sub} to $\Gamma \types \abs{x}{M}:B'\to C',\,\Delta$:
                  %   \begin{equation}
                  %       \label{proof:preservation-beta-cond2-eq1}
                  %       \Gamma \types \abs{x}{M}:\funty,\,\Delta
                  %   \end{equation}
                  %   However, we already have $\Gamma \types \abs{x}{M} : \CompPure{\funty},\,\Delta$ from \circled{2}. By \cref{lem:cut}, this combined with \cref{proof:preservation-beta-cond2-eq1} implies $\Gamma \nvdash \Delta$. Since we 
                  %%% ver. 2 I decided to use this
                  \par Assume \circled{2} holds. Since $\Gamma \nvdash \Delta$, applying \cref{lem:inversion} (\textsc{Inversion}) to $\Gamma \types \abs{x}{M} : \CompPure{\funty},\,\Delta$ yields:
                  \begin{equation}
                      \label{proof:preservation-beta-cond2-inv}
                      \exists B',\,C' \text{ such that } B' \to C' \subtyping \CompPure{\funty} \,\land\, \Gamma,\, x:B' \types M:C',\,\Delta
                  \end{equation}
                  Moreover, by the subtyping rule \textsc{Fun}, we have $B' \to C' \subtyping \funty$ (recall that $\funty \eqsyntac \botty \to \topty$ by definition). Therefore, we have both $B' \to C' \subtyping \CompPure{\funty}$ (from \cref{proof:preservation-beta-cond2-inv}) and $B' \to C' \subtyping \funty$. By \cref{thm:subtyping-soundness} (\textsc{Soundness of Subtyping}), this implies $\mng{B' \to C'} \subseteq \mng{\CompPure{\funty}}$ and $\mng{B' \to C'} \subseteq \mng{\funty}$. Consequently,
                  \begin{equation}
                      \label{proof:preservation-beta-cond2-eq1}
                      \mng{B' \to C'} \subseteq \mng{\funty} \cap \mng{\CompPure{\funty}}
                  \end{equation}
                  By \cref{def:type-meanings}, $\mng{\funty} \cap \mng{\CompPure{\funty}} = \mng{\funty} \cap \mng{\topty} \setminus \mng{\funty} = \varnothing$. Substituting this into \cref{proof:preservation-beta-cond2-eq1} gives $\mng{B' \to C'} \subseteq \varnothing$, and thus $\mng{B' \to C'} = \varnothing$.
                  However, \cref{lem:fun-to-div} states that $\abs{x}{\divtm} \in \mng{B' \to C'}$, which contradicts $\mng{B' \to C'} = \varnothing$. Therefore, our assumption that \circled{2} holds fails, leaving case \circled{1} as the only possibility.
                  \par From \circled{1}, $\Gamma \types \abs{x}{M}:B \to C,\,\Delta$. Since $\Gamma \nvdash \Delta$, applying \cref{lem:inversion} (\textsc{Inversion}), we have:
                  \begin{equation}
                      \label{proof:preservation-beta-1}
                      \exists B',\,C' \text{ such that } B' \to C' \subtyping B \to C \,\land\, \Gamma,\, x:B' \types M:C',\,\Delta
                  \end{equation}
                  By \cref{lem:subtyping}, $B' \to C' \subtyping B \to C$ implies $B \subtyping B'$ and $C' \subtyping C$. Thus we can apply \cref{lem:substitution} to $\Gamma \types N:B,\, \Delta$ (from \circled{1}) and  $\Gamma,\, x:B' \types M:C',\,\Delta$ (from \cref{proof:preservation-beta-1}) to obtain $\Gamma \types \Subst{M}{N}{x}:C',\, \Delta$. Moreover, combining $C' \subtyping C$ and $C \subtyping A$ via transitivity gives $C' \subtyping A$. By the typing rule \textsc{Sub}, this conbined with $\Gamma \types \Subst{M}{N}{x}:C',\, \Delta$ implies $\Gamma \types \Subst{M}{N}{x}:A,\, \Delta$. Given that $\Subst{M}{N}{x} \eqsyntac R'$, this completes the proof of the case.

                  \indcase{\textsc{Fix}}
                  Suppose $R \ped R'$ has shape $\fixtm{x}{M} \ped  M[\fixtm{x}{M}/x]$.
                  Then $R \eqsyntac \fixtm{x}{M}$, $R'\eqsyntac M[\fixtm{x}{M}/x]$ and $\Gamma \types \fixtm{x}{M}:A,\, \Delta$. Since $\Gamma \nvdash \Delta$,
                  by \cref{lem:inversion} (\textsc{Inversion}):
                  \begin{equation}
                      \label{proof:preservation-fix-1}
                      \exists A' \text{ such that } A' \subtyping A \text{ and } \Gamma,\,x:A' \types M:A',\,\Delta
                  \end{equation}
                  Applying the typing rule \textsc{Fix} back to \cref{proof:preservation-fix-1}, we derive:
                  \begin{equation}
                      \label{proof:preservation-fix-2} \Gamma \types \fixtm{x}{M}:A',\, \Delta
                  \end{equation}
                  By \cref{lem:substitution} (\textsc{Substitution}), \cref{proof:preservation-fix-1} and \cref{proof:preservation-fix-2} implies $\Gamma \types \Subst{M}{\fixtm{x}{M}}{x}:A',\, \Delta$. Since $A' \subtyping A$, applying the typing rule \textsc{Sub} gives $\Gamma \types \Subst{M}{\fixtm{x}{M}}{x}:A,\, \Delta$. Given that $R' \eqsyntac \Subst{M}{\fixtm{x}{M}}{x}:A'$, this proves $\Gamma \types R':A,\, \Delta$.
                  \indcase{\textsc{Op}}
                  Suppose $R \ped R'$ has shape $\opPrim{\pn{n}}{\pn{m}} \ped \pn{\opPrim{n}{m}}$.
                  Then $R \eqsyntac \opPrim{\pn{n}}{\pn{m}}$, $R'\eqsyntac \pn{\opPrim{n}{m}}$ and $\Gamma \types \opPrim{\pn{n}}{\pn{m}}:A,\,\Delta$. Since $\Gamma \nvdash \Delta$, by \cref{lem:inversion} (\textsc{Inversion}) on $\Gamma \types \opPrim{\pn{n}}{\pn{m}}:A,\,\Delta$, either:
                  \begin{circledlist}%[inline]
                      \item $\intty \subtyping A$ and $\Gamma \types \pn{n} : \intty,\,\Delta$ and $\Gamma \types \pn{m} : \intty,\,\Delta$
                      \item or, $\CompPure{\okty}\subtyping A$ and $\Gamma \types \pn{m} : \CompPure{\intty},\,\pn{n}:\CompPure{\intty},\,\Delta$.
                  \end{circledlist}
                  \par Assume \circled{2} holds. Since $\Gamma \nvdash \Delta$, applying \cref{lem:inversion} (\textsc{Inversion}) to $\Gamma \types \pn{m} : \CompPure{\intty},\,\pn{n}:\CompPure{\intty},\,\Delta$ yields:
                  \begin{align*}
                      \text{Either } \intty \subtyping \CompPure{\intty} \text{ or } \Gamma \types \pn{n}:\CompPure{\intty},\, \Delta
                  \end{align*}
                  By \cref{thm:subtyping-soundness} (\textsc{Soundness of Subtyping}), $\intty \subtyping \CompPure{\intty}$ implies $\mng{\intty} \subseteq \mng{\CompPure{\intty}}$. But by \cref{def:type-meanings}, this is impossible. Consequently, $\Gamma \types \pn{n}:\CompPure{\intty},\,\Delta$ must hold. However, since we also have $\Gamma \types \pn{n}:\intty,\,\Delta$ (by the typing rule \textsc{Int}), it follows from \cref{lem:cut} that $\Gamma \types \Delta$. Since we are in the case where $\Gamma \nvdash \Delta$, \circled{2} has leaded to contradiction, leaving \circled{1} as the only possibility.
                  \par From \circled{1}, $\Gamma \types \pn{n} : \intty,\,\Delta$ and $\Gamma \types \pn{m} : \intty,\,\Delta$. Applying the {Op} rule then yields $\Gamma \types \opPrim{\pn{m}}{\pn{n}} : \intty,\,\Delta$. Since $R'\eqsyntac \opPrim{\pn{m}}{\pn{n}}$, this proves $\Gamma \types R':A,\, \Delta$.

                  \indcase{\textsc{Match-Hit}}
                  Suppose $R \ped R'$ has shape:
                  \[
                      \matchtm{V}{\mid_{i=1}^k p_i \mapsto N_i} \ped N_j\sigma \quad (V \eqsyntac N_j\sigma)
                  \]
                  Since $V\eqsyntac p_j\sigma$ for some $1\leq j \leq k$, $R \eqsyntac \matchtm{p_j\sigma}{\mid_{i=1}^k p_i \mapsto N_i}$ and $R'\eqsyntac N_j\sigma$. Moreover:
                  \begin{equation}
                      \label{proof:preservation-match-hit-1}
                      \Gamma \types \matchtm{p_j\sigma}{\mid_{i=1}^k p_i \mapsto N_i}:A,\, \Delta
                  \end{equation}
                  \par Applying \cref{lem:inversion} (\textsc{Inversion}) to \cref{proof:preservation-match-hit-1} with $\Gamma \nvdash \Delta$, we obtain that one of the following holds:
                  \begin{circledlist}%[inline]
                      \item $\Gamma \types p_j\sigma:\textstyle\bigvee_{i=1}^k p_i\theta_i,\,\Delta$ and $\forall i\in[1,k].\,\Gamma,\,\theta_i \types p_j\sigma:\CompPure*{p_i\theta_i},\,N_i:A,\,\Delta$
                      \item or, $\CompPure{\okty}\subtyping A$ and $\Gamma \types p_j\sigma : \CompPure*{\textstyle\bigvee_{i=1}^k p_i\thetaok{p_i}},\,\Delta$.
                  \end{circledlist}
                  Since $\Gamma \nvdash \Delta$ and $p_j\sigma \eqsyntac V$ is a value, by \cref{lem:pattern-subst-value}, for any type $C$ where $\Gamma \types p_j\sigma :C,\,\Delta$, there exists a type substitution $\theta$ with $\dom(\theta)=\fv(p_j)$ satisfying:
                  \settowidth{\titleindent}{(1)}
                  \begin{mylist}
                      \item[(1)] %$\dom(\theta) = \fv(p)$ and 
                      $\forall x \in \dom(\theta). \theta(x) \subtyping \okty$, and
                      \item[(2)] $\Gamma \types \sigma\tyfit\theta,\,\Delta$
                      \item[(3)] $\Gamma \types p_j\sigma:p_j\theta,\,\Delta$ where $p_j\theta \subtyping \okty \land p_j\theta \subtyping C$
                  \end{mylist}
                  \par Assume \circled{2} holds, then:
                  \begin{equation}
                      \label{proof:preservation-match-hit-2-eq1}
                      \Gamma \types p_j\sigma: \CompPure*{\textstyle\bigvee_{i=1}^k p_i\thetaok{p_i}},\,\Delta
                  \end{equation}
                  Thus there exists $\theta$ satistifying (1)--(3) for $C \eqsyntac \CompPure*{\textstyle\bigvee_{i=1}^k p_i\thetaok{p_i}}$.
                  It is immediate from condition (1) that $\theta \subtyping \thetaok{p_j}$. Thus by the $\Leftarrow$ directon of \cref{lem:subst-subtyping-on-pattern}, $p_j\theta \subtyping p_j\thetaok{p_j}$. Since $\Gamma \types p_j\sigma:p_j\theta,\,\Delta$ (from condition (3)), applying the typing rule \textsc{Sub}, we derive:
                  \begin{equation}
                      \label{proof:preservation-match-hit-2-eq2}
                      \Gamma \types p_j\sigma:p_j\thetaok{p_j},\,\Delta
                  \end{equation}
                  By \rlnm{UnionR} and the order-reversing property of complement (\cref{thm:ortholattice}), $\CompPure*{\textstyle\bigvee_{i=1}^k p_i\thetaok{p_i}} \subtype \CompPure*{p_j\thetaok{p_j}}$.  Therefore, \cref{lem:cut}, \cref{proof:preservation-match-hit-2-eq1} and \cref{proof:preservation-match-hit-2-eq2} implies $\Gamma \types \Delta$. As we are in the case where $\Gamma \nvdash \Delta$, our assumption that \circled{2} holds leads to a contradiction. Thus \circled{1} must hold.
                  \par From \circled{1}, $\Gamma \types p_j\sigma : \textstyle\bigvee_{i=1}^k p_i\theta_i,\,\Delta$ (where $p_j\sigma \eqsyntac V$).
                  Thus there exists $\theta$ satistifying (1)--(3) for $C \eqsyntac \textstyle\bigvee_{i=1}^k p_i\theta_i$. In this case, condition (3) says:
                  \[
                      \Gamma \types p_j\sigma:p_j\theta,\,\Delta \quad \text{ where } p_j\theta \subtyping \okty \land p_j\theta \subtyping \textstyle\bigvee_{i=1}^k p_i\theta_i
                  \]
                  Applying \cref{lem:subtyping} to $p_j\theta \subtyping \textstyle\bigvee_{i=1}^k p_i\theta_i$. Then:
                  \settowidth{\titleindent}{(2)}
                  \begin{mylist}
                      \item[(1)] If $p_j\theta$ is a complement or a union type, it has to be $k=1$, where $\Gamma \types p_j\sigma:p_j\theta,\,\Delta$ and $j=1$.
                      \item[(2)] Otherwise, $p_j\theta$ is neither a complement nor a union. Thus we can apply \cref{lem:subtyping} to obtain $\exists i \in [1,k]. \, p_j\theta \subtyping p_i\theta_i$. Given that patterns are disjoint, it has to be $i=j$ and $\Gamma \types p_j\sigma:p_j\theta,\,\Delta$
                  \end{mylist}
                  In both cases, $\Gamma \types p_j\sigma:p_j\theta,\,\Delta$. Then by the $\Rightarrow$ direction of \cref{lem:subst-subtyping-on-pattern}, $\theta \subtyping \theta_j$. We also have $\forall i\in[1,k].\,\Gamma,\,\theta_i \types p_j\sigma:\CompPure*{p_i\theta_i},\,N_i:A,\,\Delta$ from \circled{1}. In particular, when $i=j$:
                  \begin{equation}
                      \label{proof:preservation-match-hit-1-eq1}
                      \Gamma,\,\theta_j \types N_j:A,\,p_j\sigma:\CompPure*{p_j\theta_j},\,\Delta
                  \end{equation}
                  \par Having established \cref{proof:preservation-match-hit-1-eq1}, $\theta \subtyping \theta_j$ and $\Gamma \types \sigma\tyfit\theta,\,\Delta$ (from condition (2)), we now verify $\dom(\theta) \cap \fv(\{p_j\sigma:\CompPure*{p_j\theta_j}\} \cup \Delta) = \varnothing$ to be able to apply \cref{lem:simultaneous-substitution} (\textsc{Simultaneous Substitution}). Given that $\fv(p_1,\ldots,p_k) \cap \fv(\Gamma \cup \Delta) = \varnothing$, $\fv(p_i) \cap \fv(\Delta) = \varnothing$ holds for all $i \in [1,k]$. In particular, setting $i=j$ gives $\fv(p_j) \cap \fv(\Delta) = \varnothing$. Since $\dom(\theta) = \fv(p_j)$, it follows that $\dom(\theta) \cap \fv(\Delta) = \varnothing$. Furthermore, as $p_j\sigma$ is a closed value, $\fv(p_j\sigma) = \varnothing$. Thus $\dom(\theta) \cap \fv(\{p_j\sigma:\CompPure*{p_j\theta_j}\} \cup \Delta) = (\dom(\theta) \cap \fv(p_j\sigma)) \cup (\dom(\theta) \cap \fv(\Delta) = \varnothing)$.
                  \par Now, applying \cref{lem:simultaneous-substitution}, we have:
                  \begin{equation}
                      \label{proof:preservation-match-hit-1-eq2}
                      \Gamma \types N_j\sigma:A,\,p_j\sigma:\CompPure*{p_j\theta_j},\,\Delta
                  \end{equation}
                  Meanwhile, since $\Gamma \types p_j\sigma:p_j\theta,\,\Delta$ (from condition (3)), by applying \cref{lem:weakening} (\textsc{Weakening}), we have:
                  \begin{equation}
                      \label{proof:preservation-match-hit-1-eq3}
                      \Gamma \types N_j\sigma:A,\,p_j\sigma:p_j\theta_j,\,\Delta
                  \end{equation}
                  By \cref{lem:cut}, since $\CompPure*{p_j\theta_j} \subtyping \CompPure*{p_j\theta_j}$, \cref{proof:preservation-match-hit-1-eq2} and \cref{proof:preservation-match-hit-1-eq3} implies $\Gamma \types N_j\sigma:A,\,\Delta$. Given that $R' \eqsyntac N_j\sigma:A$, this proves $\Gamma \types R':A,\,\Delta$.
              \end{indproof}

              \par Therefore, whenever $R$ is a redex, we have shown that $\Gamma \types R:A,\,\Delta$ implies $\Gamma \types R':A,\,\Delta$ for $R \ped R'$.
    \end{itemize}
\end{proof}
\begin{theorem}[Preservation]
    \label{thm-preservation}
    Suppose $Q \ped Q'$. If $\Gamma \types Q:A,\, \Delta$, then $\Gamma \types Q':A,\, \Delta$.
\end{theorem}
\begin{proof}
    We start with a simple case analysis on whether $\Gamma \types \Delta$.
    \begin{itemize}
        \item Suppose $\Gamma \types \Delta$. Then by \cref{lem:weakening} (\textsc{Weakening}), it follows immediately that we also have $\Gamma \types Q':A,\, \Delta$.
        \item Otherwise, $\Gamma \nvdash \Delta$. By \cref{def:reduction}, for any $Q \ped Q'$, there exists an evaluation context $\Ctx$ and a redex $R$, such that $Q \syntaceq \Ctx[R]$ and $Q'\syntaceq \Ctx[R']$ for $R \ped R'$. Therefore, it suffices to show:
              \[
                  \forall R\,(R \text{ is a redex} \land R \ped R')\,\forall \Ctx\,(\Ctx \text{ is a context}).\; \Gamma \types \Ctx[R] :A,\, \Delta \Rightarrow \Gamma \types \Ctx[R']:A,\, \Delta
              \]
              We proceed by induction on the structure of $\Ctx$.
              \begin{indproof}

                  \indcase{$\Ctx$ has shape $\Hole$}
                  Here, given $\Gamma \types R:A,\,\Delta$, we must show $\Gamma \types R':A,\, \Delta$. Since $R$ is a redex where $R \ped R'$, this follows directly from \cref{lem:preservation-redex}.

                  \indcase{$\Ctx$ has shape $\CtxF\,N$}
                  Here, given $\Gamma \types \app*{\CtxF[R]}{N}:A,\,\Delta$, we must show $\Gamma \types \app*{\CtxF[R']}{N}:A,\,\Delta$. By \cref{lem:inversion} (\textsc{Inversion}) on $\Gamma \types \app*{\CtxF[R]}{N}:A,\,\Delta$ with $\Gamma \nvdash \Delta$, either:
                  \begin{circledlist}
                      \item $\exists B, C$ with $C \subtyping A$ such that $\Gamma \types \CtxF[R]:B \to C,\,\Delta$ and $\Gamma \types N:B,\,\Delta$
                      \item $\CompPure{\okty} \subtyping A$ and $\Gamma \types \CtxF[R] : \CompPure{\funty},\,\Delta$.
                  \end{circledlist}

                  \settowidth{\titleindent}{\textit{Case} \circled{2}.}
                  \begin{mypar}
                      \item[Case \circled{1}.] Since $R$ is a redex with $R \ped R'$ and $\CtxF$ is a context, the induction hypothesis applied to $\Gamma \types \CtxF[R]:B \to C,\,\Delta$ gives $\Gamma \types \CtxF[R']:B \to C,\,\Delta$. Applying the typing rule \text{App} to $\Gamma \types \CtxF[R']:B \to C,\,\Delta$ and $\Gamma \types N:B,\,\Delta$ gives $\Gamma \types \app*{\CtxF[R']}{N}:C,\,\Delta$. Since $C \subtyping A$, the typing rule \textsc{Sub} applied to $\Gamma \types \app*{\CtxF[R']}{N}:C,\,\Delta$ gives $\Gamma \types \app*{\CtxF[R']}{N}:A,\,\Delta$.
                      \item[Case \circled{2}.] By the induction hypothesis, $\Gamma \types \CtxF[R]:\CompPure{\funty},\,\Delta$ implies $\Gamma \types \CtxF[R']:\CompPure{\funty},\,\Delta$. Applying the typing rule \textsc{AppE} to $\Gamma \types \CtxF[R']:\CompPure{\funty},\,\Delta$ then yields $\Gamma \types \app*{\CtxF[R']}{N}:\CompPure{\okty},\,\Delta$. Since $\CompPure{\okty}\subtyping A$, the typing rule \textsc{Sub} applied to $\Gamma \types \app*{\CtxF[R']}{N}:\CompPure{\okty},\,\Delta$ gives $\Gamma \types \app*{\CtxF[R']}{N}:A,\,\Delta$.
                  \end{mypar}

                  \indcase{$\Ctx$ has shape $\prj*{i}{\CtxF}\,(i \in \makeset{1,2})$}
                  Here, given $\Gamma \types \prj*{i}{\CtxF[R]}:A,\,\Delta$, we must show $\Gamma \types \prj*{i}{\CtxF[R']}:A,\,\Delta$. By \cref{lem:inversion} (\textsc{Inversion}) on $\Gamma \types \prj*{i}{\CtxF[R]}:A,\,\Delta$ with $\Gamma \nvdash \Delta$, either:
                  \begin{circledlist}
                      \item $\exists A_1,\,A_2$ such that $A_i\subtyping A$ and $\Gamma \types \CtxF[R]:(A_1, A_2),\,\Delta$
                      \item $\CompPure{\okty}\subtyping A$ and $\Gamma \types \CtxF[R']: \CompPure{\mathsf{Pair}},\,\Delta$.
                  \end{circledlist}
                  \settowidth{\titleindent}{\textit{Case} \circled{2}.}
                  \begin{mypar}
                      \item[Case \circled{1}.] Since $R$ is a redex with $R \ped R'$ and $\CtxF$ is a context, the induction hypothesis applied to $\Gamma \types \CtxF[R]:(A_1, A_2),\,\Delta$ gives $\Gamma \types \CtxF[R']:(A_1, A_2),\,\Delta$. Applying the typing rule $\textsc{Prj}_{i}$ to $\Gamma \types R':(A_1, A_2),\,\Delta$ then yields $\Gamma \types \prj*{i}{\CtxF[R']}:A_i,\,\Delta$. Since $A_i \subtyping A$, the typing rule \textsc{Sub} applied to $\Gamma \types \prj*{i}{\CtxF[R']}:A_i,\,\Delta$ gives $\Gamma \types \prj*{i}{\CtxF[R']}:A,\,\Delta$.
                      \item[Case \circled{2}.] By the induction hypothesis, $\Gamma \types \CtxF[R]:\CompPure{\mathsf{Pair}},\,\Delta$ and $R \ped R'$ implies $\Gamma \types \CtxF[R']:\CompPure{\mathsf{Pair}},\,\Delta$. Then, by applying the typing rule $\textsc{PrjE}_i$ to $\Gamma \types \CtxF[R']:\CompPure{\mathsf{Pair}},\,\Delta$, we get $\Gamma \types \prj*{i}{\CtxF[R']}:\CompPure{\okty},\,\Delta$. Since $\CompPure{\okty}\subtyping A$, the typing rule \textsc{Sub} applied to $\Gamma \types \prj*{i}{\CtxF[R']}:\CompPure{\okty},\,\Delta$ gives $\Gamma \types \prj*{i}{\CtxF[R']}:A,\,\Delta$.
                  \end{mypar}

                  \indcase{$\Ctx$ has shape $(\CtxF, N)$}
                  Here, given $\Gamma \types (\CtxF[R],N):A,\,\Delta$, we must show $\Gamma \types (\CtxF[R'],N):A,\,\Delta$. By \cref{lem:inversion} (\textsc{Inversion}) on $\Gamma \types (\CtxF[R],N):A,\,\Delta$ with $\Gamma \nvdash \Delta$:
                  \[
                      \exists B, C \text{ such that } (B,\,C) \subtyping A \text{ and } \Gamma \types \CtxF[R] : B,\,\Delta \text{ and } \Gamma \types N : C,\,\Delta
                  \]
                  Since $R$ is a redex with $R \ped R'$ and $\CtxF$ is a context, the induction hypothesis applied to $\Gamma \types \CtxF[R] : B,\,\Delta$ gives $\Gamma \types \CtxF[R'] : B,\,\Delta$. Applying the typing rule $\textsc{Pair}$ to $\Gamma \types \CtxF[R'] : B,\,\Delta$ and $\Gamma \types N : C,\,\Delta$ then yields $\Gamma \types (\CtxF[R'],N):(B,C),\,\Delta$. Since $(B,C) \subtyping A$, the typing rule \textsc{Sub} applied to $\Gamma \types (\CtxF[R'],N):(B,C),\,\Delta$ gives $\Gamma \types (\CtxF[R'],N):A,\,\Delta$.

                  %   \indcase{$\Ctx$ has shape $(V, \CtxF)$}
                  %   Here, given $\Gamma \types (V,\CtxF[R]):A,\,\Delta$, we must show $\Gamma \types (V,\CtxF[R']):A,\,\Delta$. By \cref{lem:inversion} (\textsc{Inversion}) on $\Gamma \types (V,\CtxF[R]):A,\,\Delta$ with $\Gamma \nvdash \Delta$:
                  %   \[
                  %       \exists B, C \text{ such that } (B,\,C) \subtyping A \text{ and } \Gamma \types V : B,\,\Delta \text{ and } \Gamma \types \CtxF[R] : C,\,\Delta
                  %   \]
                  %   \begin{circledlist}
                  %       \item $\exists B, C$ such that $C\subtyping A$, $\Gamma \types M:B \to C,\,\Delta$ and $\Gamma \types N:B,\,\Delta$
                  %       \item $\CompPure{\okty}\subtyping A$ and $\Gamma \types M : \CompPure{\funty},\,\Delta$.
                  %   \end{circledlist}
                  %   Since $R$ is a redex with $R \ped R'$ and $\CtxF$ is a context, the induction hypothesis applied to $\Gamma \types \CtxF[R] : C,\,\Delta$ gives $\Gamma \types \CtxF[R'] : C,\,\Delta$. Applying the typing rule $\textsc{Pair}$ to $\Gamma \types V : B,\,\Delta$ and $\Gamma \types \CtxF[R'] : C,\,\Delta$ then yields $\Gamma \types (V,\CtxF[R']):(B,C),\,\Delta$. Since $(B,C) \subtyping A$, the typing rule \textsc{Sub} applied to $\Gamma \types (V,\CtxF[R']):(B,C),\,\Delta$ gives $\Gamma \types (V,\CtxF[R']):A,\,\Delta$.

                  \indcase{$\Ctx$ has shape $\opPrim{\CtxF}{N}$}
                  Here, given $\Gamma \types \opPrim{\CtxF[R]}{N}:A,\,\Delta$, we must show $\Gamma \types \opPrim{\CtxF[R']}{N}:A,\,\Delta$. By \cref{lem:inversion} (\textsc{Inversion}) on $\Gamma \types \opPrim{\CtxF[R]}{N}:A,\,\Delta$ with $\Gamma \nvdash \Delta$, either:
                  \begin{circledlist}
                      \item $\intty \subtyping A$ and $\Gamma \types \CtxF[R] : \intty,\,\Delta$ and $\Gamma \types N: \intty,\,\Delta$
                      \item $\CompPure{\okty} \subtyping A$ and $\Gamma \types \CtxF[R]: \CompPure{\intty},\,N:\CompPure{\intty},\,\Delta$.
                  \end{circledlist}

                  \settowidth{\titleindent}{\textit{Case} \circled{2}.}
                  \begin{mypar}
                      \item[Case \circled{1}.] Since $R$ is a redex with $R \ped R'$ and $\CtxF$ is a context, the induction hypothesis applied to $\Gamma \types \CtxF[R]:\intty,\,\Delta$ gives $\Gamma \types \CtxF[R']:\intty,\,\Delta$. Applying the typing rule \text{Op} to $\Gamma \types \CtxF[R']:\intty,\,\Delta$ and $\Gamma \types N:\intty,\,\Delta$ gives $\Gamma \types \opPrim{\CtxF[R']}{N}:\intty,\,\Delta$. Since $C \subtyping A$, the typing rule \textsc{Sub} applied to $\Gamma \types \opPrim{\CtxF[R']}{N}:\intty,\,\Delta$ gives $\Gamma \types \opPrim{\CtxF[R']}{N}:A,\,\Delta$.
                      \item[Case \circled{2}.] By the induction hypothesis, $\Gamma \types \CtxF[R]: \CompPure{\intty},\,N:\CompPure{\intty},\,\Delta$ implies $\Gamma \types \CtxF[R']: \CompPure{\intty},\,N:\CompPure{\intty},\,\Delta$. Applying the typing rule \textsc{OpE} then yields $\Gamma \types \opPrim{\CtxF[R']}{N}:\CompPure{\okty},\,\Delta$. Since $\CompPure{\okty}\subtyping A$, the typing rule \textsc{Sub} applied to $\Gamma \types \opPrim{\CtxF[R']}{N}:\CompPure{\okty},\,\Delta$ gives $\Gamma \types \opPrim{\CtxF[R']}{N}:A,\,\Delta$.
                  \end{mypar}

                  \indcase{$\Ctx$ has shape $\opPrim{V}{\CtxF}$}
                  Here, given $\Gamma \types \opPrim{V}{\CtxF[R]}:A,\,\Delta$, we must show $\Gamma \types \opPrim{V}{\CtxF[R']}:A,\,\Delta$. By \cref{lem:inversion} (\textsc{Inversion}) on $\Gamma \types \opPrim{V}{\CtxF[R]}:A,\,\Delta$ with $\Gamma \nvdash \Delta$, either:
                  \begin{circledlist}
                      \item $\intty \subtyping A$ and $\Gamma \types V: \intty,\,\Delta$ and $\Gamma \types \CtxF[R] : \intty,\,\Delta$
                      \item $\CompPure{\okty} \subtyping A$ and $\Gamma \types V:\CompPure{\intty},\,\CtxF[R]: \CompPure{\intty},\,\Delta$.
                  \end{circledlist}

                  \settowidth{\titleindent}{\textit{Case} \circled{2}.}
                  \begin{mypar}
                      \item[Case \circled{1}.] Since $R$ is a redex with $R \ped R'$ and $\CtxF$ is a context, the induction hypothesis applied to $\Gamma \types \CtxF[R]:\intty,\,\Delta$ gives $\Gamma \types \CtxF[R']:\intty,\,\Delta$. Applying the typing rule \text{Op} to $\Gamma \types V:\intty,\,\Delta$ and $\Gamma \types \CtxF[R']:\intty,\,\Delta$ yields $\Gamma \types \opPrim{V}{\CtxF[R']}:\intty,\,\Delta$. Since $C \subtyping A$, the typing rule \textsc{Sub} applied to $\Gamma \types \opPrim{V}{\CtxF[R']}:\intty,\,\Delta$ gives $\Gamma \types \opPrim{V}{\CtxF[R']}:A,\,\Delta$.
                      \item[Case \circled{2}.] By the induction hypothesis, $\Gamma \types V:\CompPure{\intty},\,\CtxF[R]: \CompPure{\intty},\,\Delta$ implies $\Gamma \types V:\CompPure{\intty},\,\CtxF[R']: \CompPure{\intty},\,\Delta$. Applying the typing rule \textsc{OpE} then yields $\Gamma \types \opPrim{V}{\CtxF[R']}:\CompPure{\okty},\,\Delta$. Since $\CompPure{\okty}\subtyping A$, the typing rule \textsc{Sub} applied to $\Gamma \types \opPrim{V}{\CtxF[R']}:\CompPure{\okty},\,\Delta$ gives $\Gamma \types \opPrim{V}{\CtxF[R']}:A,\,\Delta$.
                  \end{mypar}
                  \indcase{$\Ctx$ has shape $\matchtm{\CtxF}{\mid_{i=1}^k p_i \mapsto N_i}$}
                  Here, given $\Gamma \types \matchtm{\CtxF}{\mid_{i=1}^k p_i \mapsto N_i}:A,\,\Delta$, we must show $\Gamma \types \matchtm{\CtxF[R']}{\mid_{i=1}^k p_i \mapsto N_i}:A,\,\Delta$. By \cref{lem:inversion} (\textsc{Inversion}) on $\Gamma \types \matchtm{\CtxF[R]}{\mid_{i=1}^k p_i \mapsto N_i}:A,\,\Delta$ with $\Gamma \nvdash \Delta$, either:
                  \begin{circledlist}
                      \item $\Gamma \types \CtxF[R]:\textstyle\bigvee_{i=1}^k p_i\theta_i,\,\Delta$ and $\forall i\in[1,k].\,\Gamma,\,\theta_i \types \CtxF[R]:\CompPure*{p_i\theta_i},\,N_i:A,\,\Delta$
                      \item $\CompPure{\okty}\subtyping A$ and $\Gamma \types \CtxF[R]:\CompPure*{\textstyle\bigvee_{i=1}^k p_i\thetaok{p_i}},\,\Delta$.
                  \end{circledlist}

                  \settowidth{\titleindent}{\textit{Case} \circled{2}.}
                  \begin{mypar}
                      \item[Case \circled{1}.] Since $R$ is a redex with $R \ped R'$ and $\CtxF$ is a context, the induction hypothesis applied to $\Gamma \types \CtxF[R]:\textstyle\bigvee_{i=1}^k p_i\theta_i,\,\Delta$ gives $\Gamma \types \CtxF[R']:\textstyle\bigvee_{i=1}^k p_i\theta_i,\,\Delta$. Applying the typing rule \text{Match} to $\Gamma \types \CtxF[R']:\textstyle\bigvee_{i=1}^k p_i\theta_i,\,\Delta$ and $\forall i\in[1,k].\,\Gamma,\,\theta_i \types \CtxF[R']:\CompPure*{p_i\theta_i},\,N_i:A,\,\Delta$, we get $\Gamma \types \matchtm{\CtxF[R']}{\mid_{i=1}^k p_i \mapsto N_i}:A,\,\Delta$.
                      \item[Case \circled{2}.] By the induction hypothesis, $\Gamma \types \CtxF[R]:\CompPure*{\textstyle\bigvee_{i=1}^k p_i\thetaok{p_i}},\,\Delta$ implies $\Gamma \types \CtxF[R']:\CompPure*{\textstyle\bigvee_{i=1}^k p_i\thetaok{p_i}},\,\Delta$. Applying the typing rule \textsc{MatchE} then yields $\Gamma \types \matchtm{\CtxF[R']}{\mid_{i=1}^k p_i \mapsto N_i}:\CompPure{\okty},\,\Delta$. Since $\CompPure{\okty}\subtyping A$, the typing rule \textsc{Sub} applied to $\Gamma \types \matchtm{\CtxF[R']}{\mid_{i=1}^k p_i \mapsto N_i}:\CompPure{\okty},\,\Delta$ gives $\Gamma \types \matchtm{\CtxF[R']}{\mid_{i=1}^k p_i \mapsto N_i}:A,\,\Delta$.
                  \end{mypar}
              \end{indproof}
    \end{itemize}
\end{proof}
\subsection{Progress}

\begin{theorem}[Progress]
    \label{thm:progress}
    Let $\Gamma \types Q:A,\, \Delta$, where $\Gamma \nvdash \Delta$ and $Q$ is a closed term.
    \begin{enumerate}
        \item If $A \subtyping \okty$, then either $Q$ can reduce one step, or $Q$ is a value.
        \item If $A \subtyping \CompPure{\okty}$, then either $Q$ can reduce one step, or $Q$ is stuck.
    \end{enumerate}
\end{theorem}

\begin{proof}[Proof of \cref{thm:progress}.(1)]
  If $Q$ can make a step, then the result is clear.  Otherwise $Q$ is a normal form and is, therefore, either a value $V$ or a stuck term $S$.  For (i), suppose $A \subtype \okty$.  Then also $\Gamma \types Q:\okty,\,\Delta$.  If $Q$ is a stuck term, then by Lemma~\ref{lem:type-of-stuck-term} (and weakening), $\Gamma \types Q : \CompPure{\okty},\,\Delta$, but then we obtain a contradiction by cut (Lemma~\ref{lem:cut-normal-form}), so $Q$ must be a value.  For (ii), suppose $A \subtype \CompPure{\okty}$.  Then also $\Gamma \types Q:\CompPure{\okty},\,\Delta$.  If $Q$ is a value, then by Lemma~\ref{lem:all-values-typable}, it follows that $\Gamma \types Q:\okty,\,\Delta$ and, again, we obtain a contradiction by cut (Lemma~\ref{lem:cut-normal-form}), so $Q$ must be stuck.
\end{proof}

% \begin{lemma}
%     Suppose $\Gamma \types p\sigma\tyfit\theta,\,\Delta$$\Gamma \types \sigma\tyfit\theta,\,\Delta$.
%     \begin{enumerate}
%         \item If $\Gamma \types p\sigma \tyfit p\theta,\,\Delta$
%         \item If $\forall x \in \dom(\theta). \theta(x) \tyneq \botty$, then $p\theta \tyeq \botty$.
%     \end{enumerate}

% \end{lemma}
% \begin{theorem}[Preservation]
%     If $\Gamma \types P:A,\, \Delta$ and $P \ped Q$, then $\Gamma \types Q:A,\, \Delta$.
% \end{theorem}

\subsection{Completeness}

Theorem~\ref{thm:complete-classification}:
\begin{quote}
    For all normal forms $U$, both of the following are true:
    \begin{quote}
        (i) \ If $U$ is a value then $\ \types U : \okty$. \qquad (ii) \ If $U$ is stuck then $\ \types U : \CompPure{\okty}$
    \end{quote}
\end{quote}

\begin{proof}
    This follows immediately from \begin{added}Lemmas~\ref{lem:all-values-typable}\end{added} and \ref{lem:type-of-stuck-term}.
\end{proof}

}{%
}

\end{document}
\endinput
%%
%% End of file `sample-manuscript.tex'.